\documentclass{siamltex}
\usepackage{amssymb}
\usepackage{amscd}
\usepackage{bm}
\usepackage{amsmath}
\usepackage{graphicx}

\newtheorem{alg}{Algorithm}[section]

\newcommand{\eps}{\varepsilon}

\newcommand{\pa}{\partial}

\newcommand{\al}{\alpha}

\newcommand{\na}{\nabla}

\newcommand{\Ga}{\Gamma}
\newcommand{\Om}{\Omega}
\newcommand{\de}{\delta}

\newcommand{\De}{\Delta}

\newcommand{\lam}{\lambda}

\newcommand{\bks}{\backslash}

\renewcommand{\i}{\mathbf{i}}
\newcommand{\R}{{\mathbb{R}}}

\newcommand{\C}{{\mathbb{C}}}
\renewcommand{\Re}{\mathrm{Re}\,}
\renewcommand{\Im}{\mathrm{Im}\,}

\newcommand{\bs}{\backslash}

\newcommand{\be}{\begin{eqnarray}}
\newcommand{\ee}{\end{eqnarray}}
\newcommand{\ben}{\begin{eqnarray*}}
\newcommand{\een}{\end{eqnarray*}}
\newcommand{\bee}{\begin{equation}}
\newcommand{\eee}{\end{equation}}
\newcommand{\nn}{\nonumber}

\newcommand{\debproof}{\begin{proof}}
\newcommand{\finproof}{\end{proof}}

\title{Reverse Time Migration for Reconstructing Extended Obstacles in Planar Acoustic Waveguides
\thanks{This work is  supported by National Basic Research Project under the grant 2011CB309700 and China NSF under the grants 11021101 and 11321061.}}

\author{Zhiming Chen\thanks{LSEC, Institute of Computational Mathematics and Scientific Engineering Computing,
Academy of Mathematics and System Sciences, Chinese Academy of Sciences, Beijing 100190, P.R. CHINA ({\tt zmchen@lsec.cc.ac.cn}).}
        \and Guanghui Huang\thanks{LSEC, Institute of Computational Mathematics and Scientific Engineering Computing,
Academy of Mathematics and System Sciences, Chinese Academy of Sciences, Beijing 100190, P.R. CHINA ({\tt ghhuang@lsec.cc.ac.cn}).}}

\begin{document}
\maketitle

\begin{abstract}
We propose a new reverse time migration method for reconstructing extended obstacles in the planar waveguide using acoustic waves
at a fixed frequency. We prove the resolution of the reconstruction method in terms of the aperture and the thickness of the waveguide.
The resolution analysis implies that the imaginary part of the cross-correlation imaging function is always positive and thus may have better stability properties. Numerical experiments are included to illustrate the powerful imaging quality and to confirm our resolution results.
\end{abstract}



\pagestyle{myheadings}
\thispagestyle{plain}
\markboth{ZHIMING CHEN AND GUANGHUI HUANG}{REVERSE TIME MIGRATION FOR PLANAR ACOUSTIC WAVEGUIDE}

\section{Introduction}

We propose a reverse time migration (RTM) method to find the support of an unknown obstacle embedded in a planar acoustic waveguide
from the measurement of the wave field on part of the boundary of the waveguide which is far away from the obstacle (see Figure \ref{wg_figure}). Let $\R^2_h=\{(x_1,x_2)\in\R^2: x_2\in (0,h)\}$ be the waveguide of thickness $h>0$. Denote by $\Ga_0=\{(x_1,x_2)\in \R^2: x_2 = 0\}$ and $\Ga_h=\{(x_1,x_2)\in \R^2: x_2 = h\}$
the boundaries of $\R^2_h$. Let the obstacle occupy a bounded Lipschitz domain $D$ included in $B_R=(-R,R)\times (0,h)$, $R>0$, with $\nu$ the unit outer normal to its boundary $\Ga_D$.
We assume the incident wave is a point source excited at $x_s\in\Ga_h$. The measured wave field satisfies the following equations:
\be
& & \Delta u + k^2 u =-\delta_{x_s}(x) \qquad \mbox{in } \R^2_h\bs\bar D, \label{waveguide}\\
& & \frac{\pa u}{\pa \nu} + \i k \eta(x) u = 0 \ \ \ \ \mbox{ on } \Ga_D, \label{dir_bc}\\
& & u= 0\ \ \mbox{on }\Ga_0,\ \ \ \ \frac{\partial u}{\partial x_2}= 0\ \ \mbox{on }\Ga_h. \label{gammah0}
\end{eqnarray}
Here $k>0$ is the wave number and $\eta(x)>0$ is a bounded function on $\Ga_D$. The equation \eqref{waveguide} is understood as the limit when $x_s\in\R^2_h\bs\bar D$ tends to $\Ga_h$. The impedance boundary condition in \eqref{dir_bc} is assumed only for the convenience of the analysis of this paper. The RTM method
studied in this paper does not require any a priori information of the physical
properties of the obstacle such as penetrable and non-penetrable, and for the non-penetrable obstacles, the type of boundary conditions on the boundary of
the obstacle (see section \ref{extensions} below).

\begin{figure}
    \begin{center}
    \vskip-1.cm
    \qquad \qquad \qquad \includegraphics[width=0.8\textwidth,height=3.0in]{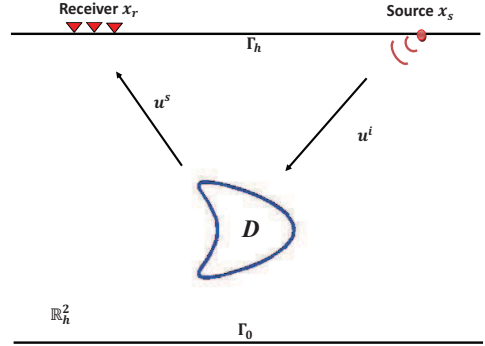}
    \end{center}
    \vskip-1.5cm
    \caption{The geometric setting of the inverse problem in planar waveguides.} \label{wg_figure}
\end{figure}

Now we introduce the radiation condition for the planar waveguide problem \cite{Xu1}. Since $D\subset B_R$, we have by separation of variables the following mode expansion:
\bee \label{expansion}
     u(x_1,x_2) = \sum_{n=1}^{\infty} u_n(x_1)\sin(\mu_n x_2 ),\ \ \ \ \forall\ |x_1|>R,
\eee
where $\mu_n=\frac{2n-1}{2h}\pi$, $n=1,2,\cdots$, are called {\it cut-off} frequencies.
In this paper we will always assume
\begin{equation}\label{cut-off}
k\not=\frac{2n-1}{2h}\pi,\ \ n=1,2,\cdots.
\end{equation}
The mode expansion coefficients $u_n(x_1)$, $n=1,2,\cdots$, satisfy the 1D Helmholtz equation:
\bee \label{mode1}
    u_n'' + \xi_n^2 u_n = 0,\ \ \ \ \forall\ |x_1|>R, \ \ n=1,2, ...,
\eee
where $\xi_n=\sqrt{k^2-\mu_n^2}$ if $k>\mu_n$ and $\xi_n=\i\sqrt{\mu_n^2-k^2}$ if $k<\mu_n$. The radiation condition
for the planar waveguide problem is then to impose the mode expansion coefficient $u_n(x_1)$ to satisfy
\bee\label{mode2}
\lim_{|x_1|\to\infty}\bigg(\frac{\partial u_n }{\partial |x_1|}-\i \xi_n u_n \bigg)=0,\ \ n=1,2,\cdots,
\eee
which guarantees the uniqueness of the solution of the 1D Helmholtz equation \eqref{mode1}.
The existence and uniqueness of the wave scattering problem \eqref{waveguide}-\eqref{gammah0} with the radiation condition \eqref{mode2} is an intensively
studied subject in the literature, see e.g. \cite{Arens1, MW1, MW2, RAMM,Xu1}. The difficulty is the possible existence of the so-called embedded
trapped modes which destroys the uniqueness of the solution \cite{lin}. In this paper we will show that the impedance boundary condition on the scatterer
guarantees the uniqueness of the scattering solution. We also prove the existence of the solution by the limiting absorption principle.

It is well known that imaging a scatterer in a waveguide is much more
challenging than in the free space. Indeed, because of the presence of two parallel infinite boundaries of the
waveguide, only a finite number of modes can propagate at long distance, while the other
modes decay exponentially \cite{Xu1}. We refer to \cite{ammari} for MUSIC type algorithm to locate small inclusions, \cite{Xu2} for the generalized dual space method, \cite{Arens2, BL, sz13} for the linear sampling method, \cite{tmp13} for a selective imaging method based on Kirchhoff migration, and the inversion method in \cite{RLAK} for reconstructing obstacles in waveguides.

The RTM method, which consists of back-propagating the complex conjugated data into the background medium and computing the
cross-correlation between the incident wave field and the backpropagation field to output the final imaging profile, is nowadays widely used in exploration geophysics \cite{ber84, cla85, bcs}. In \cite{cch_a, cch_e}, the RTM method for reconstructing extended targets using acoustic and electromagnetic waves at a fixed frequency in the free space is proposed and studied. The resolution
analysis in \cite{cch_a,cch_e} is achieved without using the small inclusion or geometrical optics assumption previously made in the literature.

The purpose of this paper is to extend the RTM method in \cite{cch_a, cch_e} to find extended targets in the planar acoustic waveguide.
Our new RTM algorithm is motivated by a generalized Helmholtz-Kirchhoff identity for the waveguide scattering problems. We show our new imaging function
enjoys the nice feature that it is always positive and thus may have better stability properties. The key ingredient in the analysis is a decay estimate
of the difference of the Green function for the waveguide problem and the half space Green function. We also refer to \cite{KE} for the study of the resolution of time-reversal experiments.

The rest of this paper is organized as follows. In section 2 we introduce some necessary results concerning the direct scattering problem.
In section 3 we prove the generalized Helmholtz-Kirchhoff identity and introduce our RTM algorithm. In section 4 we study the resolution of the finite aperture Helmholtz-Kirchhoff function which plays a key role in the resolution analysis of RTM algorithm in section 5. In section 6 we consider the extension of the resolution results for reconstructing penetrable
obstacles or non-penetrable obstacles with sound soft or sound hard boundary conditions. In section 7 we report extensive numerical experiments
to show the competitive performance of the RTM algorithm. In section 8 we include some concluding remarks. The appendix is devoted to the proof of the existence of the solution of the direct scattering
problem by the limiting absorption principle.

\section{Direct scattering problem}

We start by introducing the Green function $N(x,y)$, where $y\in\R^2_h$, which is the radiation solution satisfying the equations:
\begin{eqnarray*}
& & \Delta N(x,y) + k^2 N(x,y) = -\delta_y(x) \qquad \mbox{ in } \R^2_h,\\
& & N(x,y)=0\ \ \mbox{on }\Ga_0,\ \ \ \ \frac{\partial N(x,y)}{\partial x_2} = 0\ \ \mbox{on }\Ga_h.
\end{eqnarray*}
Let $\hat N_y(\xi,x_2)=\int^\infty_{-\infty}N(x,y)e^{-\i(x_1-y_1)\xi}dx_1$ be the Fourier transform in the first variable. It is easy to find by the assumption that $N(x,y)$ is a radiation solution that
\bee\label{spectal_green}
    \hat N_y(\xi,x_2) = \frac{\textbf{i}}{2\mu}\bigg( e^{\textbf{i}\mu|x_2-y_2|} - e^{\textbf{i}\mu(x_2+y_2)} - \frac{2\sin(\mu x_2)}{\cos(\mu h)}\sin(\mu y_2) e^{\textbf{i}\mu h} \bigg),
\eee
where $\mu=\sqrt{k^2 - \xi^2}$ and we choose the branch cut of $\sqrt{z}$ such that $\Re(\sqrt{z})\ge 0$ throughout the paper.
By using the limiting absorption principle one can obtain the following formula for the Green function by taking the inverse Fourier transform on the Sommerfeld Integral Path (SIP) (see Figure \ref{SIP}):
\bee{\label{GreenN0}}
    N(x,y) =\frac{1}{2\pi}\int_{\rm SIP}\hat N_y(\xi,x_2)e^{\textbf{i}\xi (x_1-y_1)}d\xi.
\eee
We refer to \cite[Chapter 2]{chew} for more discussion on the SIPs.
We will also use the following well-known {\it normal mode expression} for the Green function $N(x,y)$, see e.g. \cite{Xu1}:
\bee{\label{GreenN}}
    N(x,y) =\sum_{n=1}^{\infty}\frac{{\i}}{h\xi_n}\sin(\mu_n x_2)\sin(\mu_n y_2)e^{\i\xi_n|x_1-y_1|}.
\eee
It is obvious that the series in the normal mode expression is absolutely convergent if $x_1\not= y_1$. If $x_1=y_1$ but 
$x_2\not= y_2$, the series in \eqref{GreenN} is also convergent by using the method of Dirichlet's test \cite[\S 8.B.13-15]{v81}. 

\begin{figure}\label{SIP}
    \centering
    \includegraphics[width=0.7\textwidth, height=3.0in]{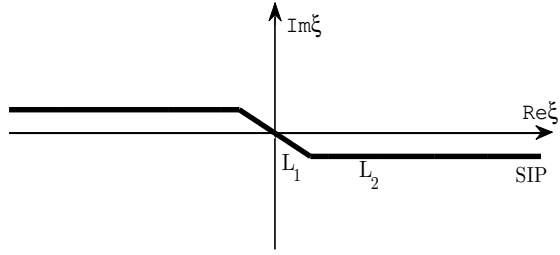}
    \vskip-2cm
    \caption{The Sommerfeld Integral Path (SIP).}
\end{figure}

\begin{lemma}\label{boundness}
    If $|x_1-y_1|\ge\alpha h$ for some constant $\al>0$, then $N(x,y)$ and $\na_x N(x,y)$ are uniformly bounded.
\end{lemma}
\debproof We only prove $N(x,y)$ is uniformly bounded. The proof for $\na_x N(x,y)$ is similar.
    Let $M$ be the integer such that $\mu_M<k<\mu_{M+1}$. Since $\frac{e^{-\al kh\sqrt{t^2-1}}}{\sqrt{t^2-1}}$ is a decreasing function in $(1,\infty)$, we know that
    \ben
        \sum_{n=M+1}^{\infty}\frac{1}{h|\xi_n|} e^{-|\xi_n||x_1-y_1|}\le\frac{1}{h|\xi_{M+1}|}+\frac 1\pi\int^\infty_1\frac{e^{-\al kh\sqrt{t^2-1}}}{\sqrt{t^2-1}}dt
        \le\frac 1{h|\xi_{M+1}|}+\frac 1{\al kh\pi}.
    \een
On the other hand, note that $|\sum_{n=1}^{M}\frac{\i}{h\xi_n}\sin(\mu_n x_2)\sin(\mu_n y_2) e^{\i\xi_n|x_1-y_1|}|<\sum_{n=1}^{M}\frac{1}{h\xi_n}$, we obtain
    \ben
        \sum_{n=1}^{M} \frac{1}{h\xi_n} \leq \frac{1}{h|\xi_M|}+\sum^{M-1}_{n=1} \frac{1}{h\xi_n}
        \leq \frac{1}{h|\xi_M|}+\frac{1}{\pi}\int_{0}^{1}\frac{1}{\sqrt{1-t^2}}dt=\frac{1}{h|\xi_M|}+\frac{1}{2},
    \een
where we have used the fact that $\frac{1}{\sqrt{1-t^2}}$ is an increasing function in $(0,1)$. This completes the proof. \qquad
\finproof

Now we consider the existence and uniqueness of the radiation solution of the following waveguide problem:
\be
& & \Delta \psi + k^2 \psi =0 \qquad \mbox{in } \R^2_h\bs\bar D, \label{l1}\\
& & \frac{\pa \psi}{\pa \nu} + \i k \eta(x) \psi = g \ \ \ \ \mbox{ on } \Ga_D, \label{l2}\\
& & \psi= 0\ \ \mbox{on }\Ga_0,\ \ \ \ \frac{\partial \psi}{\partial x_2}= 0\ \ \mbox{on }\Ga_h, \label{l3}
\ee
where $g\in H^{-1/2}(\Ga_D)$. We first show the uniqueness of the solution.

\begin{lemma}{\label{uniqueness}}
Let $\eta>0$ be bounded on $\Ga_D$. The scattering problem \eqref{l1}-\eqref{l3} has at most one radiation solution.
\end{lemma}

\debproof
We include a proof here for the sake of completeness. Let $g=0$ in \eqref{l2}. We multiply \eqref{l1} by $\bar{\psi}$ and integrate over $B_R\bs\bar D$ to obtain by integration by parts that
\bee \label{unique_id}
-\Im\int_{\Ga_D}\bar \psi\frac{\pa \psi}{\pa \nu}ds + \Im\int_{\pa B_R}\bar \psi\frac{\pa \psi}{\pa\nu}d s =0,
\eee
where $\nu$ is the unit outer normal to $\pa B_R$ on $\pa B_R$ and to $\Ga_D$ on $\Ga_D$. By the boundary condition satisfied by $\psi$,
$ \int_{(\Ga_0\cup \Ga_h)\cap\pa B_R}\bar \psi\frac{\pa \psi}{\pa \nu}ds=0$. On the other hand, for $|x_1|>R$, similar to \eqref{expansion} we have the mode expansion $\psi(x)=\sum^\infty_{n=1}\psi_n(x_1)\sin(\mu_n x_2)$ with $\psi_n(x_1)$ satisfying \eqref{mode1}-\eqref{mode2}.
Thus there exist constants $\psi_n^\pm$ such that $\psi_n(x_1) = \psi_n^{\pm} e^{\i\xi_n |x_1|}$ for $\pm\,x_1>R$.
By the Parseval identity, we have then
\ben
 \int_{\Ga_R^{+} \bigcup \Ga_R^{-}} \bar \psi\frac{\pa \psi}{\pa \nu}ds= \frac h2\sum_{n=1}^{M}\i\xi_n(|\psi_n^{+}|^2+|\psi_n^{-}|^2) - \frac h2\sum_{n=M+1}^{+\infty}|\xi_n|(|\psi_n^{+}|^2+|\psi_n^{-}|^2)e^{-2|\xi_n|R},
\een
where $\Ga_R^{\pm}=\{(x_1,x_2)\in\R^2: x_1=\pm R,x_2\in (0,h)\}.$
Thus by taking the imaginary part of the above identity and inserting it into (\ref{unique_id}) we have
\bee\label{DtN}
-\Im{\int_{\Ga_D}\bar \psi\frac{\pa \psi}{\pa\nu}ds} + \frac h2\sum_{n=1}^{M}\xi_n(|\psi_n^{+}|^2+|\psi_n^{-}|^2) = 0.
\eee
By using the impedance condition and the assumption $\eta>0$ on $\Ga_D$ we have $\psi=0$ on $\Ga_D$ and $\psi_n^{\pm}=0, n=1,2,...,M$. This implies that $\frac{\pa \psi}{\pa \nu}=0$ on $\Ga_D$. By the unique continuation principle we conclude $\psi=0$ in $\R^2_h \bks \bar D$. This completes the proof. \qquad
\finproof

In this paper, we call $\psi_n^\pm, n=1,2,\cdots, M$, which are the coefficients of the propagating modes, the {\it far-field pattern} of the radiation solution $\psi$ of the planar waveguide problem \eqref{l1}-\eqref{l3}.

We remark that under some assumption on the geometry of the obstacle, the uniqueness of the solution to the acoustic waveguide scattering problem for the sound soft obstacle was first proved in \cite{MW1} based on the Rellich type identity. The proof was refined in \cite{RAMM} and was also used in Arens \cite{Arens1} for 3D scattering problems. For general geometry of the obstacle, the embedded trapped mode may appear which makes the uniqueness fail \cite{lin}.

The following theorem which is useful in our resolution analysis for the RTM algorithm will be proved in the Appendix by using the method of limiting absorption principle.

\begin{theorem}{\label{LAP}}
Let $g \in H^{-1/2}(\Ga_D)$ and $\eta(x)>0$ be bounded on $\Ga_D$. Then the problem \eqref{l1}-\eqref{l3}
admits a unique radiation solution $\psi \in H^{1}_{\rm loc}(\R^2_h \backslash \bar D)$ . Moreover, for any bounded open set $\mathcal O\subset \R^2\bs\bar D$, there exists a constant $C$ such that
$\|\psi\|_{H^{1}(\mathcal O)}\le C\|g\|_{H^{-1/2}(\Ga_D)}$.
\end{theorem}

\section{The reverse time migration algorithm}

In this section we develop the reverse time migration type algorithm for inverse scattering problems in the planar acoustic waveguide.
Let $G(x,y)$ be the half-space Green function, where $y\in\R^2_+
=\{(x_1,x_2)\in\R^2: x_2>0 \}$, which satisfies the Sommerfeld radiation condition and the following equations:
\begin{eqnarray*}
\Delta G(x,y) + k^2 G(x,y) &=&-\delta_y(x) \qquad \mbox{ in } \R_{+}^2,\\
G(x,y) &=& 0  \qquad \qquad \ \ \ \mbox{on } \Ga_0.
\end{eqnarray*}
It is well known by the image method that
\begin{equation}\label{G}
    G(x,y) = \frac{\textbf{i}}{4}H_{0}^{(1)}(k|x-y|) - \frac{\textbf{i}}{4}H_{0}^{(1)}(k|x-y'|),
\end{equation}
where $H^{(1)}_0(z)$ is the first Hankel function of zeroth order and $y'=(y_1,-y_2)$ is the image point of $y=(y_1,y_2)$ with respect to $y_2 = 0$.

We start by proving the generalized Helmholtz-Kirchhoff identity which plays a key role in this paper.

\begin{lemma}\label{HK} Let $S(x,y)=N(x,y)-G(x,y)$. Then we have
 \begin{equation}\label{hk1}
    \int_{\Ga_h}\frac{\partial G(x,\zeta)}{\partial \zeta_2}\overline {N(\zeta,y)}ds(\zeta)=2\i\,\Im N(x,y)-S(x,y),\ \ \ \ \forall x,y\in\R^2_h.
\end{equation}
\end{lemma}

\debproof Let $x,y\in B_R=(-R,R)\times(0,h)$ for some $R>0$. Since $\Im G(x,\cdot)$ satisfies the Helmholtz equation, by the integral representation formula we obtain
\ben
\Im G(x,y)=\int_{\pa B_R}\left(\frac{\pa\, \Im G(x,\zeta)}{\pa\nu(\zeta)}N(\zeta,y)-\frac{\pa N(\zeta,y)}{\pa\nu(\zeta)}\Im G(x,\zeta)\right)ds(\zeta).
\een
Again by the integral representation formula we have
\ben
\int_{\pa B_R}\left(\frac{\pa G(x,\zeta)}{\pa\nu(\zeta)}N(\zeta,y)-\frac{\pa N(\zeta,y)}{\pa\nu(\zeta)}G(x,\zeta)\right)ds(\zeta)=-N(x,y)+G(x,y)=-S(x,y).
\een
Thus, since $\Im G(x,y)=\frac 1{2\i}(G(x,y)-\overline{G(x,y)})$, we have
\be\label{x1}
2\i\,\Im G(x,y)+S(x,y)&=&-\int_{\pa B_R}\left(\frac{\pa\overline{G(x,\zeta)}}{\pa\nu(\zeta)}N(\zeta,y)-\frac{\pa N(\zeta,y)}{\pa\nu(\zeta)}\overline{G(x,\zeta)}\right)ds(\zeta)\nn\\
&=&-\int_{\Ga_h\cap\pa B_R}\frac{\pa\overline{G(x,\zeta)}}{\pa\nu(\zeta)}N(\zeta,y)ds(\zeta)\nn\\
&-&\int_{\Ga^\pm_R}\left(\frac{\pa \overline{G(x,\zeta)}}{\pa\nu(\zeta)}N(\zeta,y)-\frac{\pa N(\zeta,y)}{\pa \nu(\zeta)}\overline{G(x,\zeta)}\right)ds(\zeta),
\ee
where we have used $\frac{\pa N(\zeta,y)}{\pa\nu(\zeta)}=0$ on $\Ga_h$ and $G(x,\zeta)=N(\zeta,y)=0$ on $\Ga_0$.
By \eqref{G} we know that $|G(x,\zeta)|=O(|x-\zeta|^{-1/2})$ and $|\frac{\pa G(x,\zeta)}{\pa \zeta_1}|=O(|x-\zeta|^{-1/2})$ as $|x-\zeta|\to\infty$. Therefore, by using Lemma \ref{boundness}
we conclude that the integral on $\Ga^\pm_R$ in \eqref{x1} vanishes as $R\to\infty$. This shows by letting $R\to\infty$ that
\ben
2\i\,\Im G(x,y)+S(x,y)=-\int_{\Ga_h}\frac{\pa\overline{G(x,\zeta)}}{\pa\nu(\zeta)}N(\zeta,y)ds(\zeta).
\een
This completes the proof by taking the complex conjugate and noticing $2\i\,\Im G(x,y)+S(x,y)=2\i\,\Im N(x,y)+\overline{S(x,y)}$ .\qquad
\finproof

Now assume that there are $N_s$ sources and $N_r$ receivers uniformly distributed on $\Ga_h^d$, where $\Ga_h^d=\{(x_1,x_2)\in\Ga_h: x_1\in (-d,d)\}$, $d>0$ is the aperture.
We denote by $\Om\subset B_d=(-d,d)\times (0,h)$ the sampling domain in which the obstacle is sought. 
Let $u^i(x,x_s)=N(x,x_s)$ be the incident wave and $u^s(x_r,x_s)=u(x_r,x_s)-u^i(x_r,x_s)$ be the scattered field measured at $x_r$, where $u(x,x_s)$ is the solution of the problem \eqref{waveguide}-\eqref{gammah0} and (\ref{mode2}). Our RTM algorithm consists of two steps. The first step is the back-propagation in which we back-propagate the complex conjugated data $\overline{u^s(x_r,x_s)}$ into the domain using the half space Green function $G(x,y)$. The second step is the cross-correlation in which we compute the imaginary part of the cross-correlation of $\frac{\pa G(x,y)}{\pa y_2}$  and the back-propagated field.

\bigskip
\begin{alg} {\sc (Reverse time migration)} \\
Given the data $u^s(x_r,x_s)$ which is the measurement of the scattered field at $x_r=(x_1(x_r),x_2(x_r))$ when the source is emitted at $x_s=(x_1(x_s),x_2(x_s))$, $s=1,\dots, N_s$, $r=1,\dots,N_r$. \\
$1^\circ$ Back-propagation: For $s=1,\dots,N_s$, compute the back-propagation field
\bee\label{back}
v_b(z,x_s)=\frac{|\Ga_h^d|}{N_r}\sum^{N_r}_{r=1}\frac{\pa G(z,x_r)}{\pa x_2(x_r)}\overline{u^s(x_r,x_s)},\ \ \ \ \forall z\in\Om.
\eee
$2^\circ$ Cross-correlation: For $z\in\Om$, compute
\bee\label{cor1}
I_d(z)=\Im\left\{\frac{|\Ga_h^d|}{N_s}\sum^{N_s}_{s=1} \frac{\pa G(z,x_s)}{\pa x_2(x_s)} v_b(z,x_s)\right\}.
\eee
\end{alg}

The back-propagation field $v_b$ can be viewed as the solution which satisfies the Sommerfeld radiation condition and the following equations:
\ben
& &\Delta v_b(x,x_s)+k^2v_b(x,x_s)=\frac{|\Ga_h^d|}{N_r}\sum^{N_r}_{r=1}\overline{u^s(x_r,x_s)}\frac{\pa}{\pa x_2}\de_{x_r}(x)\ \ \ \ \mbox{in }\R^2_+,\\
& &v_b(x,x_s)=0\ \ \ \ \mbox{on }\Ga_0.
\een
Taking the imaginary part of the cross-correlation of the incident field and the back-propagated field in \eqref{cor1} is motivated by the resolution analysis in the next section.
It is easy to see that
\bee\label{cor2}
I_d(z)=
\Im\left\{\frac{|\Ga_h^d||\Ga_h^d|}{N_sN_r}\sum^{N_s}_{s=1}\sum^{N_r}_{r=1}\frac{\pa G(z,x_r)}{\pa x_2(x_r)}\frac{\pa G(z,x_s)}{\pa x_2(x_s)}\overline{u^s(x_r,x_s)}\right\},\ \ \ \ \forall z\in\Om.
\eee
This formula is used in all our numerical experiments in section 7. By letting $N_s,N_r\to\infty$, we know that \eqref{cor2} can be viewed as an approximation of the following continuous integral:
\bee\label{cord}
\hat I_d(z)=\Im\int_{\Ga_h^d}\int_{\Ga_h^d}\frac{\pa G(z,x_r)}{\pa x_2(x_r)}\frac{\pa G(z,x_s)}{\pa x_2(x_s)}\overline{u^s(x_r,x_s)} ds(x_s)ds(x_r),\ \ \ \ \forall z\in\Om.
\eee

We will study the resolution of the function $\hat I_d(z)$ in the section 5. To this end we will first consider the resolution of the
finite aperture Helmholtz-Kirchhoff function in the next section.

To conclude this section we remark that our definition of the back-propagation field $v_b$ in \eqref{back} is motivated by the generalized Helmholtz-Kirchhoff identity in Lemma \ref{HK}. A straightforward extension of the RTM algorithm in \cite{cch_a, cch_e} would be to use $N(z,x_r)$ instead of $\frac{\pa G(z,x_r)}{\pa x_2(x_r)}$ in \eqref{back}
and $N(z,x_s)$ instead of $\frac{\pa G(z,x_s)}{\pa x_2(x_s)}$ in \eqref{cor1}.  
This would lead to the classical Kirchhoff migration imaging function \cite{bcs, tmp13}
\bee\label{cor3}
\tilde I_d(z)=
\frac{|\Ga_h^d||\Ga_h^d|}{N_sN_r}\sum^{N_s}_{s=1}\sum^{N_r}_{r=1}N(z,x_r)N(z,x_s)\overline{u^s(x_r,x_s)},\ \ \ \ \forall z\in\Om.
\eee
We will compare the performance of our imaging function $\hat I_d(z)$ and $\tilde I_d(z)$ in section 7. We note that $\tilde I_d(z)$ is divergent as $N_s,N_r\to\infty$ and $d\to\infty$.

\section{Resolution of the finite aperture Helmholtz-Kirchhoff function} 
By the Helmholtz-Kirchhoff identity \eqref{hk1} we know that for any $x,y\in\R_h$, 
\bee\label{y1}
\int_{\Ga_h^d}\frac{\pa G(x,\zeta)}{\pa\zeta_2}\overline{N(\zeta,y)}ds(\zeta)=2\i\,\Im N(x,y)-S(x,y)-S_d(x,y),
\eee
where 
\bee\label{y2}
S_d(x,y):=\int_{\Ga_h\bs\bar\Ga_h^d}\frac{\pa G(x,\zeta)}{\pa\zeta_2}\overline{N(\zeta,y)}ds(\zeta),\ \ \forall x,y\in\R_h.
\eee
The integral on the left-hand side of \eqref{y1}, $H_d(x,y)=\int_{\Ga_h^d}\frac{\pa G(x,\zeta)}{\pa\zeta_2}\overline{N(\zeta,y)}ds(\zeta)$,
will be called the finite aperture Helmholtz-Kirchhoff function in the following. In this section we will estimate $S(x,y)$ and $S_d(x,y)$ in \eqref{y1} which provides the resolution of $H_d(x,y)$. 

We assume the obstacle $D \subset\Om$ and there exist positive constants $c_0,c_1,c_2$, where $c_0,c_1\in (0,1)$, such that
\bee\label{cond-omega}
|y_1|\le c_0d,\ \ \ \ |y_2|\le c_1h,\ \ \ \ k|y_1-z_1|\le c_2\sqrt{kh},\ \ \ \ \forall y,z\in\Om.
\eee
The first condition means that the search domain should not be close to the boundary of the aperture. The second condition is rather mild in practical applications as we are interested in finding obstacles far away from the surface of the waveguide where the data is collected. The third condition indicates that the horizontal width of the search domain should not be very large comparing with the thickness of the waveguide. This is reasonable since we are interested in the case when the size of the scatterer is smaller than or comparable with the probe wavelength and the thickness $h$ is large compared with the probe wavelength, i.e., $kh\gg 1$.

We start with the following formula for $S(x,y)$.

\begin{lemma}\label{lem:S}
Let $S(x,y)=N(x,y)-G(x,y)$. Then we have
\ben
    S(x,y) =   \frac{1}{2\pi}\int_{\rm SIP}\hat S_y(\xi,x_2)e^{\i\xi|x_1-y_1|}d \xi,\ \ \hat S_y(\xi,x_2)=-\frac{2\i}{\mu }\frac{\sin(\mu x_2) }{e^{2\i \mu h } +1} \sin(\mu y_2 ) e^{2\i\mu h}.
\een
\end{lemma}

\debproof Let
\ben
\hat G_y(\xi,x_2)=\int^\infty_{-\infty}G(x,y)e^{-\i(x_1-y_1)\xi}dx_1,\ \ 
\hat S_y(\xi,x_2)=\int^\infty_{-\infty}S(x,y)e^{-\i(x_1-y_1)\xi}dx_1,
\een
be the Fourier transform of $G(x,y)$ and $S(x,y)$ in the first variable, respectively. It is easy to find that
\ben
    \hat G_y(\xi,x_2) = \frac{\textbf{i}}{2\mu}\bigg( e^{\textbf{i}\mu|x_2-y_2|} - e^{\textbf{i}\mu(x_2+y_2)} \bigg).
\een
Thus by \eqref{spectal_green} we know that
\ben
\hat S_y(\xi,x_2)=-\frac \i\mu \frac{\sin(\mu x_2) }{\cos(\mu h)} \sin(\mu y_2 ) e^{\i\mu h}.
\een
This completes the proof by taking the inverse Fourier transform along SIP. \qquad
\finproof

\begin{theorem}{\label{S}}
Let $kh>\pi/2$ and \eqref{cond-omega} be satisfied. We have
    \begin{equation}\label{s2}
    |S(x,y)|\leq \frac{C}{|\cos(kh)|}\frac{1}{\sqrt{k h}},\ \ |\na_x S(x,y)|\le\frac {Ck}{|\cos(kh)|}\frac{1}{\sqrt{kh}},\ \ \ \ \forall x,y\in\Om,
\end{equation}
where  $C$ is a constant independent of $k,h$ but may depend on $c_1, c_2$. \par
\end{theorem}

We remark that since $\mu_1=\pi/(2h)$, the condition $kh>\pi/2$ means that there exists at least one propagating mode in the received scattering field on $\Ga_h$, which is the minimum requirement that any imaging method could work. We also remark that the decay estimate \eqref{s2} can not hold uniformly for $x,y\in\R^2_h$ since $N(x,y)$ keeps oscillatory and bounded as $|x_1-y_1|\to\infty$ but $G(x,y)$ decays to $0$ as $|x_1-y_1|\to\infty$. 

\debproof
Denote by $\gamma=1/\sqrt{2kh}$ and by the assumption $kh\ge\pi/2$, $\gamma\le 1/\sqrt{\pi}$. Write $\xi=\xi_1+ \i \xi_2$, $\xi_1,\xi_2\in\R$. Let ${\rm SIP}^+$ be the part of the SIP in the fourth quadrant. By taking the coordinate transform $\xi \rightarrow-\xi$ in the second quadrant we know from Lemma \ref{lem:S} that
\ben
S(x,y)=\frac 1{2\pi}\int_{{\rm SIP}^+}\hat S_y(\xi,x_2)(e^{\i\xi(x_1-y_1)}+e^{-\i\xi(x_1-y_1)})d\xi:=S_1(x,y)+S_2(x,y),
\een
where $S_j(x,y)=\frac 1{2\pi}\int_{L_j}\hat S_y(\xi,x_2)(e^{\i\xi(x_1-y_1)}+e^{-\i\xi(x_1-y_1)})d\xi$, $j=1,2$, and $L_1,L_2$ are the sections of ${\rm SIP}^+$ (see Figure \ref{SIP}):
\ben
L_1=\{\xi\in\C:\xi_1\in (0,k\gamma),\xi_2=-\xi_1\},\ \ L_2=\{\xi\in\C:\xi_1\in (k\gamma,\infty),\xi_2=-k\gamma\}.
\een
Let $\mu=\sqrt{k^2-\xi^2}=\mu_1+\i\mu_2$, $\mu_1,\mu_2\in\R$. It is easy to see that
\bee\label{s3}
   |\mu|^2 = \sqrt{(k^2-\xi_1^2+\xi_2^2)^2 + 4\xi_1^2\xi_2^2}, \ \ \ \ \mu_2=\frac{\sqrt 2\,\xi_1|\xi_2|}{\sqrt{|\mu|^2+(k^2-\xi_1^2+\xi_2^2)}}.
\eee
It is clear by using \eqref{cond-omega} that
\ben
|e^{\i\xi(x_1-y_1)}+e^{-\i\xi(x_1-y_1)}|\le 2e^{-\xi_2|x_1-y_1|}\le 2e^{c_2/\sqrt 2},\ \ \ \ \forall\xi\in L_1\cup L_2.
\een
Thus, since $|\sin(\mu x_2)\sin(\mu y_2)e^{\i\mu h}|\le e^{-\mu_2(2h-x_2-y_2)}$, we have
\bee\label{s4}
|S_j(x,y)|\le C\left|\int_{L_j}\frac{1}{|\mu|}\frac 1{|1+e^{2\i\mu h}|}e^{-\mu_2(2h-x_2-y_2)}d\xi\right|.
\eee

We first estimate $S_1(x,y)$ and thus assume $\xi\in L_1$. By \eqref{s3} it is clear that $|\mu|\ge k$. Next
\be\label{s5}
\frac 1{|1+e^{2\i\mu h}|}&\le&\left|\frac 1{1+e^{2\i\mu h}}-\frac 1{1+e^{2\i kh}}\right|+\frac 1{|1+e^{2\i kh}|}\nn\\
&=&\frac 1{2|\cos(kh)|}\left(1+\frac{|e^{2\i\mu h}-e^{2\i kh}|}{|1+e^{2\i\mu h}|}\right).
\ee
By using the elementary inequality $1-e^{-t}\ge t-t^2/2$ for $t\ge 0$, we have $|1+e^{2\i\mu h}|\ge 1-e^{-2\mu_2h}\ge 2\mu_2h(1-\mu_2h)$. By \eqref{s3} we have $\mu_2=\frac{\sqrt 2\xi_1^2}{\sqrt{|\mu|^2+k^2}}$ which implies $\mu_2 h\le \frac{\xi_1^2 h}{k}\le \frac{1}{2}$ for $\xi\in L_1$.
Therefore
$|1+e^{2\i\mu h}|\ge \frac{\sqrt 2\xi_1^2 h}{\sqrt{|\mu|^2+k^2}}$.
On the other hand, since $\mu=\sqrt{k^2+2\i\xi_1^2}$ on $L_1$, $\mu(0)=k$, by using elementary calculus one obtains
\ben
|e^{2\i\mu h}-e^{2\i kh}|&\le&\sqrt 2\max_{t\in (0,\xi_1)}\left|\frac{de^{2\i\mu h}}{d\xi_1}\Big|_{\xi_1=t}\right| \times \xi_1\\
&=&\sqrt 2\,\max_{t\in (0,\xi_1)}\left(\frac{4hte^{-2\mu_2(t)h}}{|\mu(t)|}\right)\,\xi_1
\le 4\sqrt 2\,\frac{\xi^2_1 h}{k}.
\een
Thus by \eqref{s5}
\ben
\frac 1{|1+e^{2\i\mu h}|}\le\frac C{|\cos(kh)|}\frac{\sqrt{|\mu|^2+k^2}}{k}\le \frac C{|\cos(kh)|},
\een
where we have used the fact that $|\mu|=(k^4+4\xi_1^4)^{1/4}\le k(1+4\gamma^4)^{1/4}\le Ck$ for $\xi\in L_1$. Now it follows from \eqref{s4} that
\bee\label{s6}
|S_1(x,y)|\le \frac C{|\cos(kh)|}\gamma\le \frac C{|\cos(kh)|}\frac{1}{\sqrt{kh}}.
\eee

Now we estimate $S_2(x,y)$ and thus let $\xi\in L_2$. By \eqref{s3} we have $|\mu|^2\le k^2+\xi_1^2+\xi_2^2$ and thus
\bee\label{s7}
\mu_2\ge\frac{\xi_1|\xi_2|}{\sqrt{k^2+\xi_2^2}}=\frac{\xi_1\gamma}{\sqrt{1+\gamma^2}}\ge\frac{\xi_1\gamma}{\sqrt 2},
\eee
which implies $\mu_2h\ge kh\gamma^2/\sqrt 2= 1/(2\sqrt 2)$ and consequently $|1+e^{2\i\mu h}|\ge 1-e^{-1/\sqrt 2}$ for $\xi\in L_2$. Now by \eqref{s4} we have
\bee\label{s8}
|S_2(x,y)|\le C\int_{k\gamma}^{+\infty}\frac 1{|\mu|}e^{-\mu_2(2h-x_2-y_2)}d\xi_1.
\eee
For $\xi_1\in (k\gamma,k/\sqrt\pi)$, we know from \eqref{s3} that $|\mu|\ge\sqrt{k^2-\xi_1^2}\ge k\sqrt{1-1/\pi}$. This implies by using \eqref{s7} that
\ben
\int_{k\gamma}^{\frac k{\sqrt\pi}}\frac 1{|\mu|}e^{-\mu_2(2h-x_2-y_2)}d\xi_1&\le&\frac Ck\int_{k\gamma}^{\frac k{\sqrt\pi}}e^{-(1-c_1)\frac{\xi_1h}{\sqrt{kh}}}d\xi_1\\
&=&-\frac{C}{\sqrt{kh}}\Big[e^{-(1-c_1)\frac{\xi_1h}{\sqrt{kh}}}\Big]\Big|^{\frac k{\sqrt\pi}}_{k\gamma}\le C\frac 1{\sqrt{kh}}.
\een
For $\xi_1>\frac k{\sqrt\pi}$, by \eqref{s3}, we have $|\mu|\ge \sqrt{2\xi_1|\xi_2|}\ge \frac{\sqrt{2}}{\sqrt[4]\pi}\,k\gamma^{1/2}$. Thus
\ben
\int_{\frac k{\sqrt\pi}}^{+\infty}\frac 1{|\mu|}e^{-\mu_2(2h-x_2-y_2)}d\xi_1&\le&\frac C{k\gamma^{1/2}}\int_{\frac k{\sqrt\pi}}^{+\infty}e^{-(1-c_1)\frac{\xi_1h}{\sqrt{kh}}}d\xi_1\\
&=&-\frac{C}{\sqrt[4]{kh}}\Big[e^{-(1-c_1)\frac{\xi_1h}{\sqrt{kh}}}\Big]\Big|_{\frac k{\sqrt\pi}}^{+\infty}\le C\frac 1{\sqrt{kh}}.
\een
This shows the first estimate in \eqref{s2} upon substituting into \eqref{s8} and noticing \eqref{s6}. The estimate for $\na_xS(x,y)$ can be proved in a similar way by noticing that
\ben
& &\frac{\pa S(x,y)}{\pa x_1}=\frac 1{2\pi}\int_{{\rm SIP}^+}\frac{\xi}{\mu}\frac{\sin(\mu x_2)}{\cos(\mu h)}\sin(\mu y_2)e^{\i\mu h}(e^{\i\xi(x_1-y_1)}-e^{-\i\xi(x_1-y_1)})d\xi,\\
& &\frac{\pa S(x,y)}{\pa x_2}=\frac 1{2\pi}\int_{{\rm SIP}^+}(-\i)\frac{\cos(\mu x_2)}{\cos(\mu h)}\sin(\mu y_2)e^{\i\mu h}(e^{\i\xi(x_1-y_1)}+e^{-\i\xi(x_1-y_1)})d\xi.
\een
Here we omit the details. This completes the proof.
\finproof

Now we consider the effect of the finite aperture by estimating $S_d(x,y)$ in \eqref{y2}.
We first recall the following estimate for the first Hankel function in \cite[P.197]{watson}.

\begin{lemma}\label{lem:hankel} For any $t>0$, we have 
\ben
H^{(1)}_0(t)=\left(\frac 2{\pi t}\right)^{1/2}e^{\i(t-\pi/4)}+R_0(t),\ \ 
H^{(1)}_1(t)=\left(\frac 2{\pi t}\right)^{1/2}e^{\i(t-3\pi/4)}+R_1(t),
\een
where $|R_j(t)|\le Ct^{-3/2}$, $j=0,1$, for some constant $C>0$ independent of $t$.
\end{lemma}

For any $x,y\in\Om$, by the normal mode expression of the Green function $N(x,y)$ in \eqref{GreenN} we know that
\bee\label{y3}
S_d(x,y)=\sum^\infty_{n=1}\frac {-\i}{h\bar\xi_n}\sin(\mu_nh)\sin(\mu_ny_2)\int_{\Ga_h\bs\bar\Ga_h^d}\frac{\pa G(x,\zeta)}{\pa\zeta_2}e^{-\i\bar\xi_n|\zeta_1-y_1|}d\zeta_1.
\eee
For $\zeta=(\zeta_1,h)\in\Ga_h$, $G(x,\zeta)=\frac{\i}{4}H^{(1)}_0(k|x-\zeta|)-\frac{\i}{4}H^{(1)}_0(k|x-\zeta'|)$, where $\zeta'=(\zeta_1,-h)$. Thus, since $H^{(1)'}_0(\xi)=-H^{(1)}_1(\xi)$ for any $\xi\in\C$, 
\bee\label{y6}
\frac{\pa G(x,\zeta)}{\pa \zeta_2} =  f(x_1,\zeta_1,h+x_2)-f(x_1,\zeta_1,h-x_2),\ \ \forall\zeta\in\Ga_h,
\eee
where
\ben
f(x_1,\zeta_1,t)=\frac{\i}{4}H^{(1)}_1(k\Theta)\frac{kt}{\Theta},\ \ \Theta=\sqrt{(\zeta_1-x_1)^2+t^2},\ \ t \in [h-x_2,h+x_2].
\een
By Lemma \ref{lem:hankel} we have
\ben
f(x_1,\zeta_1,t)=\frac{\i}{4}e^{-\i\frac{3\pi}4}\left(\frac 2{\pi k\Theta}\right)^{1/2}\frac{kt}{\Theta}e^{\i k\Theta}
+\frac{\i}{4}R_1(k\Theta)\frac{kt}{\Theta}.
\een
Inserting the above equation into \eqref{y3} we obtain
\be\label{y4}
|S_d(x,y)|&\le&\max_{t\in [h-x_2,h+x_2]}\sum^\infty_{n=1}\frac {k^{1/2}t}{h|\xi_n|}\left|\int_{\Ga_h\bs\bar\Ga_h^d}
\Theta^{-3/2}e^{\i f_n(x_1,\zeta_1,t)}d\zeta_1\right|\nn\\
&+&\max_{t\in [h-x_2,h+x_2]}\sum^\infty_{n=1}\frac 1{h|\xi_n|}\left|\int_{\Ga_h\bs\bar\Ga_h^d}R_1(k\Theta)\frac{kt}{\Theta}e^{-\i\bar\xi_n|\zeta_1-y_1|}d\zeta_1\right|,
\ee
where $f_n(x_1,\zeta_1,t)=k\Theta-\bar\xi_n|\zeta_1-y_1|$. We note that for $1\le n\le M$, in which case $\xi_n=\sqrt{k^2-\mu_n^2}$ is real, $f_n(x_1,\zeta_1,t)$ has a critical point at $\zeta_1=x_1+p_n$:
\ben
\frac{\pa f_n}{\pa\zeta_1}(x_1,x_1+p_n,t)=0,\ \ \mbox{where }p_n=t\frac{\xi_n}{\mu_n},\ \ 1\le n\le M.
\een

\begin{lemma}\label{fn}
Let $1\le n\le M$ and \eqref{cond-omega} be satisfied. Then there exists a constant $C>0$ independent of $k,h,d$ such that
for any $x,y\in\Om$ and $t\in [h-x_2,h+x_2]$, we have,
\bee
\left|\int^\infty_\eta\Theta^{-3/2}e^{\i f_n(x_1,\zeta_1,t)}d\zeta_1\right|\le C\left(\frac{kt^2}{p_n^2}\right)^{-1}(\eta-x_1)^{-3/2},\ \ 
\forall\eta\ge x_1+2p_n,\label{f1}
\eee
and if $x_1+p_n/2\ge d$,
\bee
\left|\int^{x_1+p_n/2}_d\Theta^{-3/2}e^{\i f_n(x_1,\zeta_1,t)}d\zeta_1\right|\le C\left(\frac{kt^2}{p_n^2}\right)^{-1}(d-x_1)^{-3/2}.\label{f2}
\eee
\end{lemma}

\debproof It is clear that for any $\zeta_1\ge\eta\ge x_1+2p_n$,
\bee\label{f3}
\frac{\pa f_n}{\pa\zeta_1}(x_1,\zeta_1,t)=k\frac{\zeta_1-x_1}{\Theta}-\xi_n\ge k\left(\frac{2p_n}{\sqrt{4p_n^2+t^2}}-\frac{p_n}{\sqrt{p_n^2+t^2}}\right)\ge C\frac{kt^2}{p_n^2}.
\eee
Thus $\frac{\Theta(x_1,\zeta_1,t)^{-3/2}}{\i\pa_{\zeta_1}f_n(x_1,\zeta_1,t)}\to 0$ 
as $\zeta_1\to\infty$. By integration by parts we have then
\be\label{yy0}
\int^\infty_\eta\Theta^{-3/2}e^{\i f_n(x_1,\zeta_1,t)}d\zeta_1&=&-\frac{\Theta(x_1,\eta,t)^{-3/2}}{\i \pa_{\zeta_1} f_n(x_1,\eta,t)}e^{\i f_n(x_1,\eta,t)}\nn\\
&-&\int^\infty_\eta\frac{\pa}{\pa\zeta_1}\left(\frac{\Theta(x_1,\zeta_1,t)^{-3/2}}{\i \pa_{\zeta_1}f_n(x_1,\zeta_1,t)}\right)e^{\i f_n(x_1,\zeta_1,t)}d\zeta_1.
\ee
Since $\frac{\pa^2f_n(x_1,\zeta_1,t)}{\pa{\zeta_1^2}}=\frac{kt^2}{\Theta^3}$,
$\frac{\pa}{\pa\zeta_1}\left(\frac{\Theta(x_1,\zeta_1,t)^{-3/2}}{\pa_{\zeta_1}f_n(x_1,\zeta_1,t)}\right)\le 0$.
Thus
\ben
\left|\int^\infty_\eta\frac{\pa}{\pa\zeta_1}\left(\frac{\Theta(x_1,\zeta_1,t)^{-3/2}}{\i \pa_{\zeta_1}f_n(x_1,\zeta_1,t)}\right)e^{\i f_n(x_1,\zeta_1,t)}d\zeta_1\right|&\le&\int^\infty_\eta\left|\frac{\pa}{\pa\zeta_1}\left(\frac{\Theta(x_1,\zeta_1,t)^{-3/2}}{\pa_{\zeta_1}f_n(x_1,\zeta_1,t)}\right)\right|d\zeta_1\\
&=&\left|\int^\infty_\eta\frac{\pa}{\pa\zeta_1}\left(\frac{\Theta(x_1,\zeta_1,t)^{-3/2}}{\pa_{\zeta_1}f_n(x_1,\zeta_1,t)}\right)d\zeta_1\right|\\
&\le&[\pa_{\zeta_1}f_n(x_1,\eta,t)]^{-1}(\eta-x_1)^{-3/2}.
\een
This shows \eqref{f1} by using \eqref{yy0} and \eqref{f3}. The estimate \eqref{f2} can be proved similarly since for $d\le\zeta_1\le x_1+p_n/2$,
\ben
\left|\frac{\pa f_n}{\pa\zeta_1}(x_1,\zeta_1,t)\right|=\xi_n-k\frac{\zeta_1-x_1}{\Theta}\ge 
k\left(\frac{p_n}{\sqrt{p_n^2+t^2}}-\frac{p_n}{\sqrt{p_n^2+4t^2}}\right)\ge C\frac{kt^2}{p_n^2}.
\een
This completes the proof.
\finproof

We will use the following Van der Corput lemma, see e.g. in \cite[Corollary 2.6.8]{grafakos}, to estimate the 
oscillatory integral around the critical point.

\begin{lemma}\label{van}
There is a constant $C>0$ such that for any $-\infty<a<b<\infty$, for every real-valued $C^2$ function $u$ that 
satisfies $u''(t)\ge 1$ for $t\in (a,b)$, for any function $\psi$ defined on $(a,b)$ with an integrable derivative, and for any $\lambda>0$,
\ben
\left|\int^b_a e^{\i\lambda u(t)}\psi(t)dt\right|\le C\lambda^{-1/2}\left[|\psi(b)|+\int^b_a|\psi'(t)|dt\right],
\een
where the constant $C$ is independent of the constants $a,b,\lambda$ and the functions $u,\psi$.
\end{lemma}

We remark that if the function $\psi$ in Lemma \ref{van} is monotonic decreasing and non-negative in $(a,b)$, then we have
\bee\label{f5}
|\psi(b)|+\int^b_a|\psi'(t)|dt=|\psi(b)|+\left|\int^b_a\psi'(t)dt\right|=|\psi(a)|.
\eee

\begin{lemma}\label{vn}
Let $1\le n\le M$ and \eqref{cond-omega} be satisfied. Then there exists a constant $C>0$ independent of $k,h,d$ such that
for any $x,y\in\Om$ and $t\in [h-x_2,h+x_2]$, we have
\bee\label{f4}
\left|\int^{x_1+2p_n}_\eta\Theta^{-3/2}e^{\i f_n(x_1,\zeta_1,t)}d\zeta_1\right|\le Ckt^{-1}\xi_n^{-3/2},\ \ \forall\eta\ge x_1+p_n/2.
\eee
\end{lemma}

\debproof It is easy to see that for any $x_1+p_n/2\le \zeta_1\le x_1+2p_n$, 
\ben
\frac{\pa^2_{\zeta_1}f_n(x_1,\zeta_1,t)}{\pa^2_{\zeta_1}f_n(x_1,x_1+p_n,t)}=\frac{\Theta(x_1,x_1+p_n,t)^3}{\Theta(x_1,\zeta_1,t)^3}
\ge \left(\frac{p_n^2+t^2}{4p_n^2+t^2}\right)^{3/2}\ge \frac 18.
\een
We can use Lemma \ref{van} and \eqref{f5} for $\lambda=\pa^2_{\zeta_1}f_n(x_1,x_1+p_n,t)/8=\mu_n^3/(8k^2t)$ and $u(\zeta_1)=\frac{8\pa^2_{\zeta_1}f_n(x_1,\zeta_1,t)}{\pa^2_{\zeta_1}f_n(x_1,x_1+p_n,t)}$ to obtain, for any $\eta\ge x_1+p_n/2$,
\ben
\left|\int^{x_1+2p_n}_\eta\Theta^{-3/2}e^{\i f_n(x_1,\zeta_1,t)}d\zeta_1\right|\le C\left(\frac{\mu_n^3}{8k^2t}\right)^{-1/2}p_n^{-3/2},
\een
This completes the proof by using $p_n=t\xi_n/\mu_n$.
\finproof

\begin{lemma}\label{fn1}
Let $n\ge M+1$ and \eqref{cond-omega} be satisfied. Let the aperture $d\ge c_3 h$ for some constant $c_3>0$ independent of $k,h,d$. Then there exists a constant $C>0$ independent of $k,h,d$ such that for any $x,y\in\Om$ and $t\in [h-x_2,h+x_2]$, we have
\ben
\left|\int^\infty_d\Theta^{-3/2}e^{\i f_n(x_1,\zeta_1,t)}d\zeta_1\right|\le 
Ck^{-1}d^{-3/2}\chi_n,
\een
where $\chi_n=e^{-(1-c_0)|\xi_n|d}$ for $n\ge M+1$ and $\chi_n=1$ for $n\le M$.
\end{lemma}

\debproof Since $\xi_n=\i\sqrt{\mu_n^2-k^2}$ for $n\ge M+1$, by \eqref{cond-omega} we know that $|e^{-\i\bar\xi_n|\zeta_1-y_1|}|\le\chi_n$ for any $\zeta_1\ge d$. The proof of this lemma is essentially the same as that of Lemma \ref{fn} by using \eqref{yy0} and noticing that now we have $|\pa_{\zeta_1}f_n(x_1,\zeta_1,t)|\ge k\frac{|\zeta-x_1|}{\Theta}\ge Ck$ since $|\zeta_1-x_1|\ge (1-c_0)d\ge (1-c_0)c_3h$ and $|t|\le (1+c_1)h$. We omit the details.
\finproof

\begin{theorem}\label{aperture}
Let $kh>\frac {\pi}2$ and \eqref{cond-omega} be satisfied. Let $d\ge c_3 h$ for some constant $c_3>0$ independent of $k,h,d$.
Then there exists a constant $C>0$ independent of $k,h,d$ such that for any $x,y\in\Om$,
\ben
|S_d(x,y)|\le C\left(\frac 1{\sqrt{kd}}+\frac{h}d\right),\ \ 
|\na_x S_d(x,y)|\le Ck\left(\frac 1{\sqrt{kd}}+\frac hd\right).
\een
\end{theorem}

\debproof The starting point is \eqref{y4}. We first estimate the second term. Since $|R_1(k\Theta)|\le C(k\Theta)^{-3/2}$ by Lemma \ref{lem:hankel}, we have
\ben
\left|\int_{\Ga_h\bs\bar\Ga_h^d}R_1(k\Theta)\frac{kt}{\Theta}e^{\i\xi_n|\zeta_1-y_1|}d\zeta_1\right|
\le C\chi_nk^{-1/2}h\max_{|x_1|\le c_0d}\int^\infty_d\Theta^{-5/2}d\zeta_1\le\frac{C\chi_n}{\sqrt{kd}},
\een
where $\chi_n$ is defined in Lemma \ref{fn1} and we have used $d-x_1\ge (1-c_0)d\ge (1-c_0)c_3h$. This implies
\bee\label{y5}
\sum^\infty_{n=1}\frac 1{h|\xi_n|}\left|\int_{\Ga_h\bs\bar\Ga_h^d}R_1(k\Theta)\frac{kt}{\Theta}e^{\i\xi_n|\zeta_1-y_1|}d\zeta_1\right|\le \frac{C}{\sqrt{kd}},
\eee
where have used the fact that $\sum^\infty_{n=1}\frac{\chi_n}{h|\xi_n|}\le C$ by the argument in the proof of Lemma \ref{boundness}.
For estimating the first term in \eqref{y4}, we first use Lemma \ref{fn1} to obtain that 
\bee\label{yy7}
\sum^\infty_{n=M+1}\frac{k^{1/2}t}{h|\xi_n|}\left|\int_{\Ga_n\bs\bar\Ga_h^d}\Theta^{-3/2}e^{\i f_n(x_1,\zeta_1,t)}d\zeta_1\right|
\le\sum^\infty_{n=M+1}\frac{k^{1/2}t}{h|\xi_n|}\frac{C\chi_n}{kd^{3/2}}\le\frac{C}{\sqrt{kd}}.
\eee
It remains to estimate
\be\label{yy8}
& &\sum^M_{n=1}\frac{k^{1/2}t}{h|\xi_n|}\left|\int_{\Ga_n\bs\bar\Ga_h^d}\Theta^{-3/2}e^{\i f_n(x_1,\zeta_1,t)}d\zeta_1\right|\nn\\
&\le&\sum^M_{n=1}\frac{2}{h|\xi_n|}\max_{|x_1|\le c_0d}\left|\int_d^\infty k^{1/2}t\Theta^{-3/2}e^{\i f_n(x_1,\zeta_1,t)}d\zeta_1\right|.
\ee
Let $n_0\ge 1$ be such that $x_1+2p_{n_0}>d$ and $x_1+2p_{n_0+1}\le d$, which is equivalent to
\bee\label{f6}
k^{-1}\mu_{n_0}<\frac{2t}{\sqrt{d_1^2+4t^2}},\ \ k^{-1}\mu_{n_0+1}\ge\frac{2t}{\sqrt{d_1^2+4t^2}},\ \ d_1:=d-x_1.
\eee
Clearly $n_0\le M$ since $k^{-1}\mu_{n_0}< 1$. For $n\ge n_0+1$, we have $x_1+2p_n\le d$ and thus by \eqref{f1} with 
$\eta=d$ we obtain
\be
\sum^M_{n=n_0+1}\frac{k^{1/2}t}{h\xi_n}\left|\int_d^\infty \Theta^{-3/2}e^{\i f_n(x_1,\zeta_1,t)}d\zeta_1\right|
&\le&C\sum^M_{n=n_0+1}
\frac{k^{1/2}t}{h\xi_n}\left(\frac{kt^2}{p_n^2}\right)^{-1}d^{-3/2}\nn\\
&\le&C\frac{k^{1/2}t}{d^{3/2}}\sum^M_{n=n_0+1}\frac{\mu_n^{-2}}{h}\le \frac C{\sqrt{kd}},\label{yy3}
\ee
where we have used $\xi_n\le k$ for $n\le M$ and the fact that by \eqref{f6},
$\sum^M_{n=n_0+1}h^{-1}\mu_n^{-2}\le h^{-1}\mu_{n_0+1}^{-2}+\pi^{-1}\mu_{n_0+1}^{-1}\le Cd/(kh)$.

For $n\le n_0$, by \eqref{f6}, we have $\xi_n\ge \frac{kd_1}{\sqrt{d_1^2+4t^2}}\ge Ck$. By \eqref{f1} with $\eta=x_1+2p_n$ we obtain
\be\label{yy4}
\sum^{n_0}_{n=1}\frac{k^{1/2}t}{h\xi_n}\left|\int^\infty_{x_1+2p_n}\Theta^{-3/2}e^{\i f_n(x_1,\zeta_1,t)}d\zeta_1\right|
&\le&C\sum^{n_0}_{n=1}\frac{k^{1/2}t}{h\xi_n}\frac{p^{1/2}_n}{kt^2}\nn\\
&\le&C\sum^{n_0}_{n=1}\frac{1}{kt^{1/2}}\frac{\mu_n^{-1/2}}{h}\le\frac{C}{\sqrt{kd}},
\ee
where we have used the fact that $\sum^{n_0}_{n=1}h^{-1}\mu^{-1/2}_n\le \frac 2\pi\mu_{n_0}^{1/2}\le C(kh)^{1/2}/d^{1/2}$ by
\eqref{f6}.

To proceed, let $n_1\ge 1$ be such that $x_1+p_{n_1}/2>d$ and $x_1+p_{n_1+1}/2\le d$, which is equivalent to
\bee\label{f7}
k^{-1}\mu_{n_1}<\frac{t}{\sqrt{4d_1^2+t^2}},\ \ k^{-1}\mu_{n_1+1}\ge\frac{t}{\sqrt{4d_1^2+t^2}},\ \ d_1:=d-x_1.
\eee
Clearly $n_1\le n_0$. We write
\be\label{f8}
& &\sum^{n_0}_{n=1}\frac{k^{1/2}t}{h\xi_n}\left|\int^{x_1+2p_n}_d\Theta^{-3/2}e^{\i f_n(x_1,\zeta_1,t)}d\zeta_1\right|\nn\\
&\le&
\sum^{n_1}_{n=1}\frac{k^{1/2}t}{h\xi_n}\left|\int^{x_1+p_n/2}_d\Theta^{-3/2}e^{\i f_n(x_1,\zeta_1,t)}d\zeta_1\right|\nn\\
&+&
\sum^{n_1}_{n=1}\frac{k^{1/2}t}{h\xi_n}\left|\int^{x_1+2p_n}_{x_1+p_n/2}\Theta^{-3/2}e^{\i f_n(x_1,\zeta_1,t)}d\zeta_1\right|\nn\\
&+&\sum^{n_0}_{n=n_1+1}\frac{k^{1/2}t}{h\xi_n}\left|\int^{x_1+2p_n}_d\Theta^{-3/2}e^{\i f_n(x_1,\zeta_1,t)}d\zeta_1\right|
:={\rm I}+{\rm II}+{\rm III}.
\ee
Let $n_2=\frac{kh}{\sqrt{kd}}$. If $n_1\le n_2$, then since $\xi_n\ge Ck$ for $n\le n_0$,
\ben
{\rm I}\le C\sum^{n_2}_{n=1}\frac{k^{1/2}t}{h\xi_n}d^{-1/2}\le C\frac{n_2}{\sqrt{kd}}\le C\frac hd.
\een
Otherwise, if $n_1\ge n_2+1$, we split the sum and use \eqref{f2} to have
\ben
{\rm I}&\le&C\frac hd+\sum^{n_1}_{n=n_2+1}\frac{k^{1/2}t}{h\xi_n}\left|\int^{x_1+p_n/2}_d\Theta^{-3/2}e^{\i f_n(x_1,\zeta_1,t)}d\zeta_1\right|\\&\le&C\frac hd+C\sum^{n_1}_{n=n_2+1}\frac{k^{1/2}t}{h\xi_n}\left(\frac{kt^2}{p_n^2}\right)^{-1}d^{-3/2}\le C\frac hd,
\een
where we have used the fact that 
\ben
\sum^{n_1}_{n=n_2+1}h^{-1}\mu_n^{-2}\le h^{-1}\mu_{n_2+1}^{-2}+\pi^{-1}\mu_{n_2+1}^{-1}\le Ch n_2^{-1}=Ck^{-1/2}d^{1/2}.
\een
Therefore, we have ${\rm I}\le Ch/d$. By using \eqref{f4} with $\eta=x_1+p_n/2$ for the term ${\rm II}$ and with 
$\eta=d\ge x_1+p_n/2$ for the term ${\rm III}$, we have
\ben
{\rm II}+{\rm III}\le C\sum^{n_0}_{n=1}\frac{k^{1/2}t}{h\xi_n}kt^{-1}\xi_n^{-3/2}\le C\sum^{n_0}_{n=1}\frac 1{h\xi_n}
\le C\frac hd,
\een
where we have used $n_0\le Ckh^2/d$ by \eqref{f6}. Therefore, by \eqref{f8},
\ben
\sum^{n_0}_{n=1}\frac{k^{1/2}t}{h\xi_n}\left|\int^{x_1+2p_n}_d\Theta^{-3/2}e^{\i f_n(x_1,\zeta_1,t)}d\zeta_1\right|
\le C\frac hd+\frac{C}{\sqrt{kd}}.
\een
This shows the estimate for $|S_d(x,y)|$ by \eqref{y4}, \eqref{y5}, \eqref{yy7}-\eqref{yy8}, \eqref{yy3}-\eqref{yy4}, and the above estimate.

The estimate for $\na_x S_d(x,y)$ can be proved similarly by noticing that
\ben
& &\frac{\pa^2 G(x,\zeta)}{\pa x_2\pa\zeta_2}=\frac{\pa f}{\pa t}(x_1,\zeta_1,h+x_2)+\frac{\pa f}{\pa t}(x_1,\zeta_1,h-x_2),\\
& &\frac{\pa^2 G(x,\zeta)}{\pa x_1\pa\zeta_2}=\frac{\pa f}{\pa x_1}(x_1,\zeta_1,h+x_2)-\frac{\pa f}{\pa x_1}(x_1,\zeta_1,h-x_2),
\een
where after using the identity $H^{(1)'}_1(\xi)=H^{(1)}_0(\xi)-\frac{1}{\xi}H^{(1)}_1(\xi)$ for any $\xi\in\C$, 
\ben
& &\frac{\pa f}{\pa t}(x_1,\zeta_1,t)=\frac{\i}{4}H^{(1)}_0(k\Theta)\frac{k^2t^2}{\Theta^2}+\frac{\i}4H^{(1)}_1(k\Theta)
\left(\frac k\Theta-\frac{2kt^2}{\Theta^3}\right),\\
& &\frac{\pa f}{\pa x_1}(x_1,\zeta_1,t)=\frac{\i}{4}H^{(1)}_0(k\Theta)\frac{k^2(x_1-\zeta_1)t}{\Theta^2}
-\frac{\i}4H^{(1)}_1(k\Theta)\frac{2k(x_1-\zeta_1)t}{\Theta^3}.
\een
We omit the details. This completes the proof.
\finproof

We remark that by Theorem \ref{S} and Theorem \ref{aperture}, the resolution of the finite aperture Helmholtz-Kirchhoff function $H_d(x,y)$ is the same as the resolution of $\Im N(x,y)$ for $x,y\in\Om$ when $kh\gg 1$ and $kd/(kh)\gg 1$.

\section{The resolution analysis of the RTM algorithm}

In this section we study the resolution of the imaging function in (\ref{cor2}). We first notice that
since $S(\cdot,z)$ satisfies the Helmholtz equation in $\R^2_h$, it follows from Theorem \ref{S} that
\be\label{SS1}
\|S(\cdot,z)\|_{H^{1/2}(\Ga_D)}+\Big\|\frac{\pa S(\cdot,z)}{\pa\nu}\Big\|_{H^{-1/2}(\Ga_D)}
&\le&C\|S(\cdot,z)\|_{H^1(D)}+\|\Delta S(\cdot,z)\|_{L^2(D)}\nn\\
&\le&\frac{C}{|\cos(kh)|}\frac 1{\sqrt{kh}},\ \ \ \ \forall z\in\Om,
\ee
for some constant $C$ independent of $h$. Similarly, since $S_d(\cdot,z)$ also satisfies the Helmholtz equation, by Theorem \ref{aperture}, for any $z\in\Om$,
\be\label{SS2}
\|S_d(\cdot,z)\|_{H^{1/2}(\Ga_D)}+\Big\|\frac{\pa S_d(\cdot,z)}{\pa\nu}\Big\|_{H^{-1/2}(\Ga_D)}
&\le&C\left(\frac 1{\sqrt{kh}}+\frac hd\right),
\ee
for some constant $C$ independent of $h,d$. Here we have used the assumption $d\ge c_3h$.

\begin{theorem}\label{resolution}
Let $kh>\pi/2$, $d\ge c_3 h$, and \eqref{cond-omega} be satisfied. For any $z\in\Om$, let $\psi(x,z)$ be the radiation solution of the problem
\be\label{ps1}
& &\De\psi(x,z)+k^2\psi(x,z)=0\ \ \ \ \mbox{in }\R^2_h\bks\bar D,\\
    & &    \psi=0\ \ \mbox{on }\Ga_0,\ \ \ \ \frac{\partial \psi}{\partial x_2} = 0 \ \ \mbox{on }\Ga_h, \label{ps2}\\
    & &     \frac{\pa  \psi}{\pa \nu} + \i k \eta  \psi = -\bigg(\frac{\pa  \,\Im N(x,z)}{\pa \nu} + \i k \eta  \,\Im N(x,z) \bigg) \ \ \ \ \mbox{on }\Ga_D. \label{ps3}
\ee
Then we have, for any $z\in\Om$,
\ben
\hat I_d(z)=2h \sum_{n=1}^{M}\xi_n(|\psi_n^{+}|^2+|\psi_n^{-}|^2)+ 4 k\int_{\Ga_D}\eta(\zeta)\left|\psi(\xi,z)+\Im N(\zeta,z)\right|^2ds(\xi)+ w_{\hat I}(z),
\een
where $\psi^\pm_n$, $n=1,2,\cdots,M$, are the far-field pattern of the radiation solution of $\psi(\cdot,z)$ and $\|w_{\hat I}\|_{L^\infty(\Om)}\le \frac {C}{|\cos(kh)|}\frac 1{\sqrt{kh}}+C\frac hd$.
\end{theorem}

\begin{proof} By the integral representation formula we know that
\ben
u^s(x_r,x_s)=\int_{\Ga_D}\left(u^s(\zeta,x_s)\frac{\pa N(x_r,\zeta)}{\pa\nu(\zeta)}-\frac{\pa u^s(\zeta,x_s)}{\pa\nu(\zeta)}N(x_r,\zeta)\right)ds(\zeta).
\een
By using Lemma \ref{S} we obtain that, for any $z\in\Om$,
\ben
\hat v_b(z,x_s)&=&\int_{\Ga_h^d} \frac{\pa G(z,x_r)}{\pa x_2(x_r)} \overline{u^s(x_r,x_s)}ds(x_r)\\
&=&\int_{\Ga_D}\overline{ u^s(\zeta,x_s)}\frac{\pa}{\pa\nu(\zeta)}\big(2\i \Im{N(\zeta,z)} -  S(\zeta,z)-S_d(\zeta,z)\big)ds(\zeta)\\
&-&\int_{\Ga_D}\frac{\pa \overline{ u^s(\zeta,x_s)} }{\pa\nu(\zeta)}\big(2\i \Im{N(\zeta,z)} -  S(\zeta,z)-S_d(\zeta,z)\big)ds(\zeta),
\een
where we have used the reciprocity relation $N(\zeta,z)=N(z,\zeta), G(\zeta,z)=G(z,\zeta)$.
By (\ref{cor2}) we obtain then
\be\label{cor4}
\hat I_d(z)&=& \Im\int_{\Ga_D}v_s(\zeta,z)\frac{\pa}{\pa\nu(\zeta)}\big(2\i\,\Im{N(\zeta,z)} -  S(\zeta,z)-S_d(\zeta,z)\big)ds(\zeta)\nn\\
&-&\,\Im\int_{\Ga_D}\frac{\pa v_s(\zeta,z)}{\pa\nu(\zeta)}\big(2\i \Im{N(\zeta,z)} -  S(\zeta,z)-S_d(\zeta,z)\big)ds(\zeta),
\ee
where $v_s(\zeta,z)=\int_{\Ga_h^d} \frac{\pa G(z,x_s)}{\pa x_2(x_s)}\overline{u^s(\zeta,x_s)}ds(x_s)$. By taking the complex conjugate, we have
\ben
\overline{v_s(\zeta,z)}=\int_{\Ga_h^d} \frac{\pa\overline{G(z,x_s)}}{\pa x_2(x_s)}u^s(\zeta,x_s)ds(x_s).
\een
Thus $\overline{v_s(\zeta,z)}$ is a weighted superposition of the scattered waves $u^s(\zeta,x_s)$. Therefore, $\overline{v_s(\zeta,z)}$ is the radiation solution of the Helmholtz equation
\ben
& & \De_\zeta\overline{v_s(\zeta,z)}+k^2\overline{v_s(\zeta,z)}=0\ \ \ \ \mbox{in } \R^2_h \bks\bar D, \\
 & &    \overline{v_s(\zeta,z)} = 0 \ \ \mbox{on }\Ga_0,\ \ \ \ \frac{\partial \overline{v_s(\zeta,z)}}{\partial \zeta_2} = 0\ \ \mbox{on }\Ga_h,
\een
satisfying the impedance boundary condition
\ben
& &\left(\frac{\pa}{\pa\nu(\zeta)}+\i k\eta(\zeta)\right)\overline{v_s(\zeta,z)}\\
&=&\int_{\Ga_h^d} \frac{\pa\overline{G(z,x_s)}}{\pa x_2(x_s)}\left(\frac{\pa}{\pa\nu(\zeta)}+\i k\eta(\zeta)\right)u^s(\zeta,x_s)ds(x_s)\\
&=&\int_{\Ga_h^d} \frac{\pa\overline{G(z,x_s)}}{\pa x_2(x_s)}\left(\frac{\pa}{\pa\nu(\zeta)}+\i k\eta(\zeta)\right)(-N(\zeta,x_s))ds(x_s) \\
&=&  \left(\frac{\pa}{\pa\nu(\zeta)}+\i k\eta(\zeta)\right)( 2\i\,\Im{N(\zeta,z)} + \overline{ S(\zeta,z)}+\overline{S_d(\zeta,z)}) \ \ \ \ \mbox{on }\Ga_D,
\een
where we have used \eqref{y1} in the last equality. This implies by using (\ref{ps1})-(\ref{ps3}) that $\overline{v_s(\zeta,z)}=-2\i\psi(\zeta,z)+w(\zeta,z)$, where
$w(\cdot,z)$ satisfies the impedance scattering problem in Theorem \ref{LAP} with $g(\cdot)= \left(\frac{\pa}{\pa\nu}+\i k\eta(\cdot)\right)( \overline{ S(\cdot,z) }+\overline{S_d(\cdot,z)})$.

By Theorem \ref{LAP}, \eqref{SS1}-\eqref{SS2}, and the boundary condition satisfied
by $w$ on $\Ga_D$, we know that $w$ satisfies
\ben
\|w(\cdot,z)\|_{H^{1/2}(\Ga_D)}+\Big\|\frac{\pa w(\cdot,z)}{\pa\nu}\Big\|_{H^{-1/2}(\Ga_D)}\le \frac{C}{|\cos(kh)|}\frac 1{\sqrt{kh}}+C\frac hd.
\een
Moreover, since $N(\cdot,z)=G(\cdot,z)+S(\cdot,z)$, we also have
\ben
\|N(\cdot,z)\|_{H^{1/2}(\Ga_D)}+\Big\|\frac{\pa N(\cdot,z)}{\pa\nu}\Big\|_{H^{-1/2}(\Ga_D)}\le C, \ \ \ \ \forall z\in\Om.
\een
Now substituting $v_s(\zeta,z)=2\i\overline{\psi(\zeta,z)}+\overline{w(\zeta,z)}$ into (\ref{cor4}) we obtain
\ben
 \hat I(z)
&=&-4\,\Im\int_{\Ga_D}\left(\overline{\psi(\zeta,z)}\frac{\pa\Im N(\zeta,z)}{\pa\nu(\zeta)}-\frac{\pa\overline{\psi(\zeta,z)}}
{\pa\nu(\zeta)}\Im N(\zeta,z)\right)ds(\zeta)+w_{\hat I}(z)\\
&=&4\,\Im\int_{\Ga_D}\left(\psi(\zeta,z)\frac{\pa\Im N(\zeta,z)}{\pa\nu(\zeta)}-\frac{\pa\psi(\zeta,z)}{\pa\nu(\zeta)}\Im N(\zeta,z)\right)ds(\zeta)
+w_{\hat I}(z),
\een
where
\ben
|w_{\hat I}(z)|&\le&2\|\psi(\cdot,z)\|_{H^{1/2}(\Ga_D)}\Big\|\frac{\pa (S(\cdot,z)+S_d(\cdot,z))}{\pa\nu}\Big\|_{H^{-1/2}(\Ga_D)}\\
&&+2\|S(\cdot,z)+S_d(\cdot,z)\|_{H^{1/2}(\Ga_D)}\Big\|\frac{\pa \psi(\cdot,z)}{\pa\nu}\Big\|_{H^{-1/2}(\Ga_D)}\\
&&+\|w(\cdot,z)\|_{H^{1/2}(\Ga_D)}\Big\|\frac{\pa (2\i\,\Im N(\cdot,z)-S(\cdot,z)-S_d(\cdot,z))}{\pa\nu}\Big\|_{H^{-1/2}(\Ga_D)}\\
&&+\|2\i\,\Im N(\cdot,z)-S(\cdot,z)-S_d(\cdot,z)\|_{H^{1/2}(\Ga_D)}\Big\|\frac{\pa w(\cdot,z)}{\pa\nu}\Big\|_{H^{-1/2}(\Ga_D)}\\
&\le&\frac{C}{|\cos(kh)|}\frac 1{\sqrt{kh}}+C\frac hd.
\een
By (\ref{ps2}) we have
\ben
&& \Im\int_{\Ga_D}\left(\psi(\zeta,z)\frac{\pa\Im N(\zeta,z)}{\pa\nu(\zeta)}-\frac{\pa\psi(\zeta,z)}{\pa\nu(\zeta)}\Im N(\zeta,z)\right)ds(\zeta)\\
& =&\Im\int_{\Ga_D}\Big[\psi(\zeta,z)\left(\frac{\pa\Im N(\zeta,z)}{\pa\nu(\zeta)}-\i k\eta(\zeta)\Im N(\zeta,z)\right)\\
&&-\left(\frac{\pa\psi(\zeta,z)}{\pa\nu(\zeta)}+\i k\eta(\zeta)\psi(\zeta,z)\right)\Im N(\zeta,z)+2\i k\eta(\zeta)\Im N(\zeta,z)\psi(\zeta,z)\Big]ds(\zeta)\\
&=&\Im\int_{\Ga_D}\Big[-\psi(\zeta,z)\cdot \left(\frac{\pa\overline{\psi(\zeta,z)}}{\pa\nu(\zeta)}-\i k\eta(\zeta)\overline{\psi(\zeta,z)}\right)\\
&&+\left(\frac{\pa\Im N(\zeta,z)}{\pa\nu(\zeta)}+\i k\eta(\zeta)\Im N(\zeta,z)\right)\Im N(\zeta,z)+2\i k\eta(\zeta)\Im N(\zeta,z)\psi(\zeta,z)\Big]ds(\zeta)\\
&=&-\Im\int_{\Ga_D}\psi(\zeta,z)\frac{\pa\overline{\psi(\zeta,z)}}{\pa\nu(\zeta)}ds(\zeta)+k\int_{\Ga_D}\eta(\zeta)\left|\psi(\zeta,z)+\Im N(\zeta,z)\right|^2ds(\zeta).
\een
By \eqref{DtN} we know that the far-field pattern $\psi^\pm_n$, $n=1,2,\cdots,M$, satisfy
\ben
-\,\Im{ \int_{\Ga_D}\psi\frac{\pa\bar \psi}{\pa\nu}ds }=\frac h2\sum_{n=1}^{M}\xi_n(|\psi_n^{+}|^2+|\psi_n^{-}|^2).
\een
This completes the proof. \qquad
\finproof

We remark that $\psi(x,z)$ is the scattering solution of the Helmholtz equation in the waveguide with the incoming field $\Im N(x,z)$. Since
\ben
\Im N(x,z)&=&\Im G(x,z)+\Im S(x,z)\\
&=& \frac 14 J_0(k|x-z|)-\frac 14 J_0(k|x-z'|)+\Im S(x,z),
\een
where $J_0(t)$ is the first kind Bessel function of zeroth order and $z'=(z_1,-z_2)$ is the imagine point of $z=(z_1,z_2)$. It is well-known that $J_0(t)$ peaks at $t=0$ and decays like $t^{-1/2}$ away from the origin. By Theorem \ref{S}, $S(x,z)$ is small when $kh\gg 1$ which implies $\Im N(x,z)$ of the problem \eqref{ps1}-\eqref{ps3} will peak at the boundary of the scatterer $D$ and becomes small when $z$ moves away from $\pa D$. Thus we expect that the imaging function $\hat I_d(z)$ will have contrast at the boundary of the scatterer $D$ and decay outside the boundary $\pa D$ if $kh\gg 1$ and $kd/(kh)\gg 1$. This is indeed confirmed in our numerical experiments.

\section{Extensions}{\label{extensions}}
In this section we consider the reconstruction of the sound soft and penetrable obstacles in the planar waveguide by our RTM algorithm.
For the sound soft obstacle, the measured data $u^s(x_r,x_s)=u(x_r,x_s)-u^i(x_r,x_s)$, where $u(x,x_s)$ is the radiation solution of the following problem
\be
& & \Delta u + k^2 u =-\delta_{x_s}(x) \qquad \mbox{in } \R^2_h\bs\bar D, \label{h1}\\
& & u = 0 \ \ \ \ \mbox{ on } \Ga_D, \label{h2}\\
& & u= 0\ \ \mbox{on }\Ga_0,\ \ \ \ \frac{\partial u}{\partial x_2}= 0\ \ \mbox{on }\Ga_h. \label{h3}
\ee
The well-posedness of the problem under some geometric condition of the obstacle $D$ is known \cite{MW1, RAMM}. Here we assume that the scattering problem
\eqref{h1}-\eqref{h3} has a unique solution. By modifying
the argument in Theorem \ref{resolution} we can show the following result whose proof is omitted.

\begin{theorem}\label{res_dir}
Let $kh>\pi/2$, $d\ge c_3h$, and \eqref{cond-omega} be satisfied. For any $z\in\Om$, let $\psi(x,z)$ be the radiation solution of the problem
\ben
& &\De\psi(x,z)+k^2\psi(x,z)=0\ \ \ \ \mbox{in }\R^2_h\bks\bar D,\\
    & &    \psi=0\ \ \mbox{on }\Ga_0,\ \ \ \ \frac{\partial \psi}{\partial x_2} = 0 \ \ \mbox{on }\Ga_h, \\
    & &     \psi(x,z) = - \Im N(x,z) \ \ \ \ \mbox{on }\Ga_D.
\een
Then we have, for any $z\in\Om$,
\ben
\hat I_d(z)=2h \sum_{n=1}^{M}\xi_n(|\psi_n^{+}|^2+|\psi_n^{-}|^2)+ w_{\hat I}(z),
\een
where $\psi^\pm_n$, $n=1,2,\cdots,M$, are the far-field pattern of the radiation solution of $\psi(\cdot,z)$ and $\|w_{\hat I}\|_{L^\infty(\Om)}\le \frac {C}{|\cos(kh)|}\frac 1{\sqrt{kh}}+C\frac hd$.
\end{theorem}

For the penetrable obstacle, the measured data $u^s(x_r,x_s)=u(x_r,x_s)-u^i(x_r,x_s)$, where $u(x,x_s)$ is the radiation solution of the following problem
\be
& & \Delta u + k^2n(x)u =-\delta_{x_s}(x) \qquad \mbox{in } \R^2_h, \label{hh1}\\
& & u= 0\ \ \mbox{on }\Ga_0,\ \ \ \ \frac{\partial u}{\partial x_2}= 0\ \ \mbox{on }\Ga_h, \label{hh2}
\ee
where $n(x)\in L^\infty(\R^2_h)$ is a positive function which is equal to $1$ outside the scatterer $D$.
The well-posedness of the problem under some condition on $n(x)$ is known \cite{Li}. Here we assume that the scattering problem \eqref{hh1}-\eqref{hh2} has a unique solution. By modifying the argument in Theorem \ref{resolution}, the following theorem can be proved. We refer to \cite[Theorem 3.1]{cch_a} for a similar result. Here we
omit the details.

\begin{theorem}\label{res_penetrable}
Let $kh>\pi/2$, $d\ge c_3h$, and \eqref{cond-omega} be satisfied. For any $z\in\Om$, let $\psi(x,z)$ be the radiation solution of the problem
\ben
& &\De\psi(x,z)+k^2n(x)\psi(x,z) = -k^2(n(x)-1)\Im N(x,z)\ \ \ \ \mbox{in }\R^2_h,\\
    & &    \psi=0\ \ \mbox{on }\Ga_0,\ \ \ \ \frac{\partial \psi}{\partial x_2} = 0 \ \ \mbox{on }\Ga_h.
\een
Then we have, for any $z\in\Om$,
\ben
\hat I_d(z)=2h \sum_{n=1}^{M}\xi_n(|\psi_n^{+}|^2+|\psi_n^{-}|^2)+ w_{\hat I}(z),
\een
where $\psi^\pm_n$, $n=1,2,\cdots,M$, are the far-field pattern of the radiation solution of $\psi(\cdot,z)$ and $\|w_{\hat I}\|_{L^\infty(\Om)}\le \frac {C}{|\cos(kh)|}\frac 1{\sqrt{kh}}+C\frac hd$.
\end{theorem}

We remark that for the penetrable scatterers, $\psi(x,z)$ is again the scattering solution with the incoming field $\Im N(x,z)$. Therefore we again expect the imaging function $\hat I_d(z)$ will have contrast on the boundary of the scatterer and decay outside the scatterer if
$kh\gg 1$ and $kd/(kh)\gg 1$.

\section{Numerical experiments}

In this section we present several numerical examples to demonstrate the effectiveness of our RTM method
for planar acoustic waveguide. To synthesize the scattering data we compute the solution $u^s(x_r,x_s)$ of the scattering problem by representing the ansatz solution as the double layer potential with the Green function $N(x,y)$ as the kernel and discretizing the integral equation by standard Nystr\"{o}m methods \cite{colton_kress}. The boundary integral equations on $\Ga_D$ are solved on a uniform mesh over the boundary with ten points per probe wavelength. The sources and receivers are both placed on the surface $\Ga_h^d$ with equal-distribution, where $d$ is the aperture. The boundaries of the obstacles used in our numerical experiments are parameterized as follows:
\ben
&\mbox{Circle:}\ \  \ &x_1=\rho\cos(\theta),\ \ x_2=\rho\sin(\theta),\ \ \theta\in (0,2\pi],\\
&\mbox{Kite:}\ \ \ &x_1=\cos(\theta) + 0.65\cos(2\theta) - 0.65,\ \ x_2=1.5 \sin (\theta),\ \ \theta\in (0,2\pi],\\
&\mbox{Rounded Square:} \ \ \ \ &x_1=0.5(\cos^3(\theta) + \cos(\theta)),\   x_2=0.5(\sin^3(\theta) + \sin(\theta)),\ \theta\in (0,2\pi].
\een

\textbf{Example 1}.
In this example we consider the imaging of a sound soft circle of radius $\rho=1$. We compare the results by
using our RTM function \eqref{cor2} and the Kirchhoff migration imaging function \eqref{cor3} for different values of
the aperture $d$. We take the probe wavelength $\lam = 0.5$, where $\lam=2\pi/k$, the thickness $h=10$, 
and $N_s=N_r=401$. We choose the aperture $d=10,20,30,50$ for the tests.

\begin{figure}
    \centering
    \includegraphics[width=0.24\textwidth, height=1.2in]{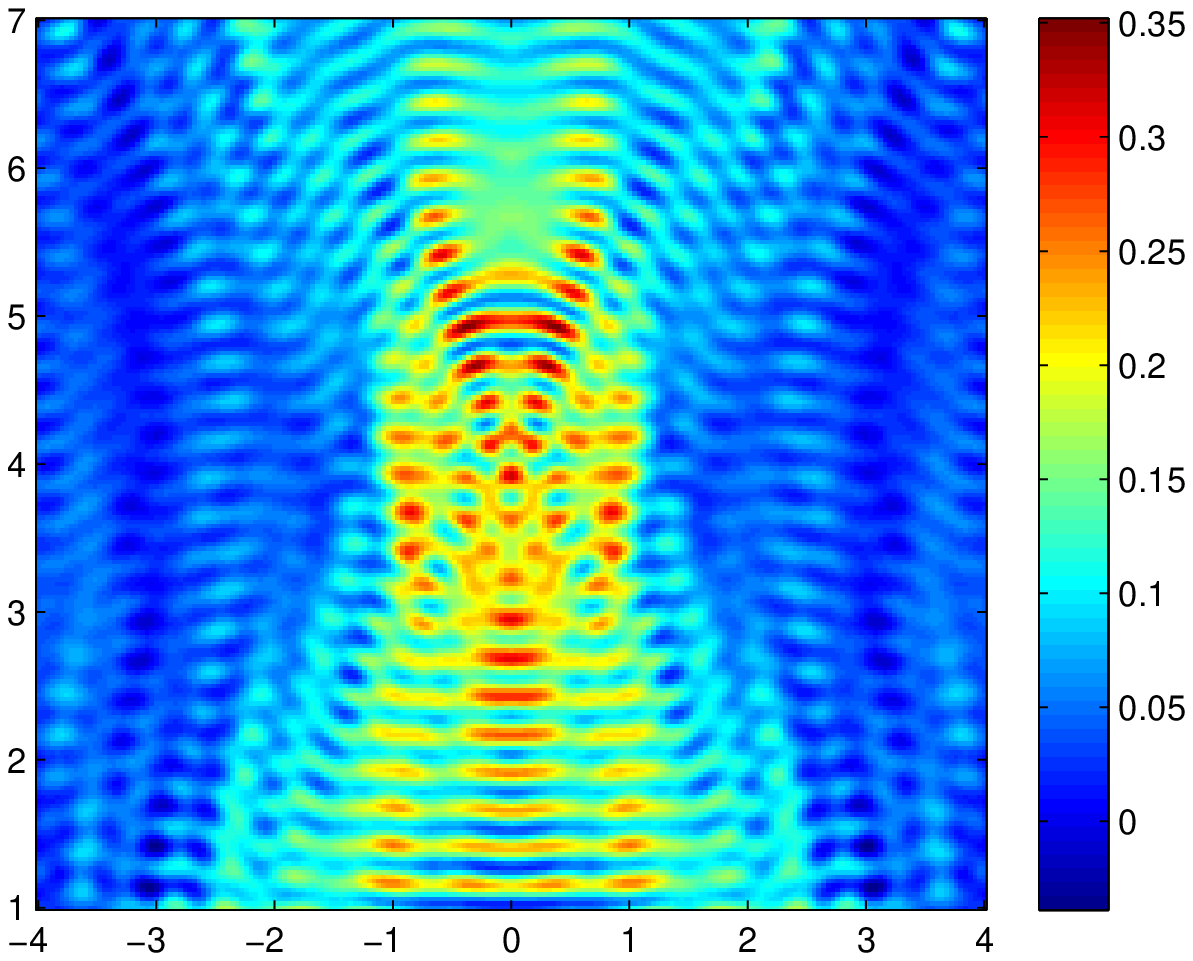}
    \includegraphics[width=0.24\textwidth, height=1.2in]{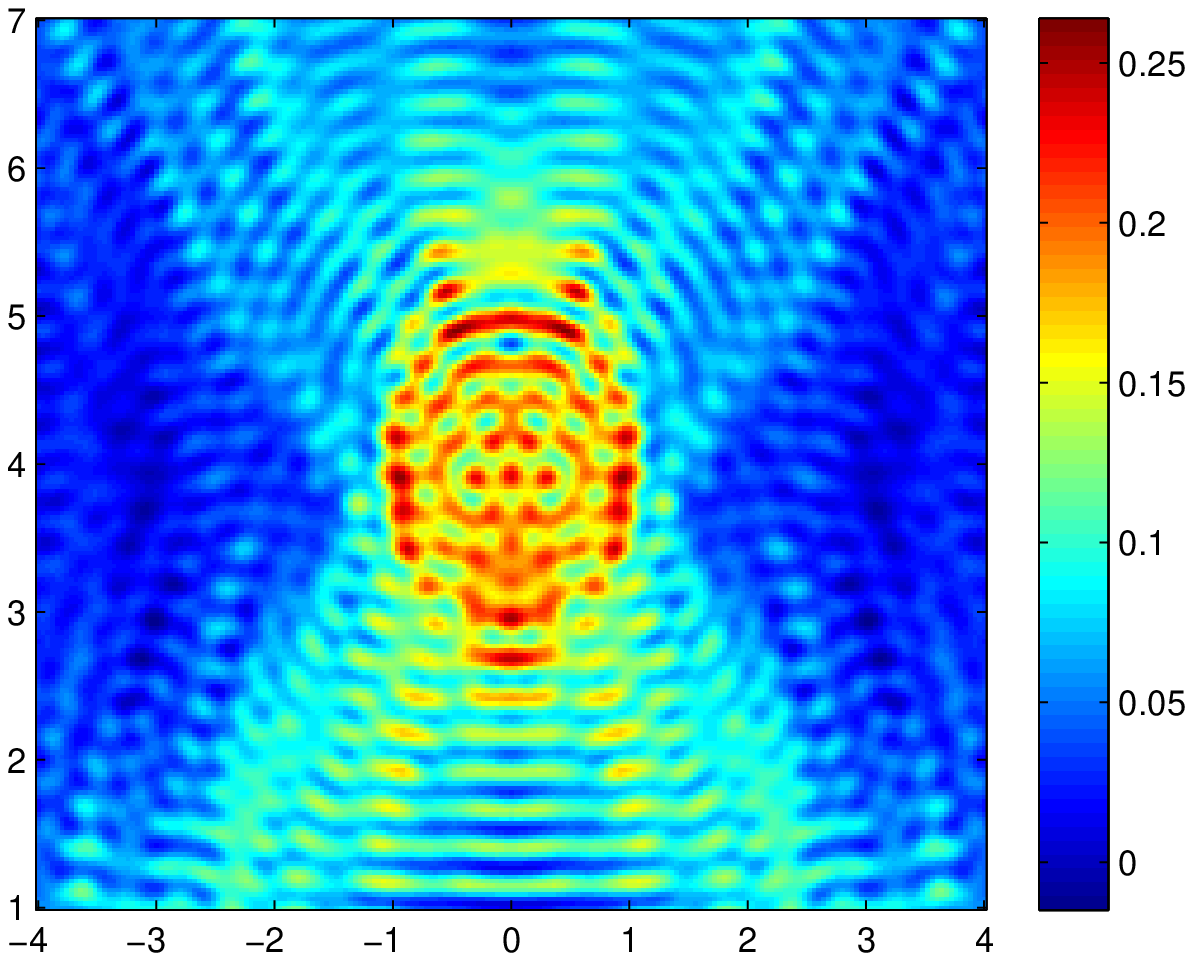}
    \includegraphics[width=0.24\textwidth, height=1.2in]{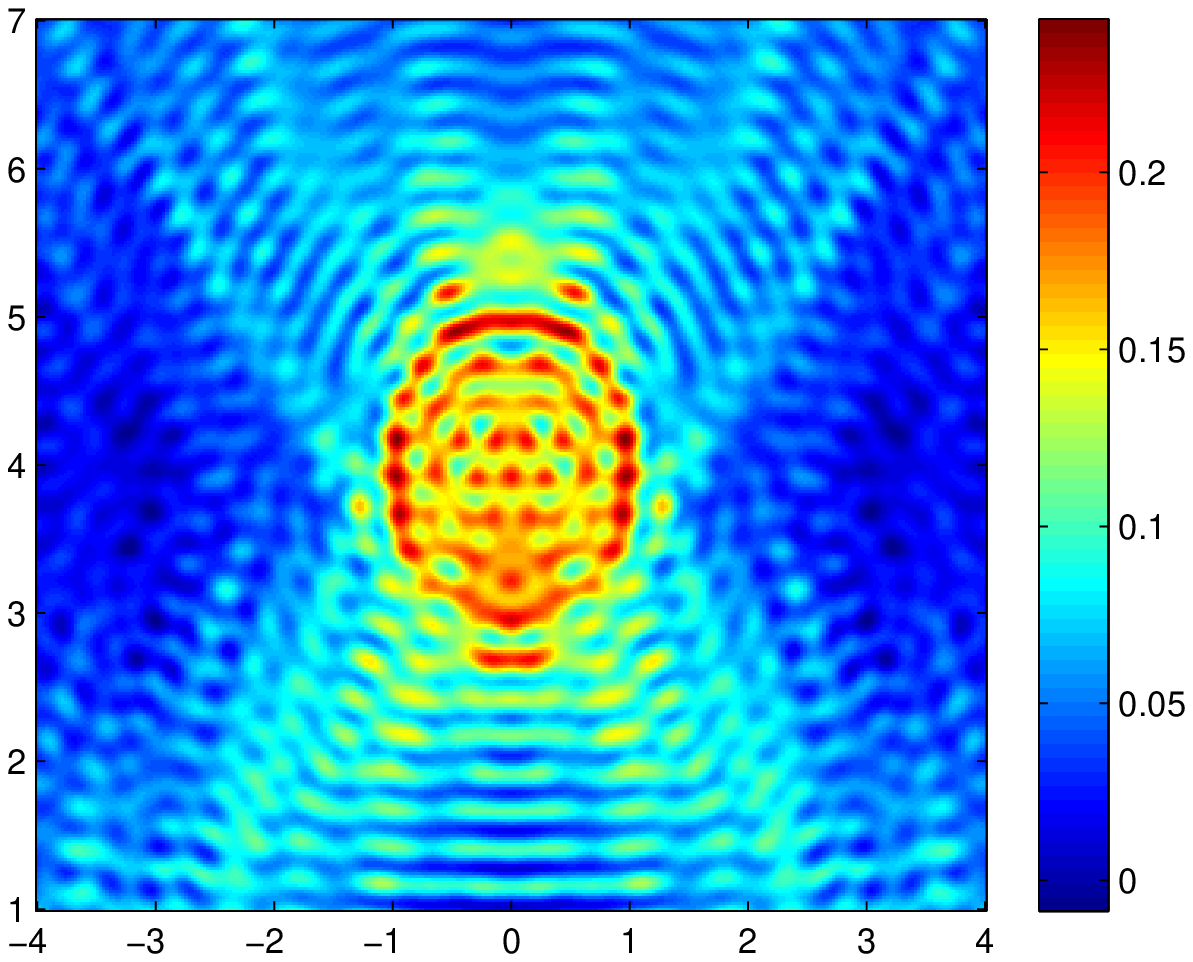}
    \includegraphics[width=0.24\textwidth, height=1.2in]{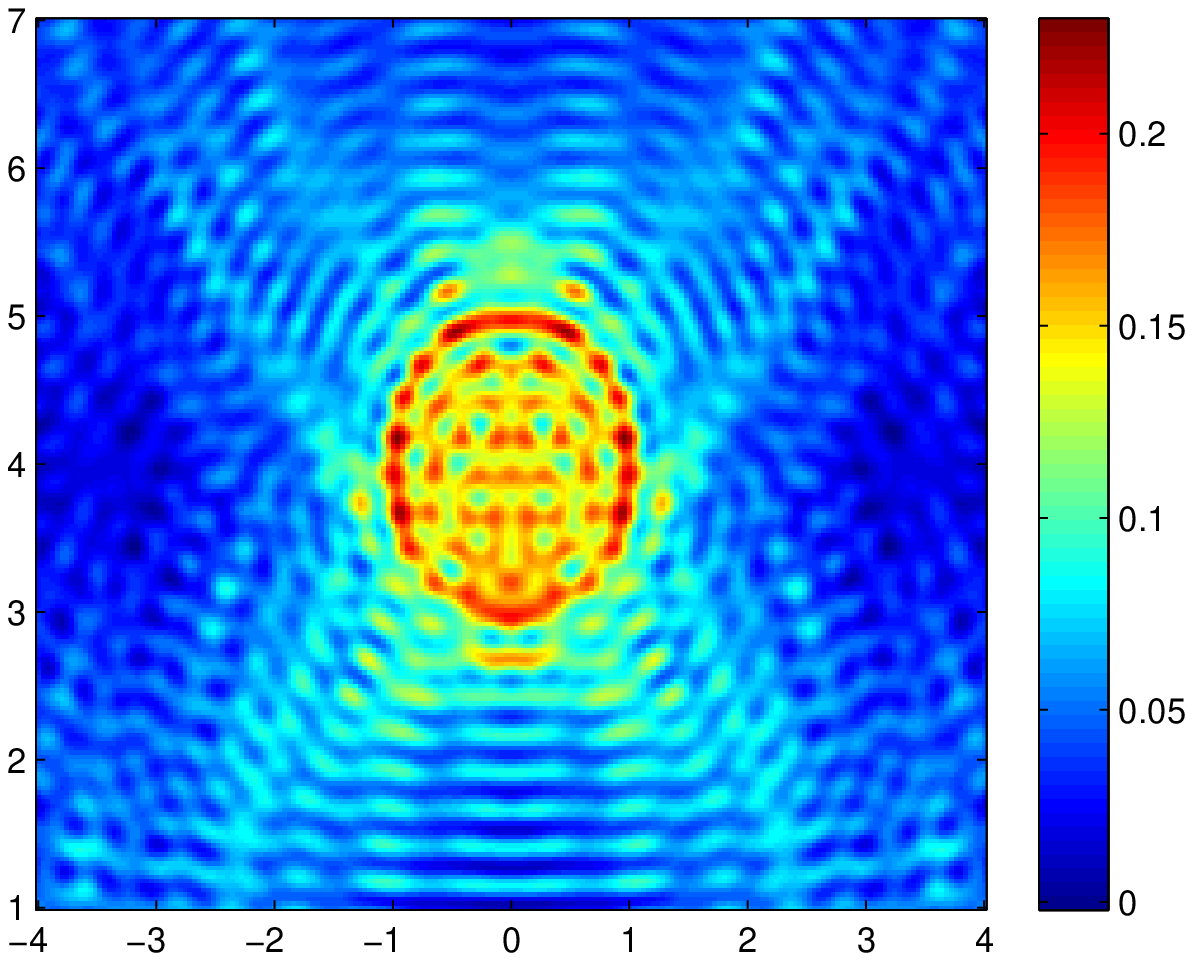}
    \includegraphics[width=0.24\textwidth, height=1.2in]{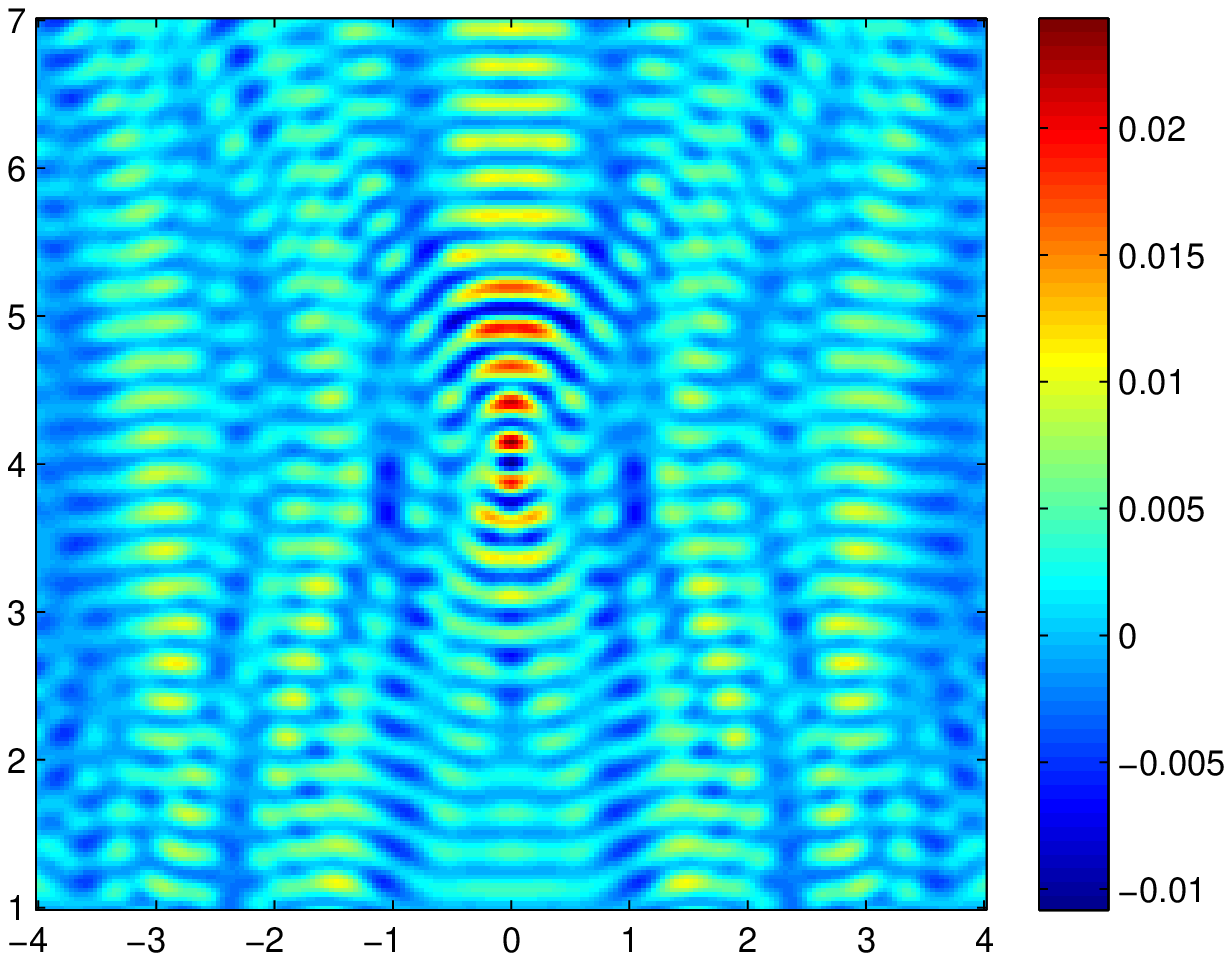}
    \includegraphics[width=0.24\textwidth, height=1.2in]{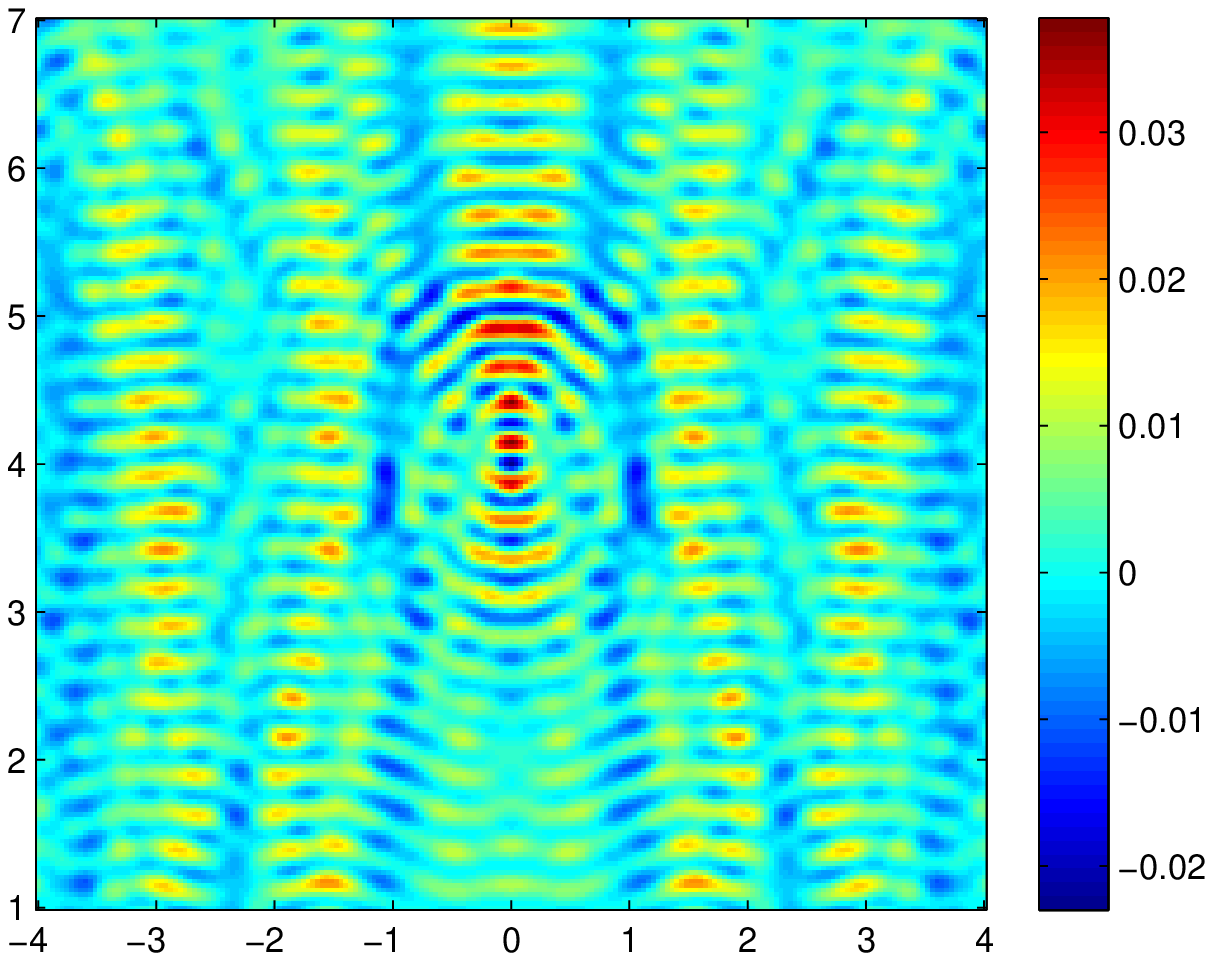}
    \includegraphics[width=0.24\textwidth, height=1.2in]{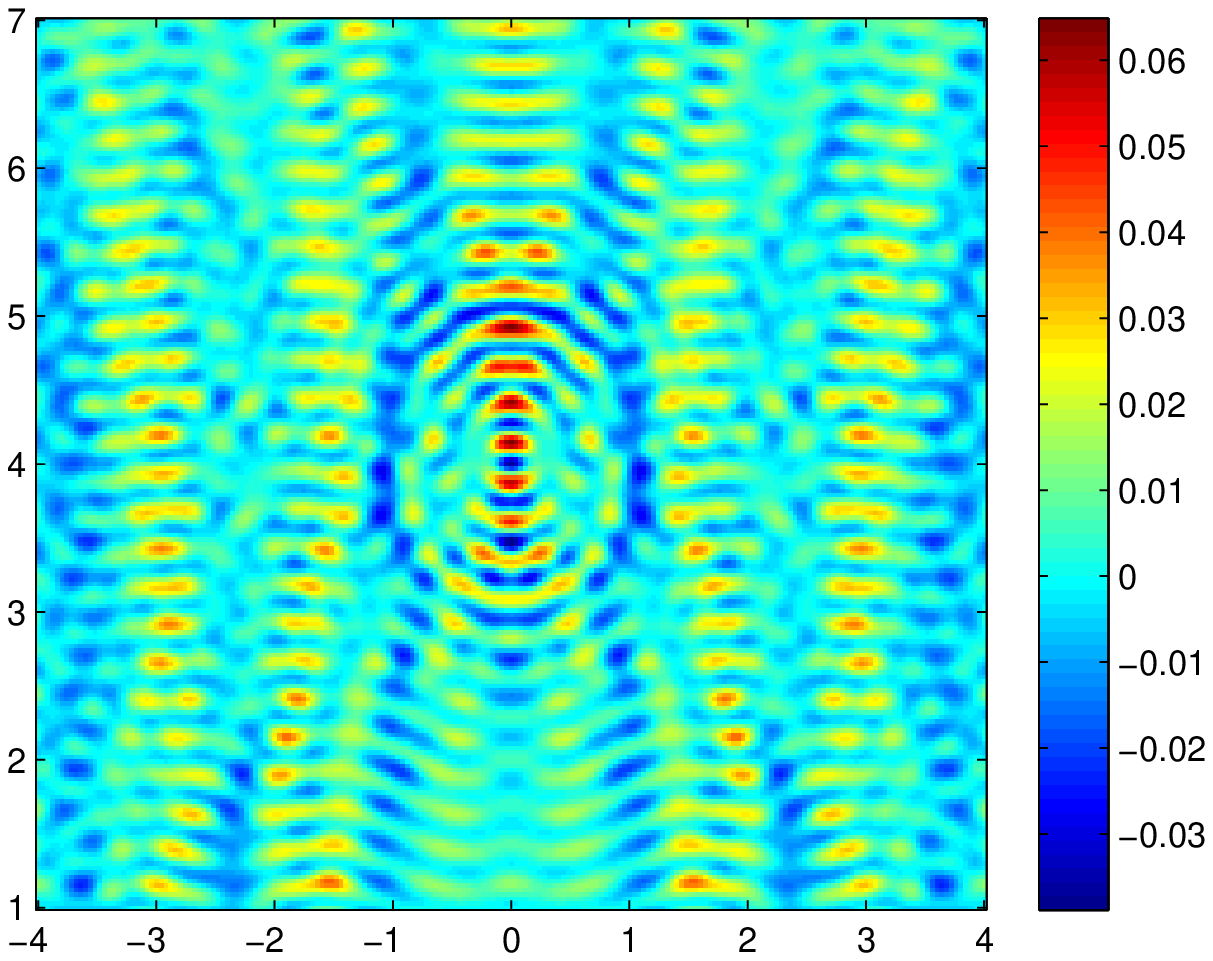}
    \includegraphics[width=0.24\textwidth, height=1.2in]{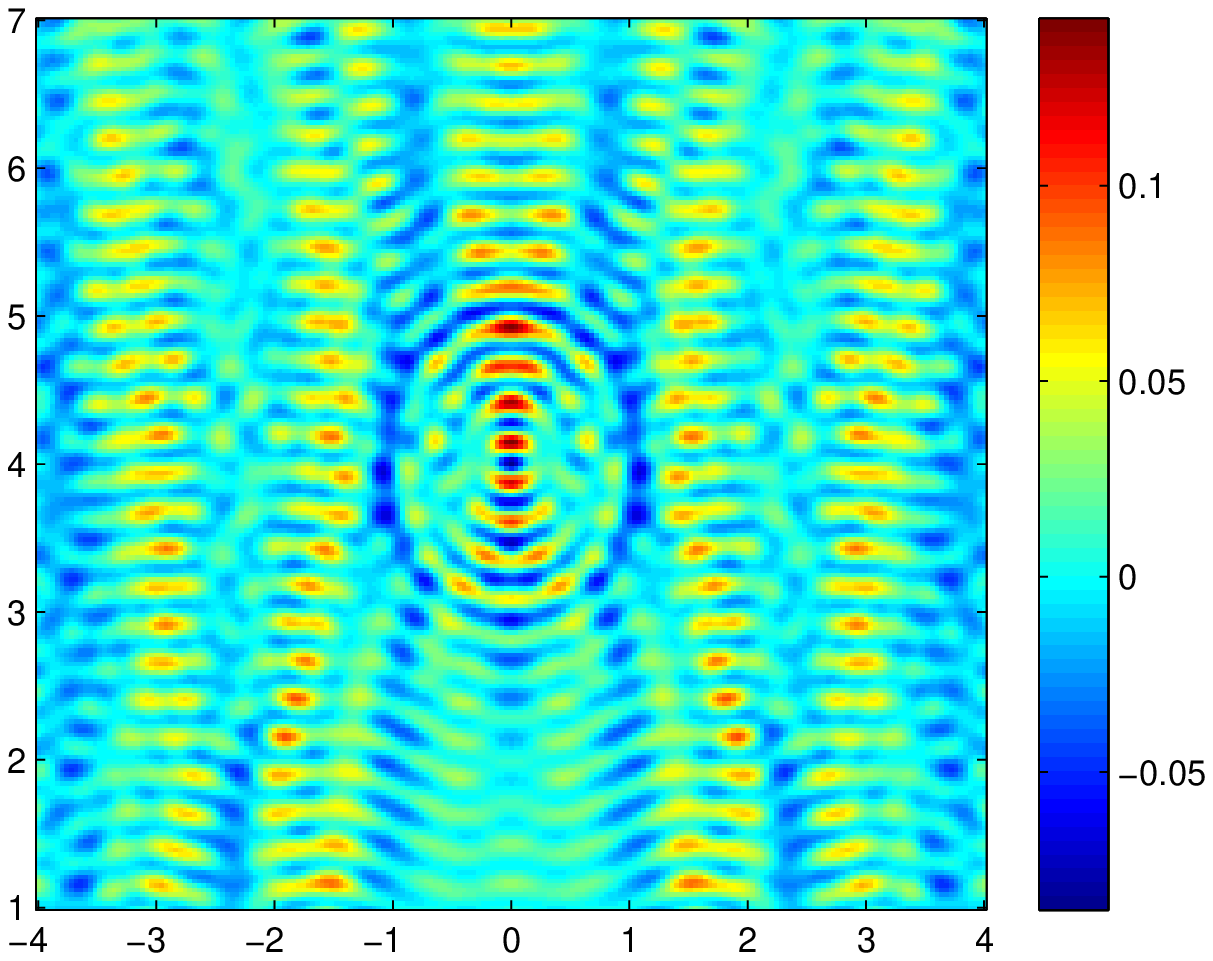}
    \includegraphics[width=0.24\textwidth, height=1.2in]{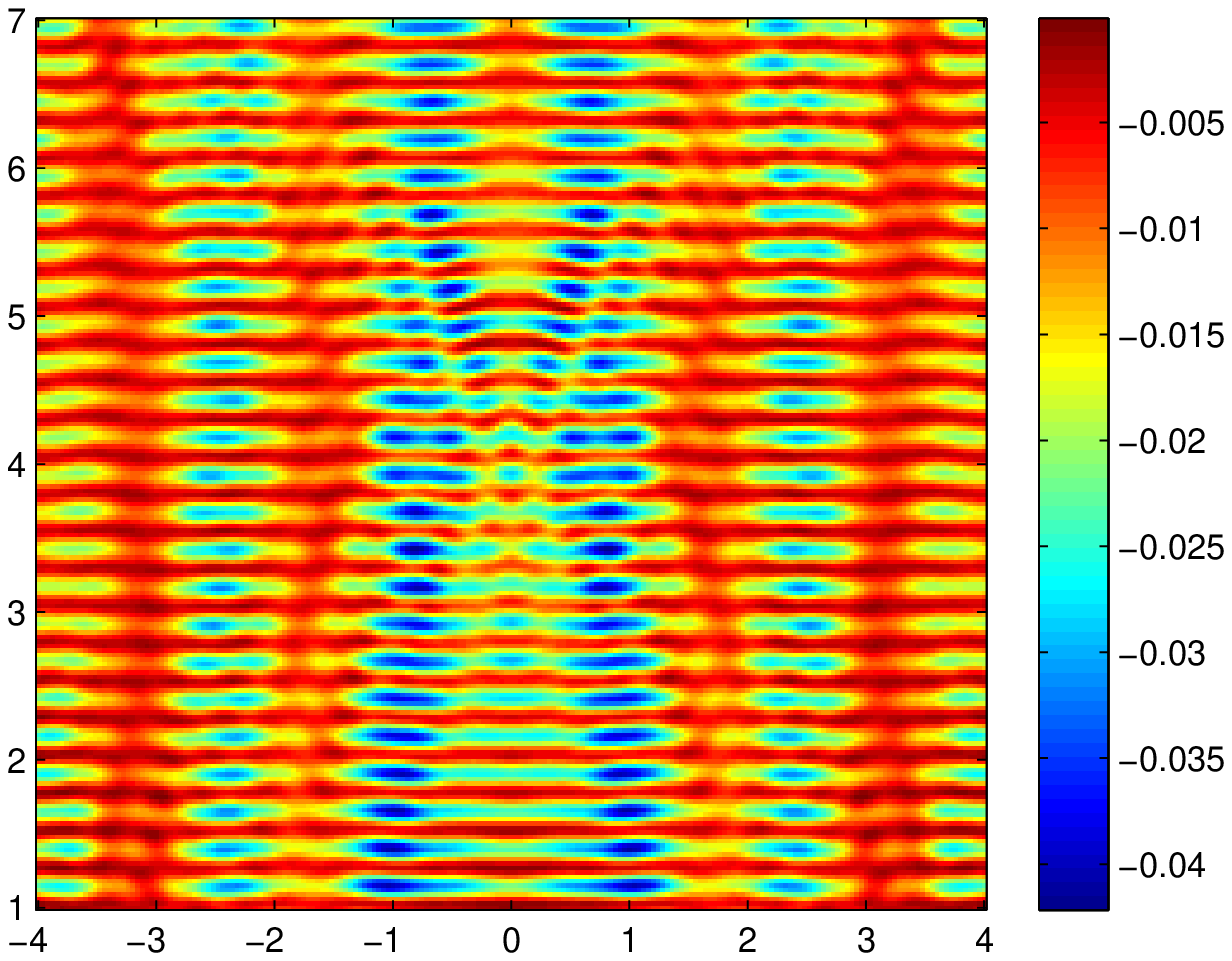}
    \includegraphics[width=0.24\textwidth, height=1.2in]{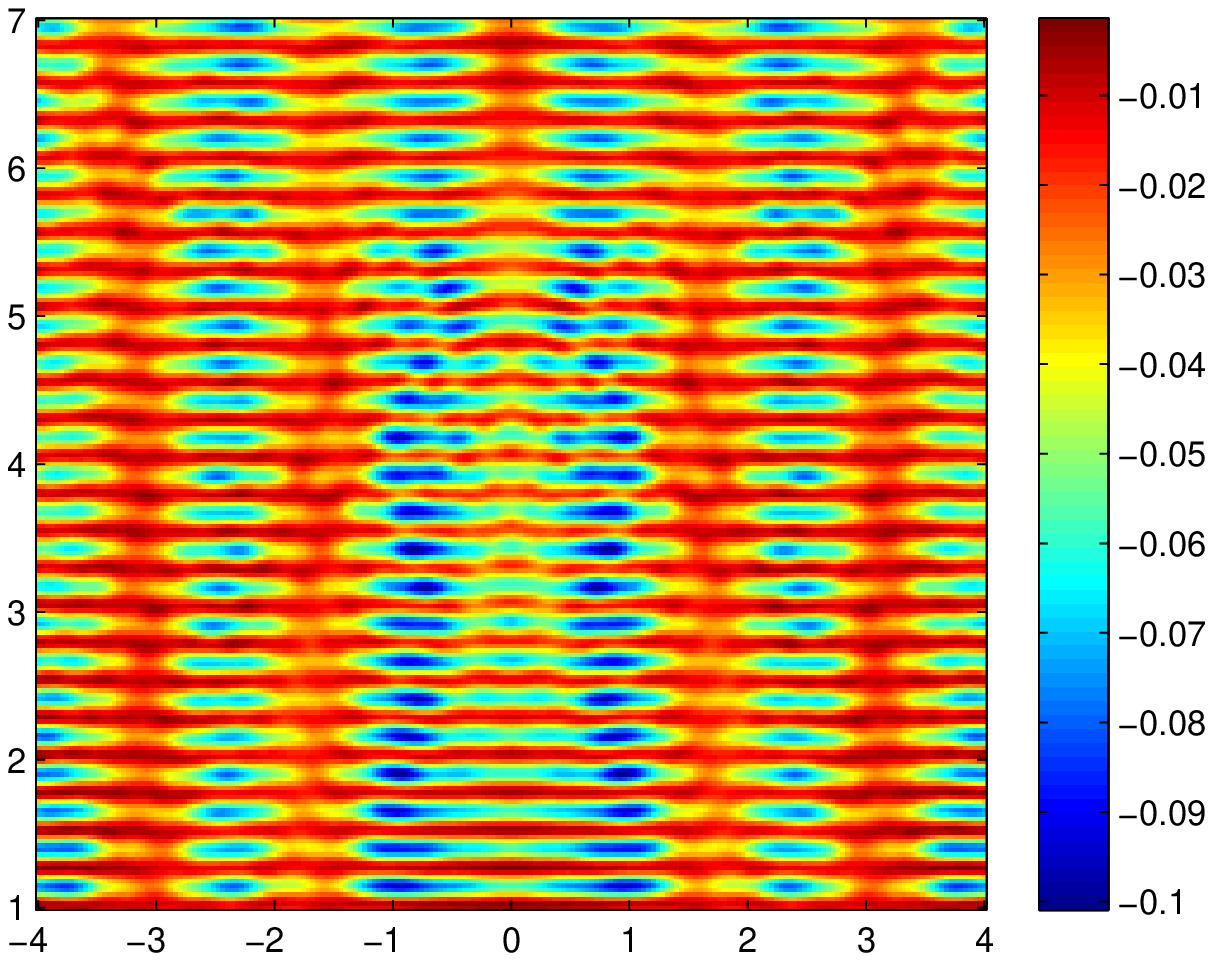}
    \includegraphics[width=0.24\textwidth, height=1.2in]{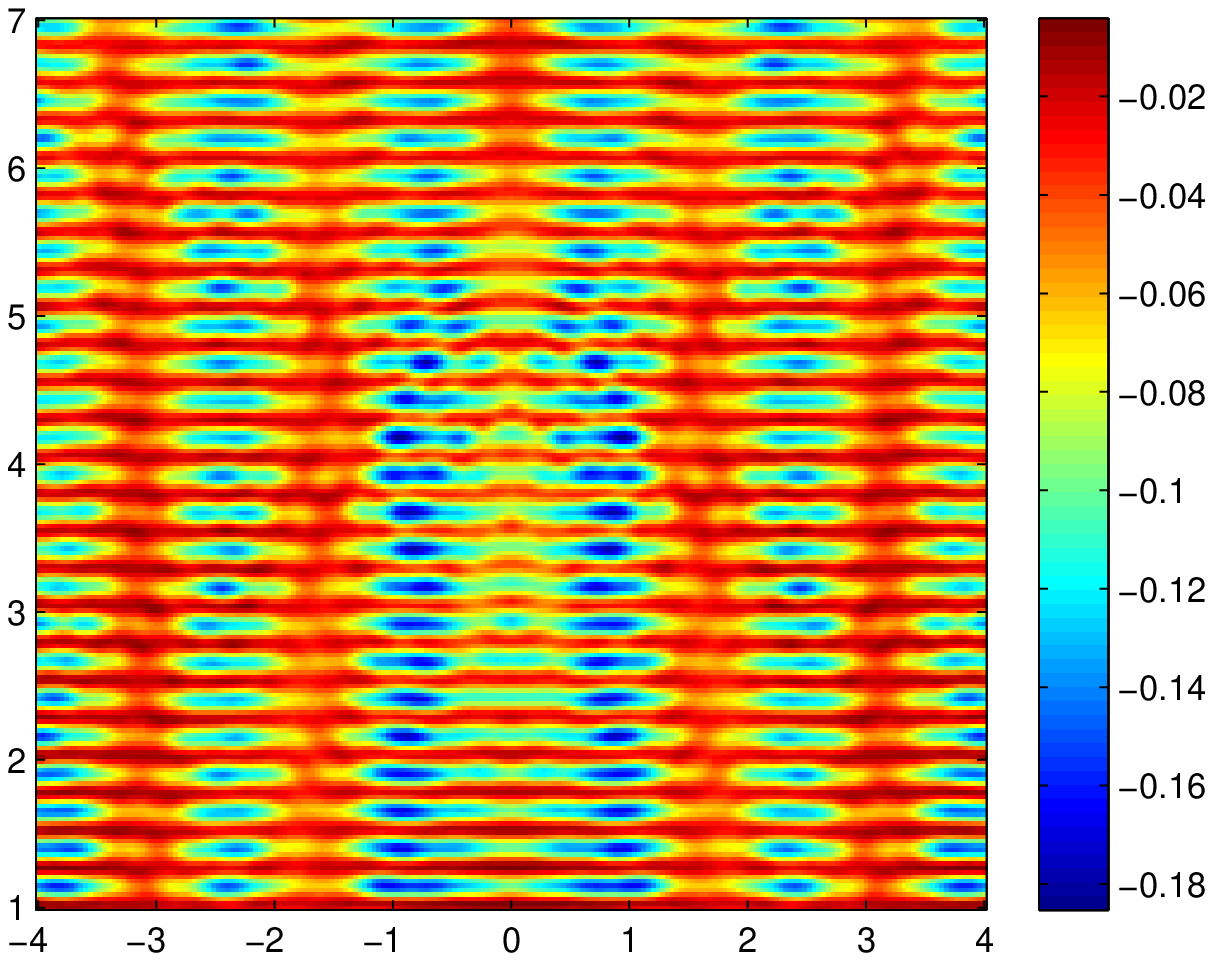}
    \includegraphics[width=0.24\textwidth, height=1.2in]{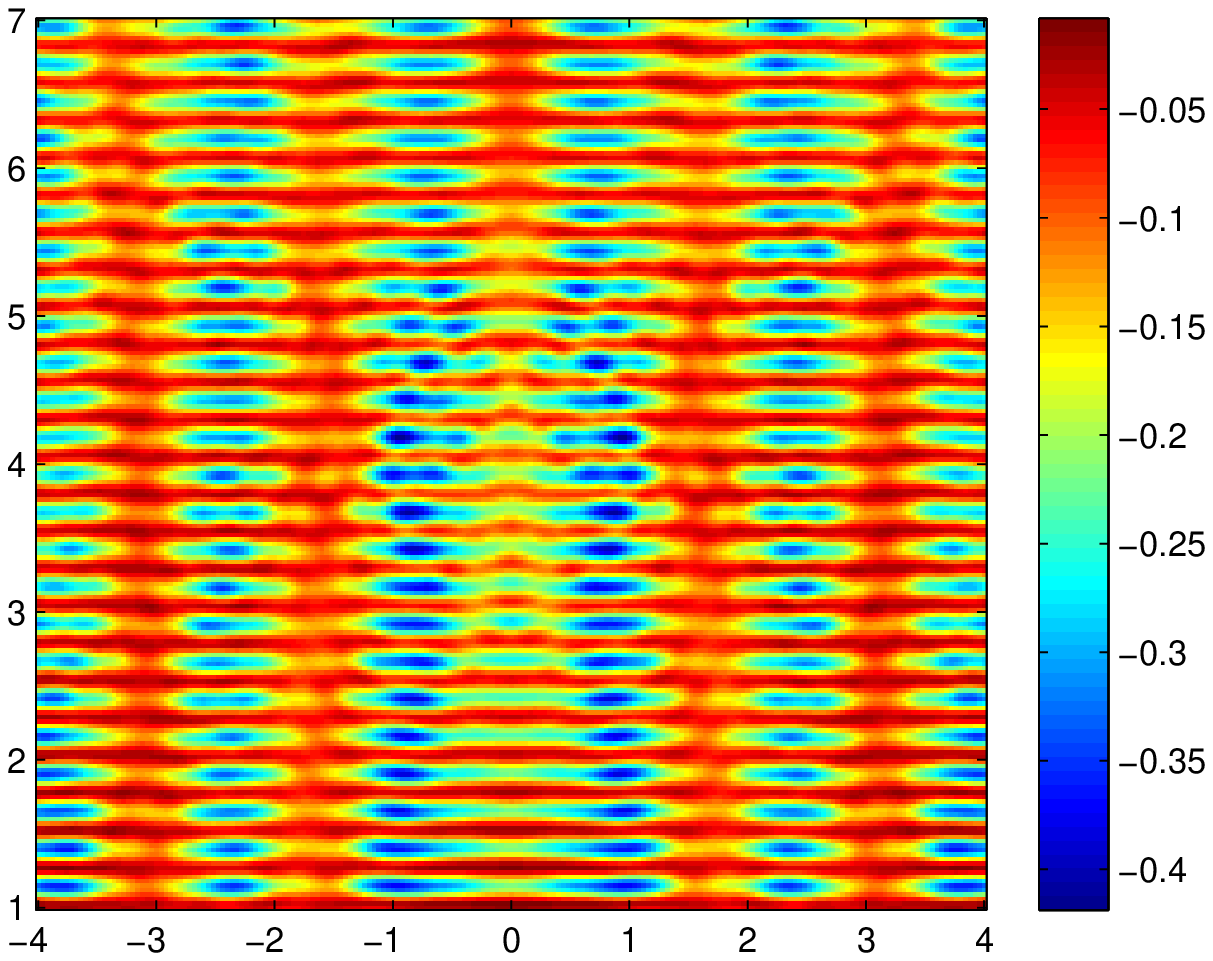}
    \caption{The imaging results with the aperture $d=10,20,30,50$ from left to right, respectively.
     The top row shows the imaging results of our RTM algorithm \eqref{cor2}.
     The middle and the bottom rows show the real and imaginary part of the Kirchhoff migration function \eqref{cor3}. } \label{figure_nn}
\end{figure}

The imaging results are shown in Figure \ref{figure_nn}. We observe that our RTM imaging function peaks at the boundary of the obstacle, while the imaging function $\tilde I_d(z)$ in \eqref{cor3} does not have this property. We remark that in \cite{tmp13}
the Kirchhoff migration type imaging algorithm is successfully used in a setting different from ours:
the sources and receivers in \cite{tmp13} span the full lateral direction of the waveguide which 
is perpendicular to the waveguide boundaries.

\bigskip
\textbf{Example 2}.
In this example we first consider the imaging of a circle of radius $\rho=1$, a kite, and  a rounded square with the impedance boundary condition with $\eta=1$ or $\eta=1000$ on $\Ga_D$. Let $\Om=(-4,4)\times(1,7)$ be the search region. The imaging function is computed at the nodal points
of a uniform $201\times 201$ mesh with the probed wavelength $\lam=0.5$. The imaging results on the top and bottom row shown in Figure \ref{figure_1}
correspond to the surface impedance $\eta=1$ and $\eta=1000$, respectively. We observe our imaging algorithm is quite robust with respect to the magnitude of the surface impedance $\eta$. 

\begin{figure}
    \centering
    \includegraphics[width=0.32\textwidth, height=1.4in]{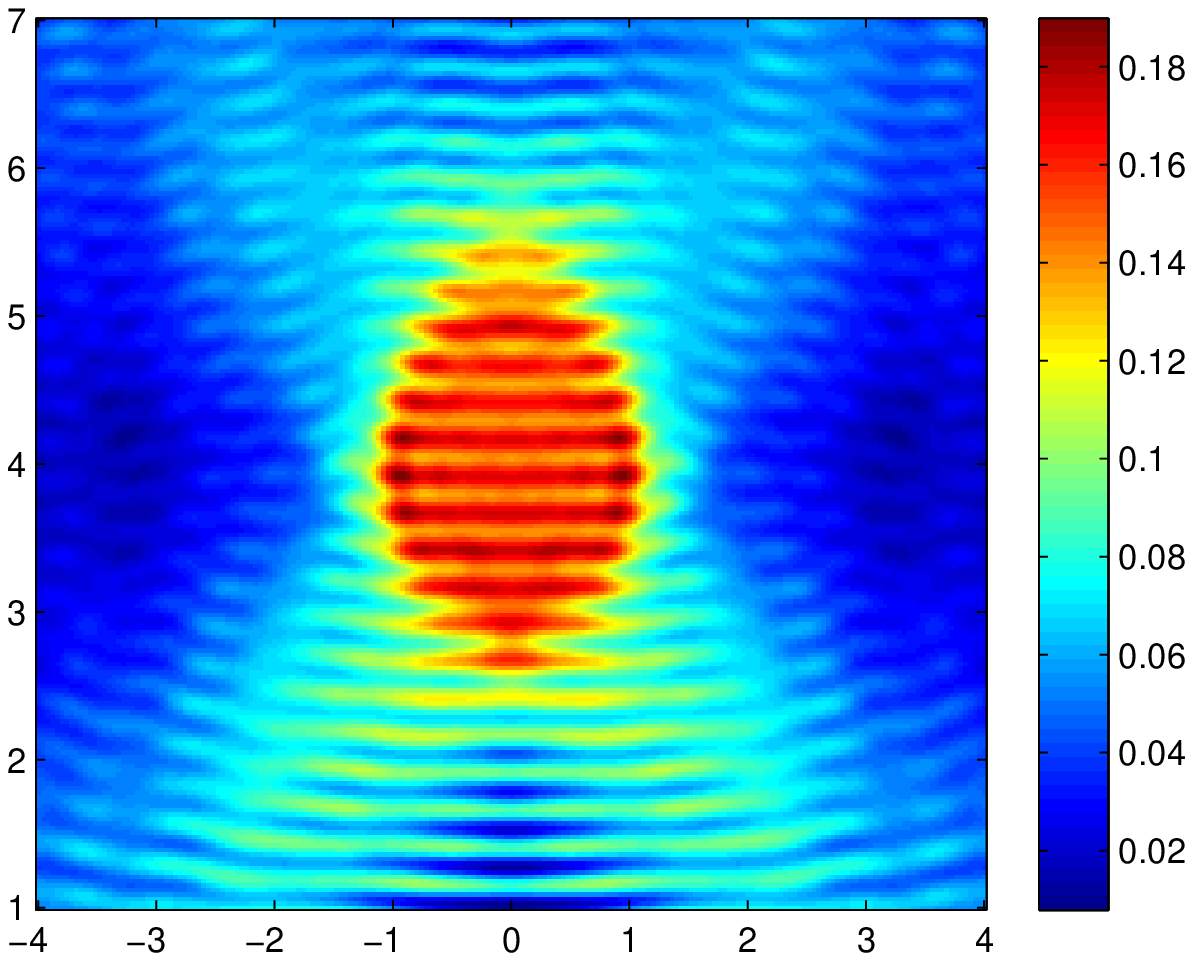}
    \includegraphics[width=0.32\textwidth, height=1.4in]{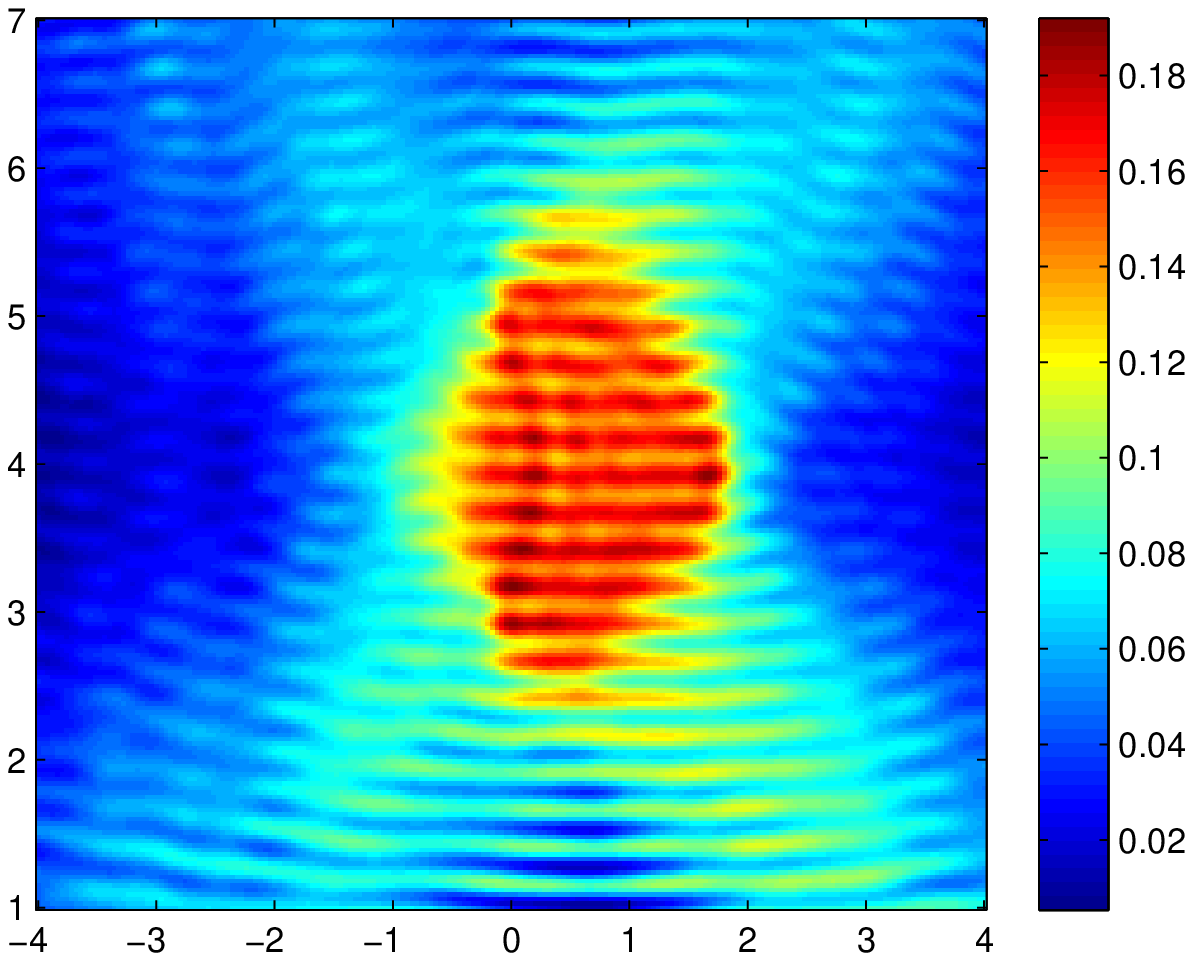}
    \includegraphics[width=0.32\textwidth, height=1.4in]{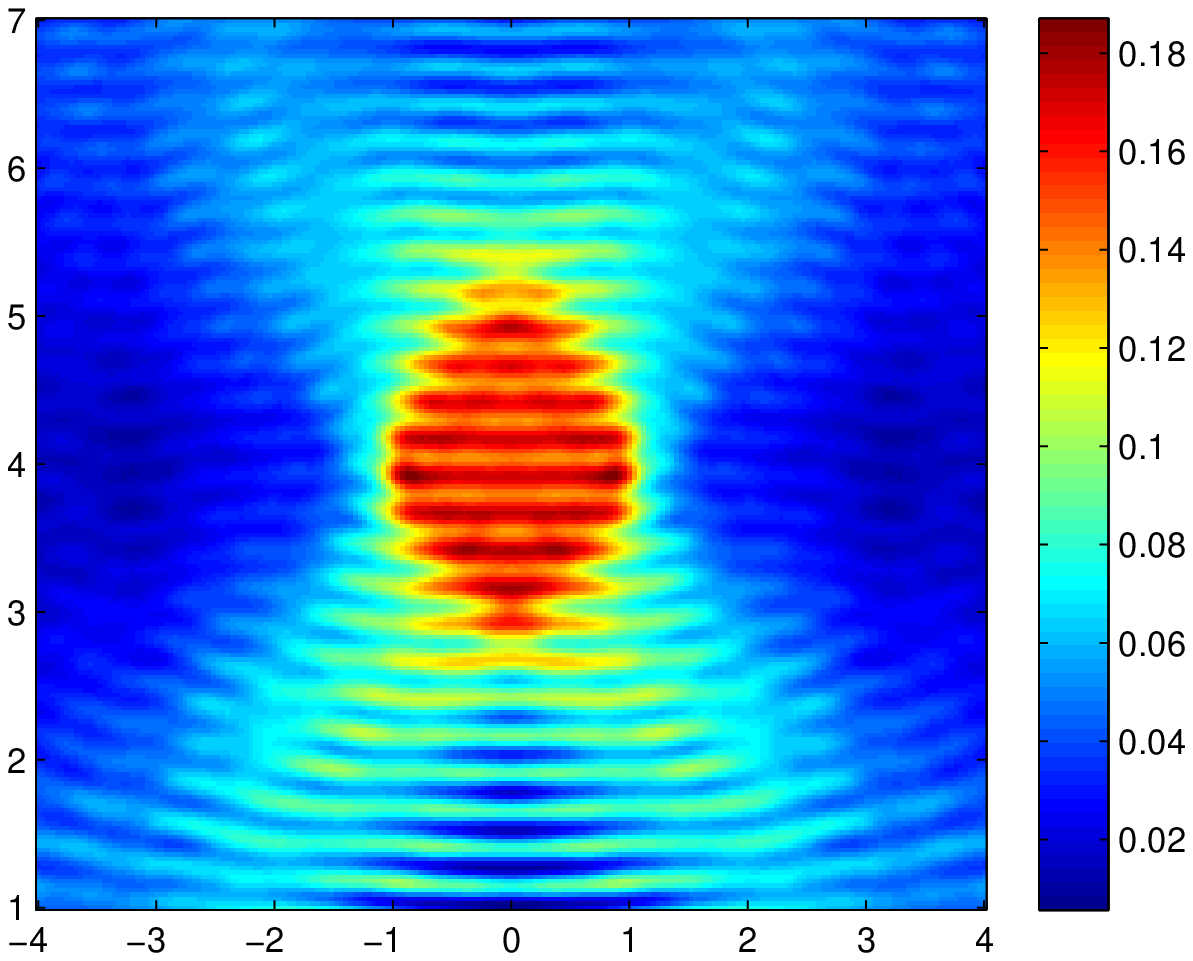}
    \includegraphics[width=0.32\textwidth, height=1.4in]{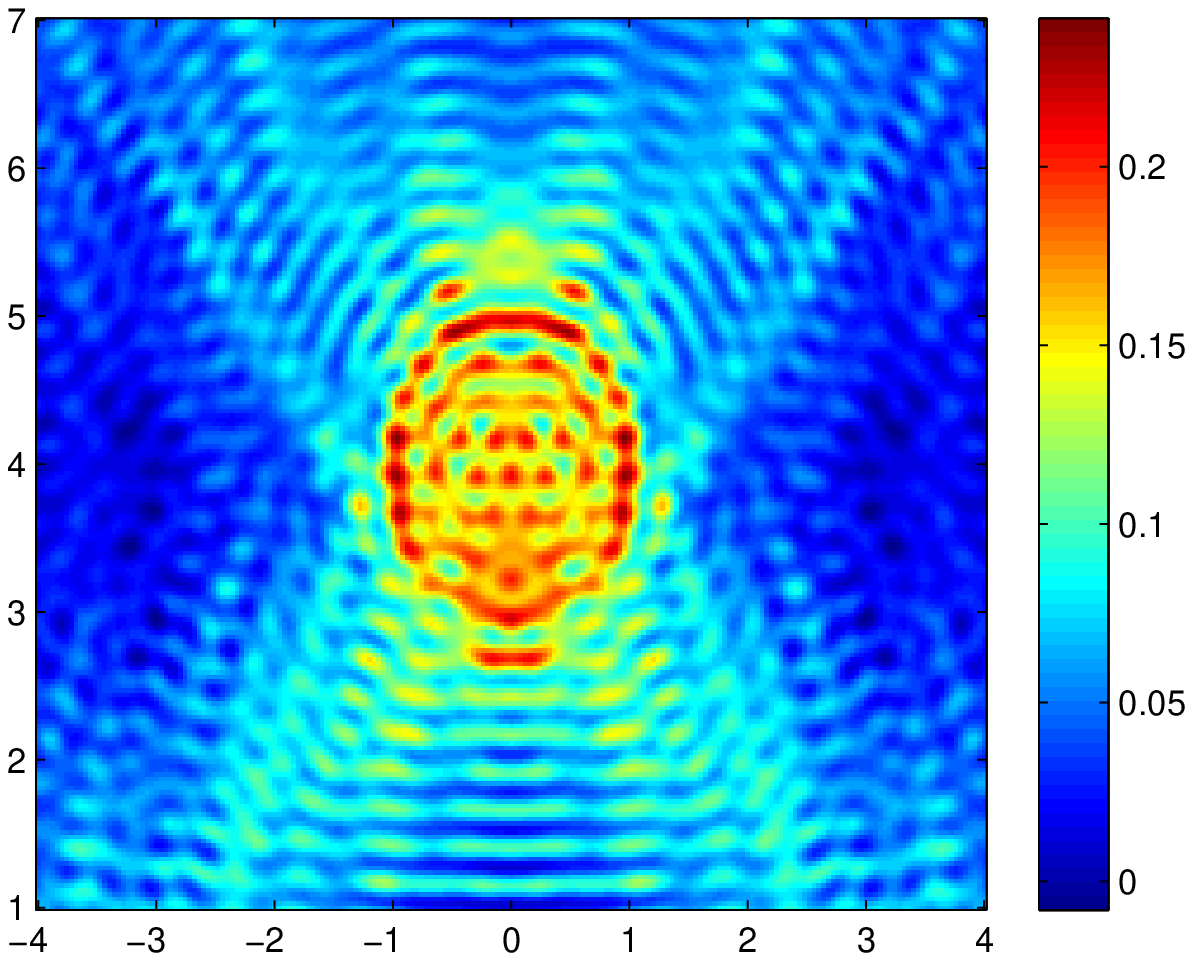}
    \includegraphics[width=0.32\textwidth, height=1.4in]{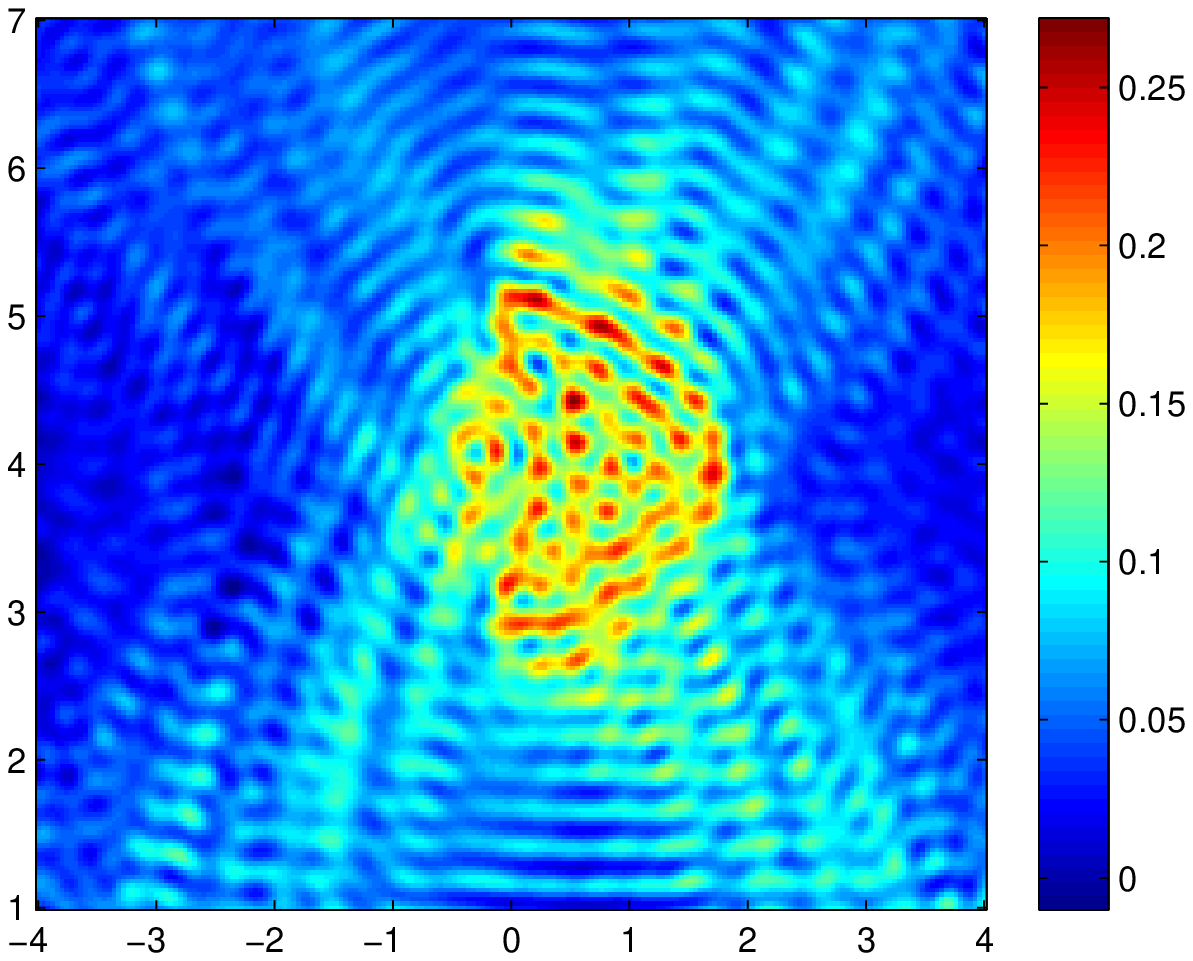}
    \includegraphics[width=0.32\textwidth, height=1.4in]{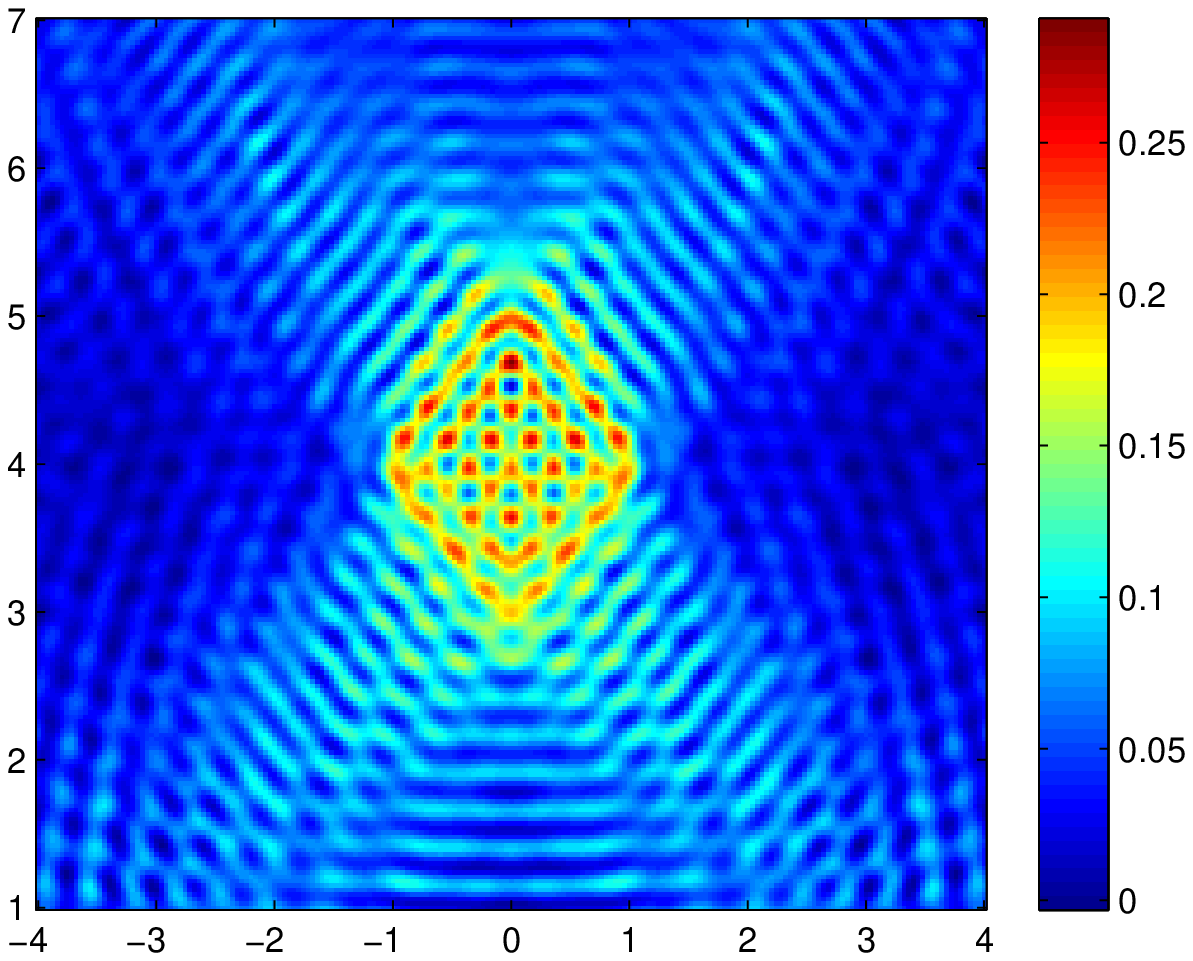}
    \caption{The imaging results with the probe wavelength $\lambda=0.5$, the thickness $h=10$, the aperture $d=30$, and $N_s=N_r=401$ for a circle, a kite, and a rounded square, respectively, The top and  bottom row show the imaging results of the RTM method for the surface impedance $\eta=1$ and $\eta=1000$, respectively. } \label{figure_1}
\end{figure}


\begin{figure}
 \centering
    \includegraphics[width=0.24\textwidth, height=1.2in]{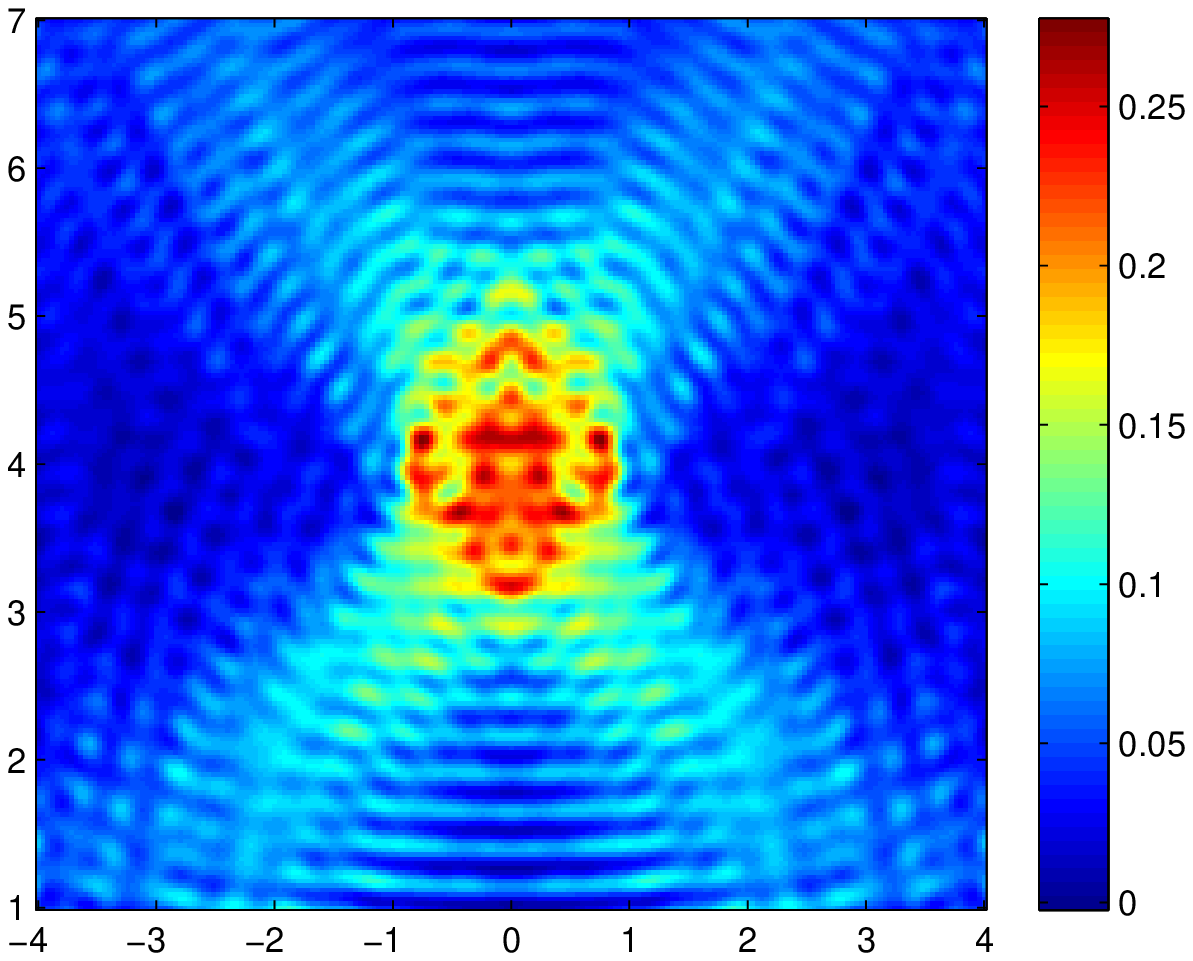}
    \includegraphics[width=0.24\textwidth, height=1.2in]{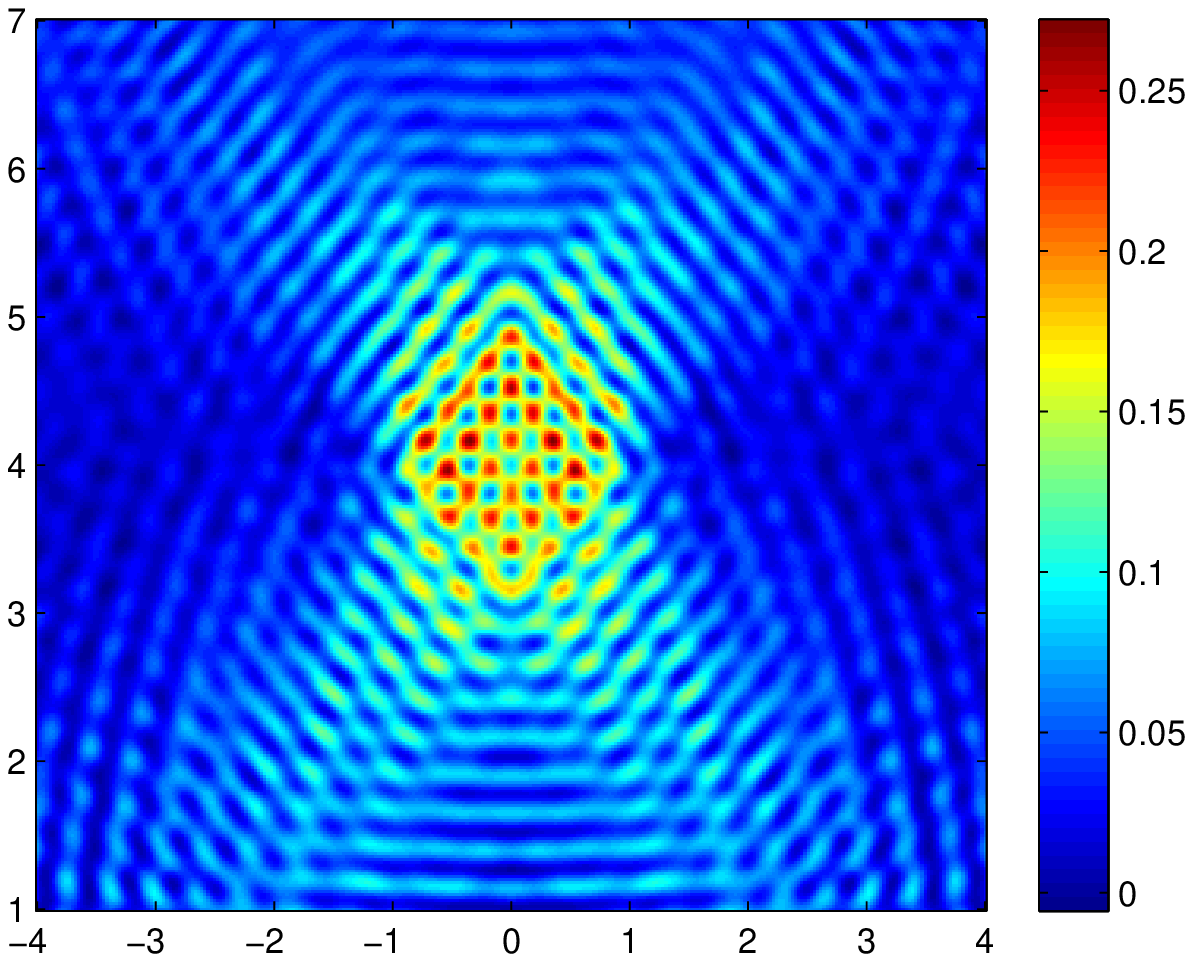}
    \includegraphics[width=0.24\textwidth, height=1.2in]{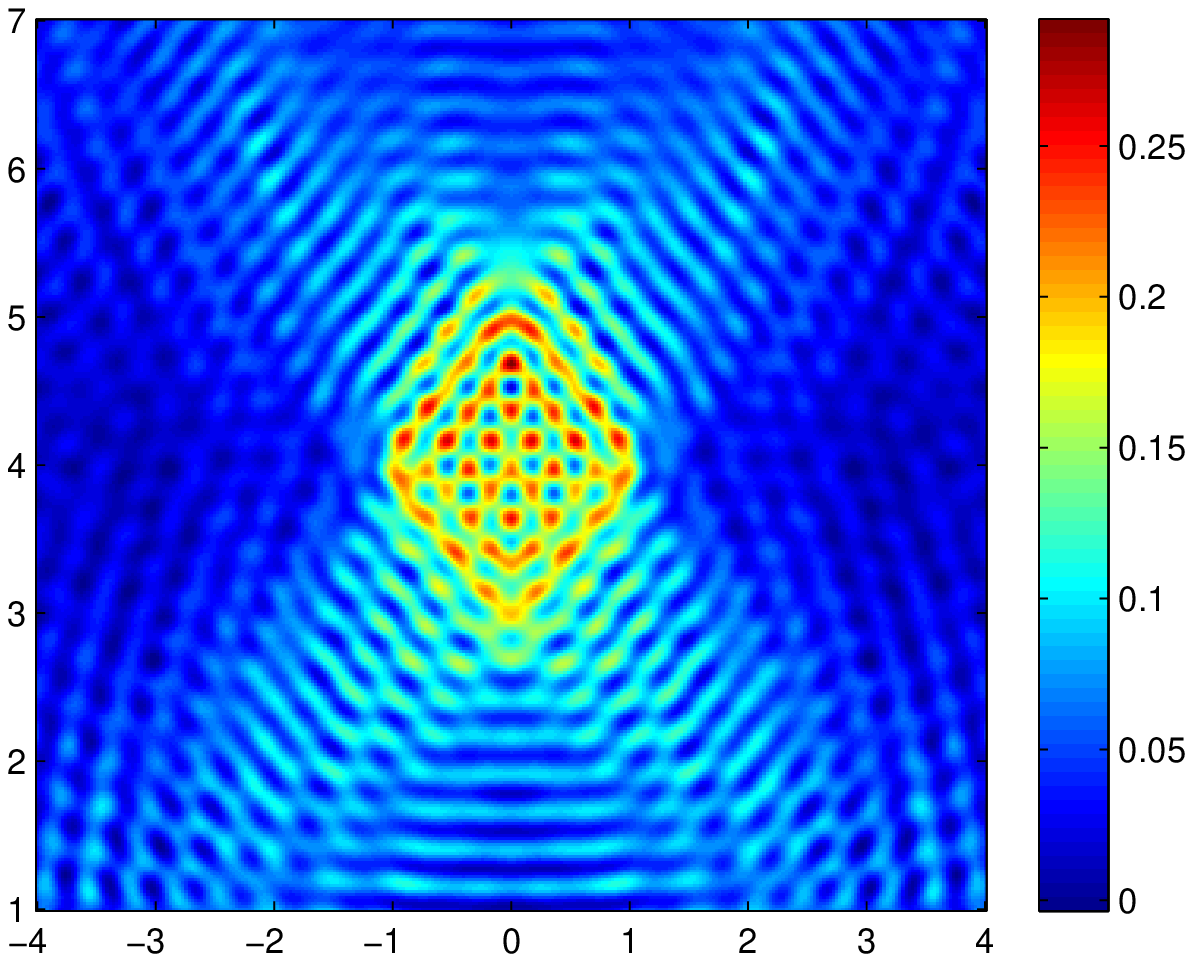}
    \includegraphics[width=0.24\textwidth, height=1.2in]{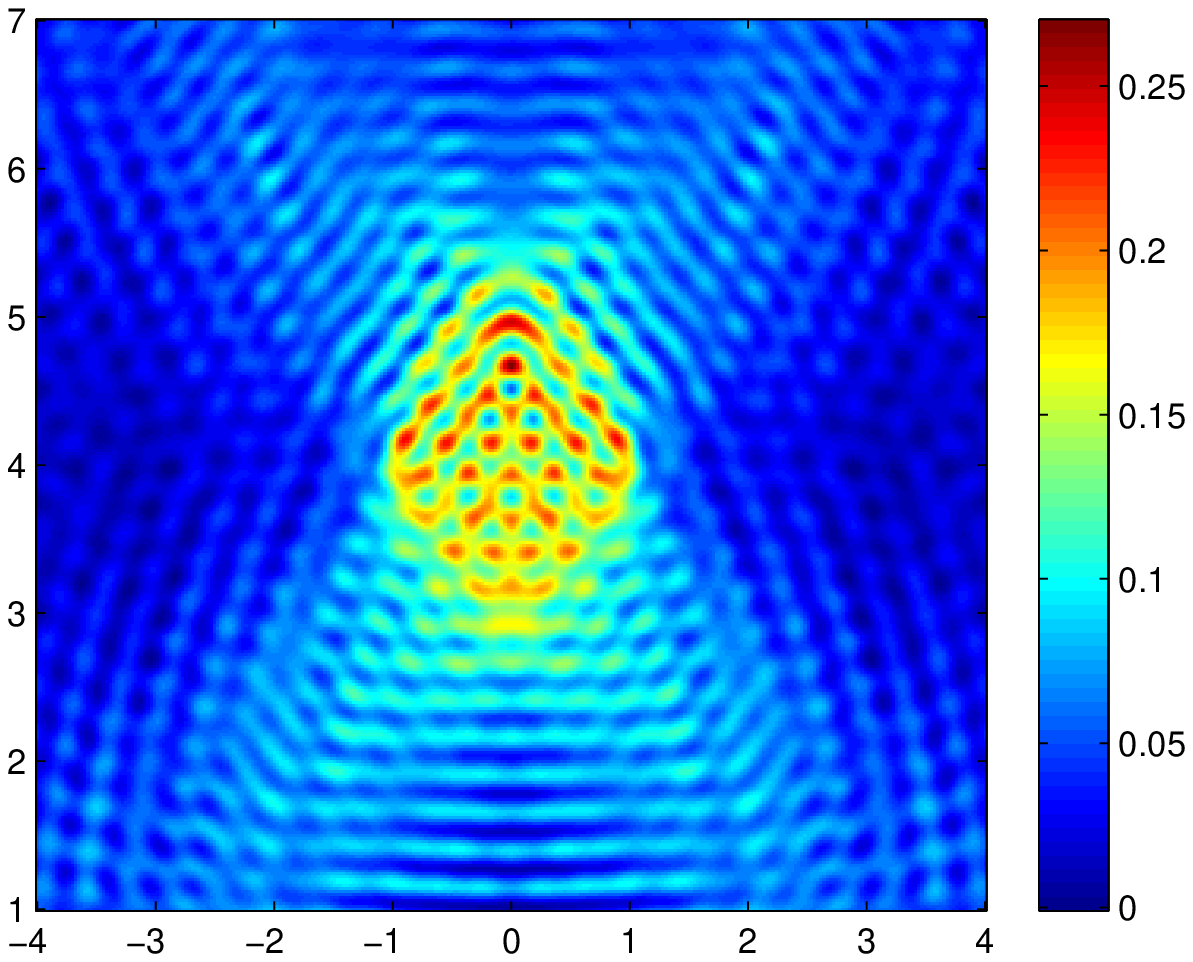}
    \caption{The imaging results, from left to right, for the penetrable obstacle with $n(x)=0.25$, a non-penetrable obstacle with homogeneous Neumann condition, homogeneous Dirichlet condition, and partially coated impedance boundary condition ($\eta=1000$ on the upper half boundary and $\eta=1$ on the lower half boundary), respectively. The probe wavelength $\lambda=0.5$, the thickness $h=10$, the aperture $d=30$, and $N_s=N_r=401$.} \label{figure_bc}
\end{figure}

We then consider to find a penetrable obstacle with the refraction index $n(x)=0.25$, a non-penetrable obstacle with homogeneous Neumann, homogeneous Dirichlet, and partially coated impedance boundary condition ($\eta=1000$ on the upper half boundary and $\eta=1$ on the lower half boundary), respectively. The results are shown in Figure \ref{figure_bc} which indicates clearly that our RTM method can reconstruct the boundary of the obstacle without a priori information on penetrable or non-penetrable obstacles, and in the case of non-penetrable obstacles, the type of the boundary conditions on the boundary of the obstacle.

\begin{figure}
 \centering
    \includegraphics[width=0.24\textwidth, height=1.1in]{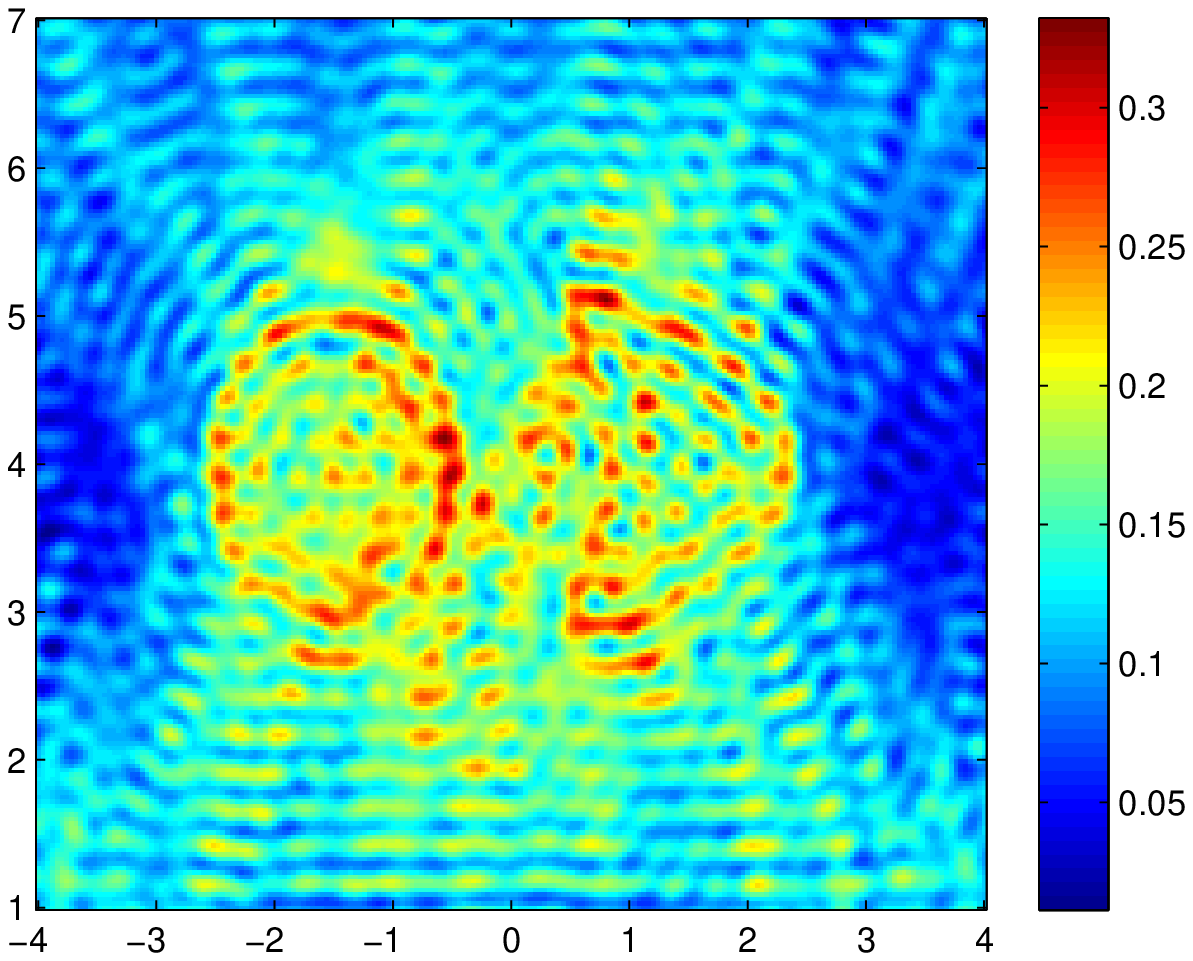}
    \includegraphics[width=0.24\textwidth, height=1.1in]{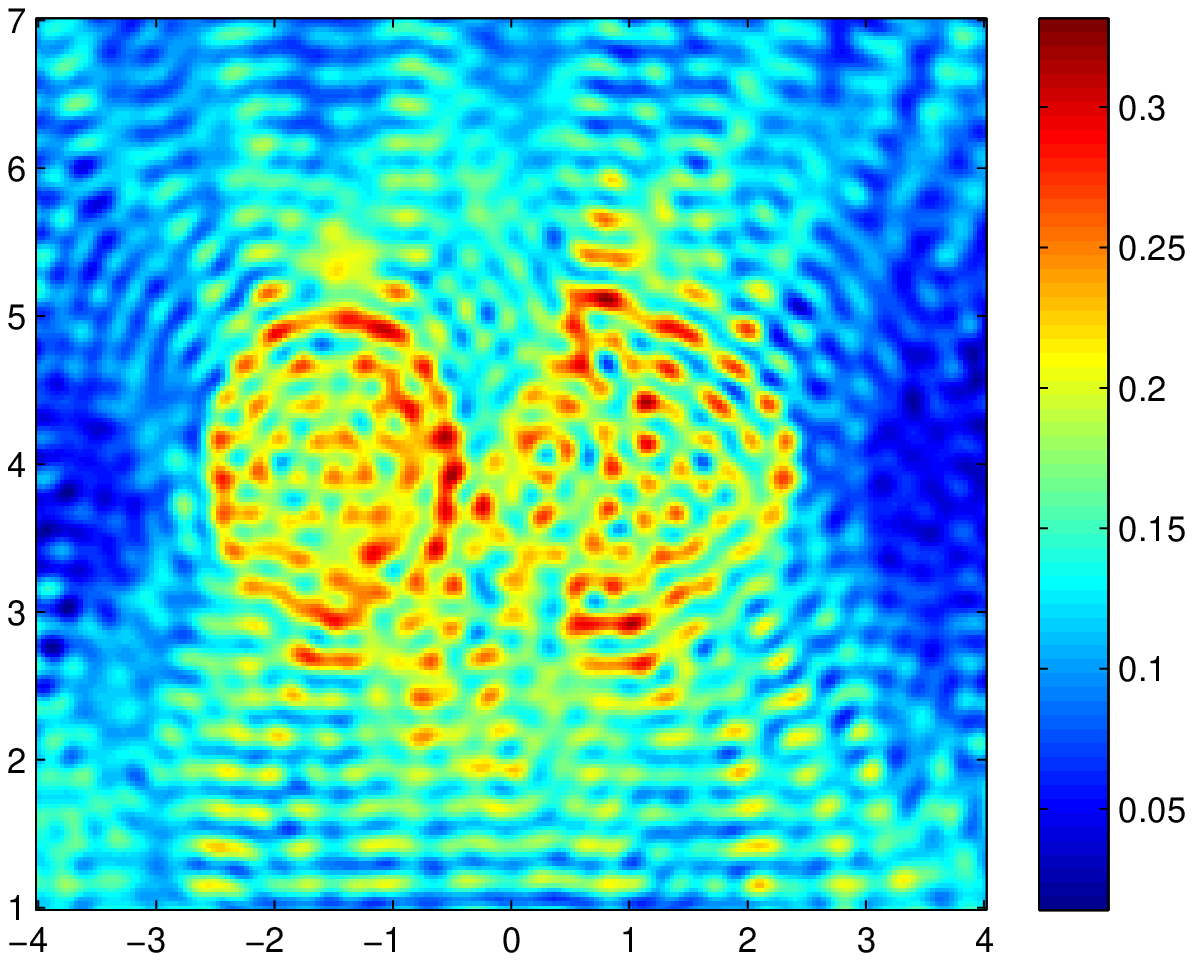}
    \includegraphics[width=0.24\textwidth, height=1.1in]{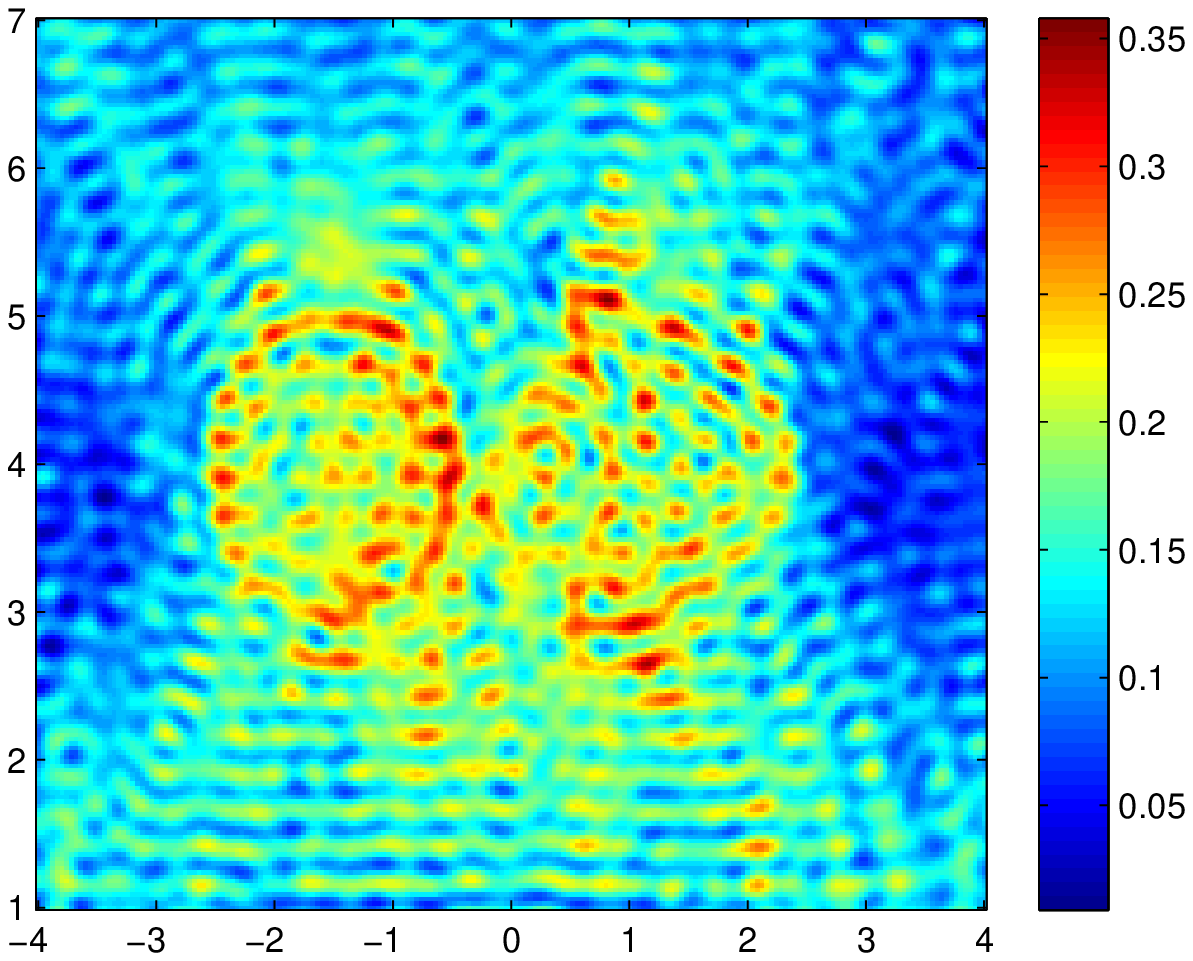}
    \includegraphics[width=0.24\textwidth, height=1.1in]{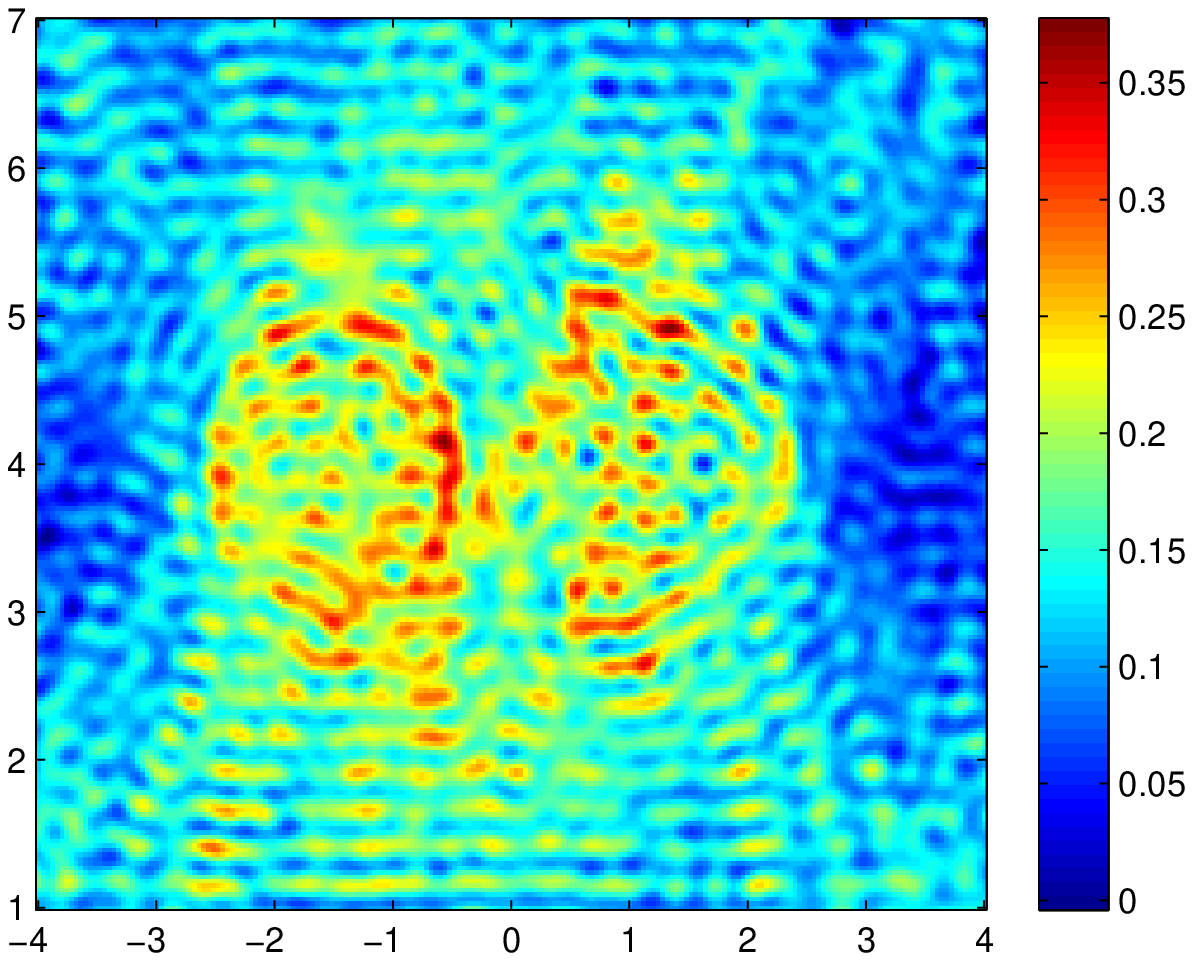}
    \caption{The imaging results using data added with additive Gaussian noise and $\mu = 10\%, 20\%, 30\%, 40\%$ from left to right,  respectively. The probe wavelength $\lambda=0.5$, the thickness $h=10$, the aperture $d=30$, and $N_s=N_r=401$.} \label{figure_3}
\end{figure}

\bigskip
\textbf{Example 3}. In this example we consider the stability of the imaging function with respect to the complex additive Gaussian random noise. We introduce the additive Gaussian noise as follows (see e.g. \cite{cch_a}):
    \begin{equation*}
        u_{noise} = u_s + \nu_{\rm noise},
    \end{equation*}
where $u_s$ is the synthesized data and $\nu_{\rm noise}$ is the complex Gaussian noise with mean zero and standard deviation $\mu$ times the maximum of  the data $|u_s|$, i.e. $\nu_{\rm noise}=\frac{\mu \max|u_s|}{\sqrt{2}}(\eps_1 + \i \eps_2)$, and $\eps_j ~ \thicksim \mathcal{N}(0,1)$ for the real $(j=1)$ and imaginary part $(j=2)$. \par
For the fixed probe wavelength $\lambda=0.5$, we choose one kite and one circle in our test.  The search domain is $\Omega=(-4, 4)\times(1,7)$ with a sampling $201\times 201$ mesh.
Figure \ref{figure_3} shows the imaging results with the noise level $\mu = 10\%, 20\%,30\%, 40\%$ in the single frequency scattered data, respectively.
 The left table in Table \ref{table1} shows the noise level in this case, where $\sigma=\mu \max_{x_r,x_s}|u^s(x_s,x_r)|$, $\|u_s\|_{\ell^2}^2=\frac{1}{N_sN_r}\sum^{N_s,N_r}_{s,r=1}|u^s(x_s.x_r)|^2$, $\|\nu_{\rm noise}\|_{\ell^2}^2 = \frac{1}{N_sN_r}\sum^{N_s,N_r}_{s,r=1}|\nu_{\rm noise}(x_s.x_r)|^2$.

The imaging quality can be improved by using multi-frequency data as illustrated in Figure \ref{figure_4}, in which we show the imaging results added with  the noise level $\mu =10\%, 20\%, 30\%, 40\%$ Gaussian noise by summing the imaging functions for the probed wavelengths $\lam=1/1.8, 1/1.9, 1/2.0, 1/2.1, 1/2.2$. The right table in Table \ref{table1} shows the noise level in the case of multi-frequency data, where $\sigma$, $\|u_s\|_{\ell^2}$, and $\|\nu_{\rm noise}\|_{\ell^2}$ are the arithmetic mean of the corresponding values for different frequencies, respectively. The imaging result is visually much more better than the single frequency imaging result, and the noise
 is greatly suppressed after the summation over individual frequency imaging results.

\begin{figure}[h]
 \centering
    \includegraphics[width=0.24\textwidth, height=1.1in]{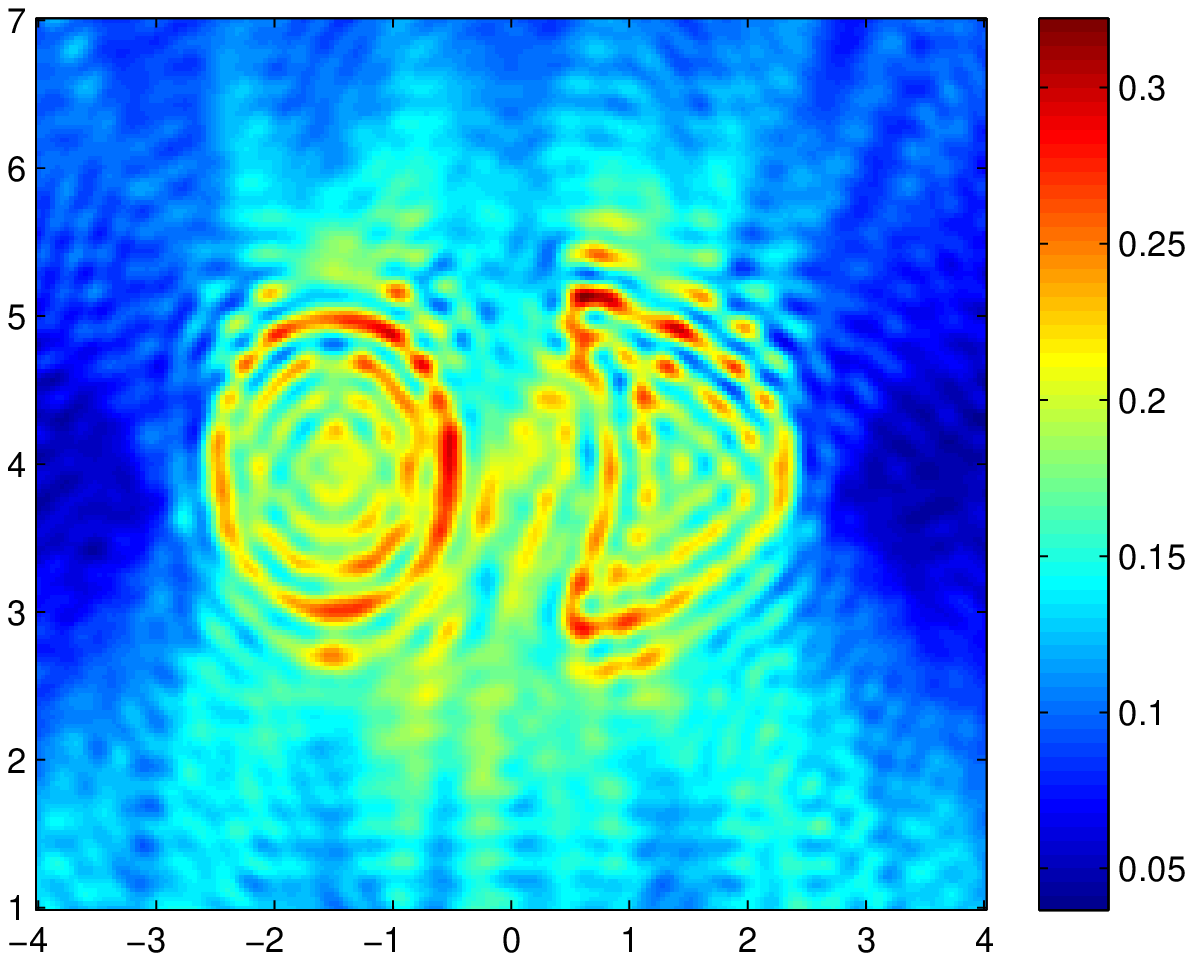}
    \includegraphics[width=0.24\textwidth, height=1.1in]{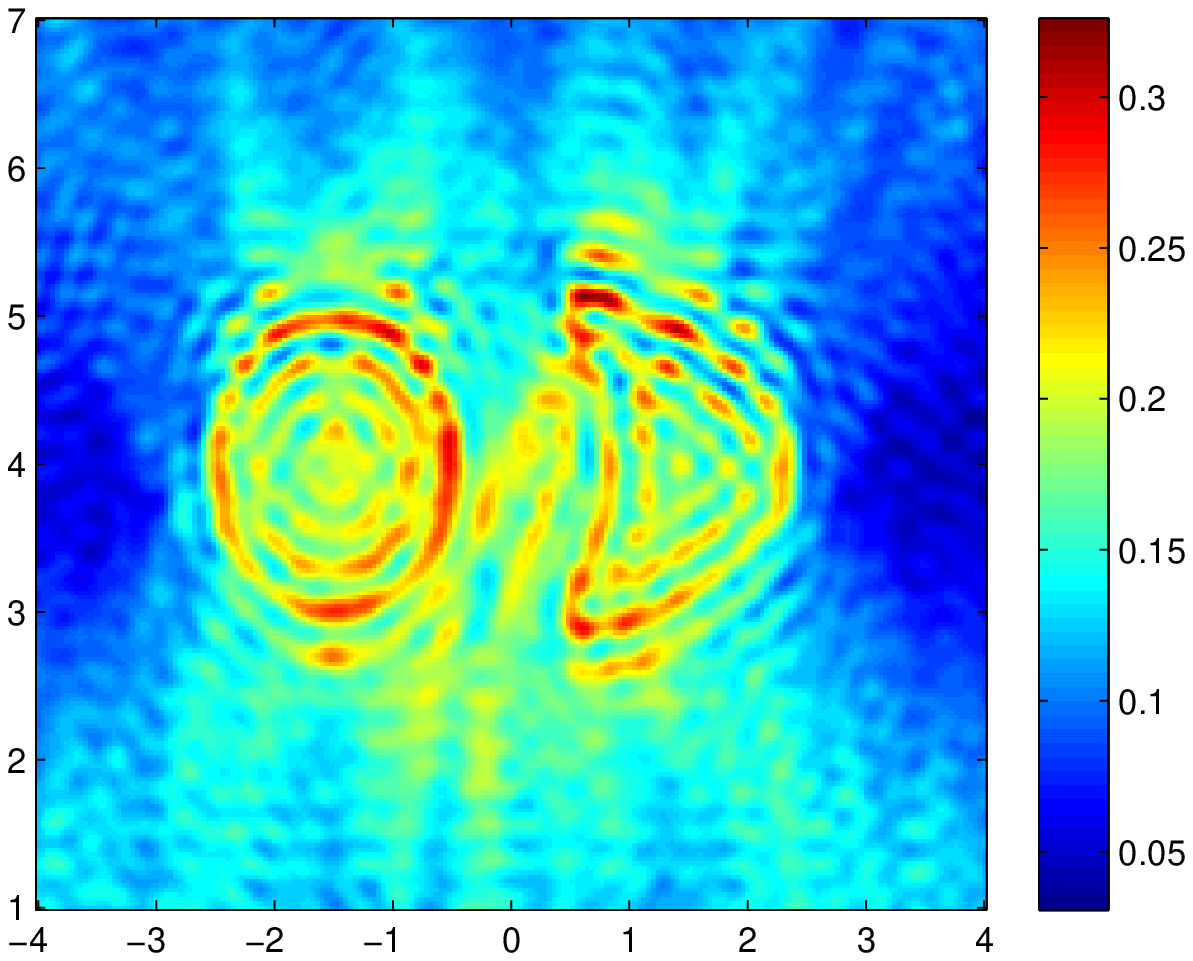}
    \includegraphics[width=0.24\textwidth, height=1.1in]{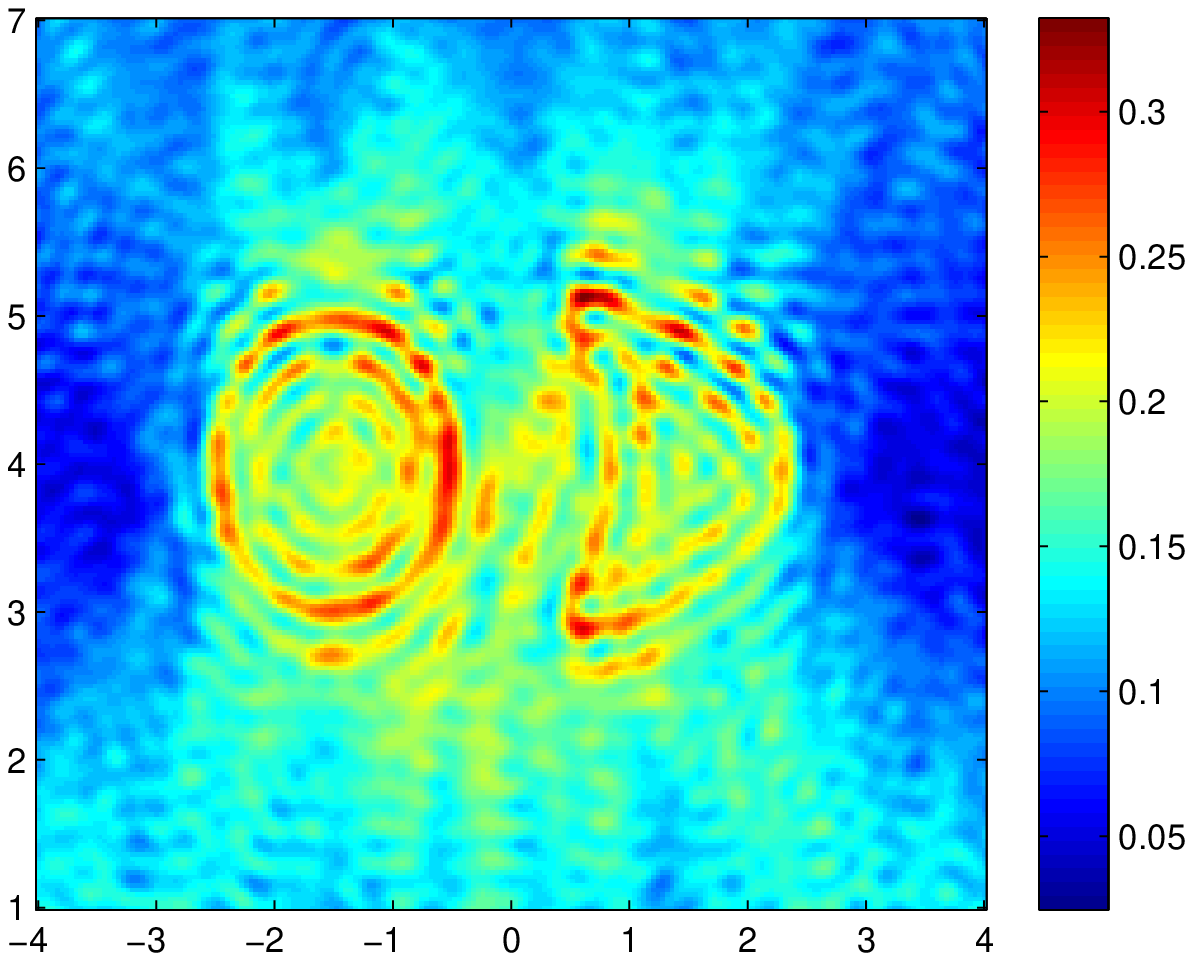}
    \includegraphics[width=0.24\textwidth, height=1.1in]{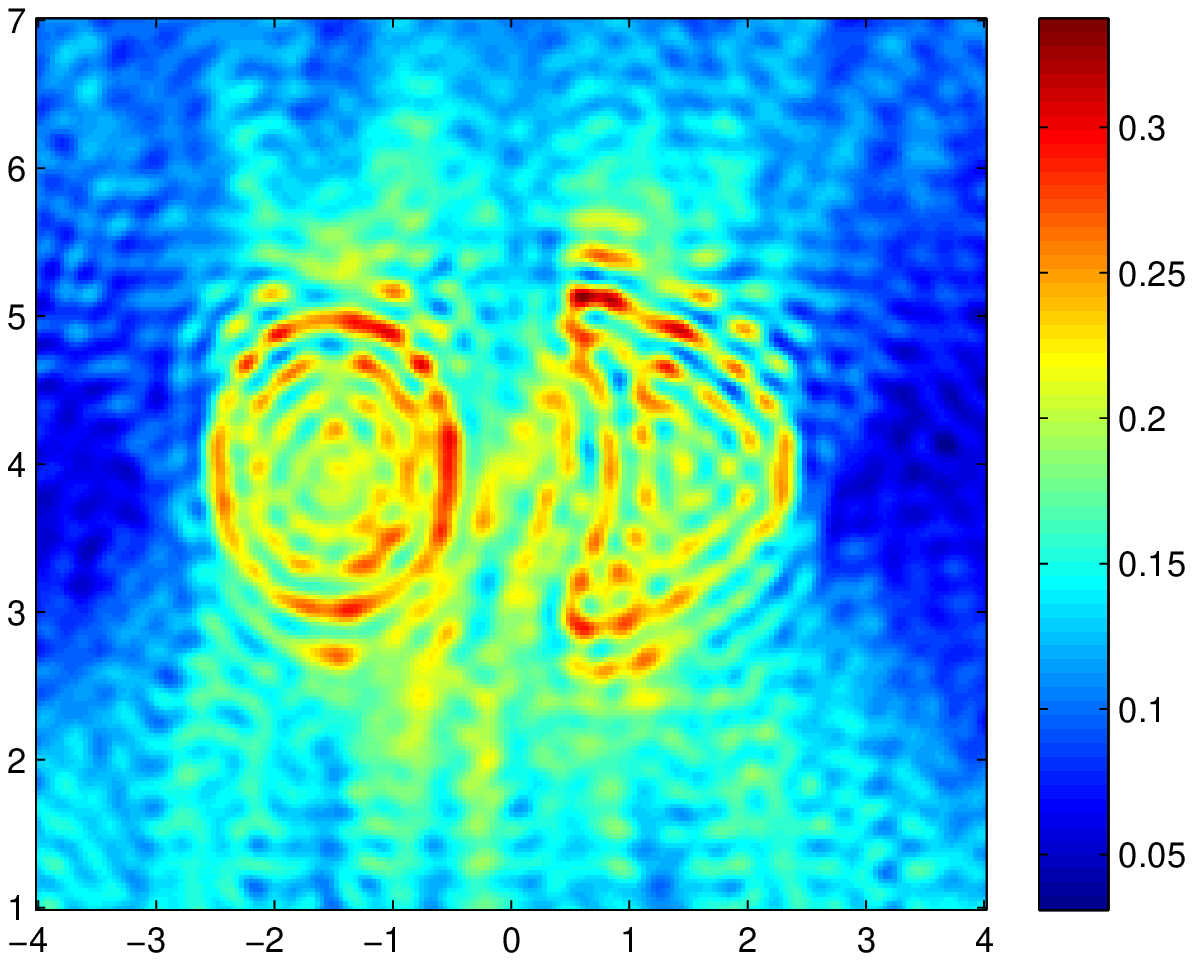}
    \caption{The imaging results using multi-frequency data added with additive Gaussian noise and $\mu = 10\%, 20\%, 30\%, 40\%$ from left to right,  respectively. The probe wavelengths $\lam=1/1.8, 1/1.9, 1/2.0, 1/2.1, 1/2.2$, the thickness $h=10$, the aperture $d=30$, and $N_s=N_r=401$.} \label{figure_4}
\end{figure}

\begin{table}
\caption{The signal level and noise level in the case of single frequency data (left) and multi-frequency data (right).}\label{table1}

\begin{center}
\begin{tabular}{ | c | c| c | c |  }
\hline
$\mu$ & $\sigma$ & $\|u_s\|_{\ell^2}$      & $\|\nu_{\rm noise}\|_{\ell^2}$ \\ \hline
0.1    &    0.0360    & 0.1033    &   0.0293  \\ \hline
0.2    &    0.0720    & 0.1033    &   0.0589  \\ \hline
0.3    &    0.1079    & 0.1033    &   0.0876  \\ \hline
0.4    &    0.1439    & 0.1033    &   0.1178  \\ \hline
\end{tabular} \ \ \ \
\begin{tabular}{ | c | c| c | c |  }
\hline
$\mu$ & $\sigma$ & $\|u_s\|_{\ell^2}$      & $\|\nu_{\rm noise}\|_{\ell^2}$ \\ \hline
0.1     &   0.0355      & 0.1033    &   0.0290  \\ \hline
0.2     &   0.0710      & 0.1033    &   0.0580 \\ \hline
0.3     &   0.1064      & 0.1033    &   0.0869 \\ \hline
0.4     &   0.1419      & 0.1033    &   0.1159 \\ \hline
\end{tabular}
\end{center}
\end{table}

\section{Concluding remarks}
In this paper we have developed a novel reverse time migration algorithm based on the generalized Helmholtz-Kirchhoff identity for the obstacle shape reconstruction in planar acoustic waveguide. The algorithm consists of using the half space Green function instead of the waveguide Green function in both the back-propagation and cross-correlation processes. The algorithm is quite robust with respect to the random noise. Our numerical experiments indicate that the RTM algorithm based on multiple frequency 
superposition can effectively suppress the random noise. Extending the results in this paper to the electromagnetic and elastic waveguide imaging problem is of considerable practical interests and will be pursued in our future works.

\section{Appendix: Proof of Theorem \ref{LAP}}\label{A2}
We will prove the existence of the radiation solution of the problem \eqref{l1}-\eqref{l3} by the method of limiting absorption principle. The
argument is standard and generalizes that for Helmholtz scattering problem in the free space, see e.g. \cite{leis}. Here we only outline the main steps.

For any $z=1+\i\eps$, $\eps>0$, $f\in L^2(\R^2_h)$ with compact support in
$B_R=(-R,R)\times(0,h)$, where $R>0$, we consider the problem
\be
 & &   \Delta u_z + zk^2 u_z = -f \qquad \mbox{in } \R^2_h, \label{wg_ab}\\
 & &    u_z=0\ \ \mbox{on }\Ga_0,\ \ \ \ \frac{\partial u_z}{\partial x_2} = 0\ \ \mbox{on }\Ga_h.\label{wg_abc}
\ee
By Lax-Milgram lemma we know that (\ref{wg_ab})-\eqref{wg_abc} has a unique solution $u_z \in H^{1}(\R^2_h) $. For any domain $\mathcal D\subset\R^2_h$,
we define the weighted space $L^{2,s}(\mathcal D),s \in \R$, by
\ben
    L^{2,s}(\mathcal D)=\{v \in L^2_{\rm loc}(\mathcal D): (1+|x_1|^2)^{s/2}v \in L^2(\mathcal D) \}
\een
with the norm $\| v \|_{ L^{2,s}(\mathcal D)} = (\int_{\mathcal D}(1+|x_1|^2)^{s}|v|^2 dx )^{1/2}$. The weighted Sobolev space $H^{1,s}(\mathcal D),s \in \R$,
is defined as the set of functions in $L^{2,s}(\mathcal D)$ whose first derivative is also in $L^{2,s}(\mathcal D)$. The norm
$\| v \|_{ H^{1,s}(\mathcal D)} = (\| v \|^2_{ L^{2,s} (\mathcal D)} + \| \nabla v \|^2_{ L^{2,s}(\mathcal D)})^{1/2}$.

\begin{lemma}{\label{newton}}
Let  $ f  \in L^2 (\R^2_h) $ with compact support in $B_R$. For any $z=1+\i\eps$, $0<\eps<1$, we have, for any $s>1/2$,
$\|u_z\|_{H^{1,-s}(\R^2_h)}\le C\|f\|_{L^2(\R^2_h)}$ for some constant independent of $\eps, u_z$, and $f$.
\end{lemma}

\begin{proof} We first note that by testing \eqref{wg_ab} by $(1+|x_1|^2)^{-s}\bar u_z$, $s>1/2$, one can obtain $\|u_z\|_{H^{1,-s}(\R^2_h)}\le C\|u_z\|_{L^{2,s}(\R^2_h)}+
C\|f\|_{L^2(\R^2_h)}$ by standard argument. It remains to show $\|u_z\|_{L^{2,s}(\R^2_h)}\le C\|f\|_{L^2(\R^2_h)}$. It is obvious that we only need to prove the estimate for
$f\in C^\infty_0(\R^2_h)$. We start with the following integral representation formula
\bee\label{b1}
    u_z(x) = \int_{\R^2_h}N^z(x,y)f(y)dy, \ \ \ \ x \in \R^2_h.
\eee
Here $N^z(x,y)$ is the Green function of the problem (\ref{wg_ab})-\eqref{wg_abc} with the complex wave number $kz^{1/2}$, where $\Im(z^{1/2})>0$ for $\eps>0$.
Similar to (\ref{GreenN}), it is easy to check that
\bee\label{b2}
    N^z(x,y) = \sum_{n=1}^{\infty}\frac{\textbf{i}}{h\xi^z_n}\sin(\mu_n x_2)\sin(\mu_n y_2)e^{\textbf{i}\xi^z_n|x_1-y_1|},
\eee
where $\xi^z_n = \sqrt{zk^2 - \mu_n^2}$ whose imaginary part $\Im{\xi^z_n} \ge 0$.
It follows from \eqref{b1}-\eqref{b2} that $u_z$ has the mode expansion
\bee\label{b3}
    u_z(x) =  \sum_{n=1}^{\infty} u_n^z(x_1)\sin(\mu_n x_2 ),
\eee
where, since $f$ is supported in $B_R$,
\ben
u_n^z(x_1)=\frac h2\int_{-R}^{R}\frac{\textbf{i}}{h\xi^z_n}e^{\textbf{i}\xi^z_n|x_1-y_1|}f_n(y_1)dy_1,\ \ \ \ f_n(x_1)=\frac 2h\int_0^hf(x)\sin(\mu_nx_2)dx_2.
\een
Since $\Im\xi^z_n\ge 0$ and $|\xi_n^z|^2=\sqrt{(\mu_n^2-k^2)^2+(k^2\eps)^2} \ge |\mu_n^2-k^2|$, we have
\ben
|u_n^z(x_1)|\le\frac{h}2\frac{\sqrt{2R}}{h|\mu_n^2-k^2|^{1/2}}\,\|f_n\|_{L^2(\R)}.
\een
Therefore
\ben
\|u\|_{L^{2,-s}(\R^2_h)}^2 &=& \frac h2\sum^\infty_{n=1}\int^\infty_{-\infty}(1+|x_1|^2)^{-s}|u^z_n(x_1)|^2dx_1\\
&\le&\left(\frac h2\right)^2 \|f\|^2_{L^2(\R^2_h)} \sum^\infty_{n=1}\int_{-\infty}^{+\infty}
\frac{2R}{h^2|\mu^2_n-k^2|}(1+|x_1|^2)^{-s}dx_1\\
&\le&C\|f\|_{L^2(\R^2_h)}^2.
\een
where we have used $\|f\|^2_{L^2(\R^2_h)}=\frac h2\sum^\infty_{n=1}\|f_n\|^2_{L^2(\R)}$. This completes the proof. \end{proof}

Now we are ready to prove Theorem \ref{LAP}.

{\it Proof of Theorem \ref{LAP}}.
For any $0<\eps<1$, we consider the problem
\be  {\label{wg2}}
 & &   \Delta u_{\eps} + (1+\i \eps)k^2 u_{\eps} = 0 \qquad \mbox{in } \R^2_h\bks\bar D, \\
 & &    u_{\eps}=0\ \ \mbox{on }\Ga_0,\ \ \ \ \frac{\partial u_{\eps}}{\partial x_2} = 0\ \ \mbox{on }\Ga_h, {\label{wg3}}\\
 & &    \frac{\pa u_{\eps}}{\pa \nu} + \i k \eta u_{\eps} = g \ \ \ \mbox{ on } \ \ \Ga_D. {\label{wg4}}
\ee
We know that the above problem has a unique solution $u_{\eps} \in H^1(\R^2_h \backslash \bar D)$ by the Lax-Milgram Lemma.

Let $\chi \in C_{0}^{\infty}(\R^2_h)$ be the cut-off function such that $0 \leq \chi \leq 1$, $\chi=0$ in $B_R$, and $\chi=1$
outside of $B_{R+1}$. Let $v_{\eps}=\chi u_{\eps}$, then $v_{\eps}$ satisfies the equation (\ref{wg_ab}) with
$z=1+\i \eps$ and $f=u_{\eps}\Delta \chi + 2\nabla u_{\eps}\cdotp \nabla \chi$. Obviously, $f$ is supported in $B_{R+1}$. By Lemma \ref{newton}, we have $\|v_{\eps}\|_{H^{1,-s}(\R^2_h \backslash \bar{D})}\le C\|u_{\eps}\|_{H^1(B_{R+1}\backslash \bar{D})}$. Since $\chi=1$ outside $B_{R+1}$, we have then
\be \label{newton2}
\|u_{\eps}\|_{H^{1,-s}(\R^2_h \backslash \bar{D})}\le C\|u_{\eps}\|_{H^1(B_{R+1}\backslash \bar{D})}.
\ee
Next let $\chi_1 \in C_{0}^{\infty}(\R^2_h)$
be the cut-off function with that  $0 \leq \chi_1 \leq 1$, $\chi_1=1$ in $B_{R+1}$, and $\chi_1=0$
outside of $B_{R+2}$. For $g\in H^{-1/2}(\Ga_D)$, let $u_g \in H^{1}(\R^2_h \backslash \bar{D})$ be the lifting function such that $ \frac{\pa u_{g}}{\pa \nu} + \i k \eta u_{g}=g \mbox{ on } \Ga_D$ and $\|u_{g}\|_{H^{1}(\R^2_h \backslash \bar{D})}\le C\|g\|_{H^{-1/2}(\Ga_D)}$ hold. By testing \eqref{wg2} with
$\chi_1^2 ( \overline{u_{\eps}-u_{g}} )$, we have by the standard argument
\be \label{u1}
\|u_{\eps}\|_{H^{1}(B_{R+1} \backslash \bar{D})}\le C( \|u_{\eps}\|_{L^{2}(B_{R+2} \backslash \bar{D})} + \|g\|_{H^{-1/2}(\Ga_D)} ).
\ee
Now we claim
\bee \label{u2}
  \|u_{\eps}\|_{L^{2}(B_{R+2} \backslash \bar D)} \leq C \|g\|_{H^{-1/2}(\Ga_D)} ,
\eee
for any $g \in H^{-1/2}(\Ga_D)$ and $\eps >0$. If it were false, there would exist sequences $\{g_m\} \subset H^{-1/2}(\Ga_D)$ and $\{\eps_m\} \subset (0,1)$, and $\{u_{\eps_m}\}$ be the corresponding solution of (\ref{wg2})-(\ref{wg4}) such that
\be {\label{contradict}}
    \|u_{\eps_m}\|_{L^{2}(B_{R+2} \backslash \bar D)} = 1 \ {\rm{ and }} \ \|g_m\|_{H^{-1/2}(\Ga_D)} \leq \frac{1}{m}.
\ee
Then $\|u_{\eps_m}\|_{H^{1,-s}(\R^2_h \backslash \bar D)}\le C $, and thus there is a subsequence of $\{\eps_m\}$, which is
still denoted by $\{\eps_m\}$, such that $\eps_m \to \eps' \in [0,1]$, and a subsequence of $\{u_{\eps_m}\}$,
which is still denoted by $\{u_{\eps_m}\}$, such that it converges weakly to some $u_{\eps'} \in H^{1,-s}(\R^2_h \backslash \bar D)$.
The function $u_{\eps'}$ satisfies (\ref{wg2})-(\ref{wg4}) with $g=0$ and $\eps=\eps'$ .
By the integral representation formula, we have, for $x\in\R^2_h\bs \bar D$,
 \be \label{SL_wg}
    u_{\eps'}(x) = -\int_{\Ga_D} \left(\frac{\partial N^{1+\i\eps'}(x,y)}{\partial \nu(y)} u_{\eps'}(y)-N^{1+\i\eps'}(x,y)\frac{\partial u_{\eps'}(y)}{\partial \nu(y)}\right)ds(y).
 \ee
If $\eps'>0$, we deduce  from (\ref{SL_wg}) that $u_{\eps'}$ decays exponentially and thus $u_{\eps'} \in H^{1}(\R^2_h \backslash \bar D) $, then $u_{\eps'} =0$ by the uniqueness of the solution in $H^{1}(\R^2_h \backslash \bar D) $ with positive absorption.
If $\eps'=0$, (\ref{SL_wg}) implies that $u_{\eps'}$ satisfies the mode radiation condition (\ref{mode2}), and then
 $u_{\eps'}=0$ by the uniqueness Lemma \ref{uniqueness}. Therefore, in any case, $u_{\eps'}=0$, which, however, contradicts to (\ref{contradict}).

This shows \eqref{u2}. Consequently, by (\ref{newton2}) and (\ref{u1}),
\be \label{stab}
\|u_{\eps}\|_{H^{1,-s}(\R^2_h \backslash \bar D)}\le C \|g\|_{H^{-1/2}(\Ga_D)}.
\ee
Now, it is easy to see that $u_{\eps}$ has a convergent subsequence which converges weakly to some $ u \in H^{1,-s}(\R^2_h \backslash \bar D) $ and satisfies (\ref{l1})-(\ref{l3}). The desired estimate follows from (\ref{stab}). This completes the proof. \qquad $\Box$

{\bf Acknowledgement.} We would like to thank the referees for their insightful comments which lead to great improvement
of the paper.

\end{document}